%% file: main.tex
\documentclass[prd,superscriptaddress,aps,amsfonts,showpacs,longbibliography,notitlepage,twocolumn,nofootinbib]{revtex4-1}
\pdfoutput=1

\usepackage{CJKutf8}

\include{preamble}
\renewcommand{\revadd}[1]{{\color{red}#1}}

\renewcommand{\Bsectorname}{bit-type}
\renewcommand{\Csectorname}{check-type}
\hypersetup{
    pdftitle={Spatial overhead reduction for 2D hypergraph product codes}
}
\begin{document}
\title{Spatial overhead reduction for 2D hypergraph product codes}

\author{Aarav Pabla}
\email[Current affiliation and contact: UCLA, ]{apabla@ucla.edu}
\affiliation{Joint Center for Quantum Information and Computer Science, NIST \& University of Maryland, College Park, MD 20742, USA/}

\author{\begin{CJK}{UTF8}{gbsn}Yu-Xin Wang (王语馨)\end{CJK}}
\affiliation{Joint Center for Quantum Information and Computer Science, NIST \& University of Maryland, College Park, MD 20742, USA}

\author{Yifan Hong}
\email[Current affiliation and contact: NVIDIA, ]{yihong@nvidia.com}
\affiliation{Joint Center for Quantum Information and Computer Science, NIST \& University of Maryland, College Park, MD 20742, USA}
\affiliation{Joint Quantum Institute, NIST \& University of Maryland, College Park, MD 20742, USA}

\begin{abstract}
The hypergraph product creates a quantum stabilizer code from two input classical linear codes; a paradigmatic example being the surface code as a hypergraph product of two classical repetition codes. Many properties of the hypergraph product code can be inherited from those of the classical codes such as the code dimension, minimum distance and certain fault-tolerant gadgets. We investigate ways to reduce the number of physical qubits in hypergraph product codes while maintaining some of their useful properties for fault tolerance. We show that the code dimension, canonical logical basis, and minimum distances of the hypergraph product code are preserved through this reduction. We also provide syndrome-measurement schedules that preserve the effective distance as well as examples of reduced hypergraph product codes with parameter improvements such as $[\![610,64,6]\!] \rightarrow [\![441,64,6]\!]$ and $[\![1225,49,11]\!] \rightarrow [\![931,49,11]\!]$. In memory simulations with circuit-level depolarizing noise, the reduced codes can exhibit logical error rates comparable to their unreduced versions at the sampled low physical-error rates, while using fewer physical qubits. Finally, we show how overhead reduction can be compatible with homomorphic measurement gadgets, fold-transversal gates and automorphisms, which extends the savings to logical computation.
\end{abstract}

\date{\today}

\maketitle

\tableofcontents

\section{Introduction}

Quantum error correction (QEC) protects quantum information by encoding it in the nonlocal degrees of freedom of quantum error-correcting codes~\cite{Shor_1995_QEC, Cerf_1997}. A class of codes known as stabilizer codes~\cite{gottesman1997thesis} repeatedly measures commuting observables to extract syndromes to infer errors. One of the earliest quantum error-correcting codes is the surface code, which encodes one logical qubit among a square lattice of physical qubits~\cite{Kitaev_1997_QC, Kitaev_1997_QEC, Kitaev_2003, bravyi1998surface}. It is one of the leading candidates for fault-tolerant quantum computation due to its weight-4 stabilizer checks, relatively high fault-tolerance threshold, planar connectivity and suite of fault-tolerant gadgets~\cite{Horsman_2012, Brown_2017_poking, Litinski_2019, Brown_2020, Chamberland_2022, gidney2024cultivation, cain2025fast, serraperalta2025, Turner_2026}. However, one of the challenges of the surface code is its demanding physical overhead: a distance-$L$ surface code requires at least $L^2$ physical qubits per logical qubit, which is around $10^3$ in the practical regime~\cite{Gidney_2021}. It has also been shown that this overhead scaling cannot be parametrically improved for any geometrically local stabilizer code in two spatial dimensions~\cite{BPT_2010}. 

If the relevant hardware supports long-range connectivity, then one can reduce this spatial overhead to a constant by using constant-rate quantum low-density parity-check (qLDPC) codes~\cite{Freedman_2002_Z2, Delfosse_2013, Breuckmann_2016, HGP_codes, gottesman2014constant, Fawzi_2020_constant, LP_codes, qTanner_codes, DHLV_code, Tamiya_2026, Dinur_2024_cube, nguyen2025constant}. Like the surface code, these codes also possess constant-weight stabilizer checks; but unlike the surface code, the support of these checks could not be geometrically constrained~\cite{Breuckmann_2021_review}. However, since the physical-to-logical overhead is asymptotically constant, one can simply grow the length of the code block until a certain code distance is achieved. Despite these asymptotic promises, the underlying constants can still be quite formidable in practice and present a significant challenge for realizing many qLDPC code families in near-term hardware. For instance, hypergraph product (HGP) codes~\cite{HGP_codes} require a physical qubit block count exceeding a thousand before seeing advantages over the surface code~\cite{Xu_2024_constant}. More modern constructions such as the lifted product~\cite{LP_codes, Panteleev_2022_good} and quantum Tanner~\cite{qTanner_codes} begin to show similar advantages with a few hundred qubits~\cite{Xu_2024_constant, BB_codes, radebold2025explicit}.

Although the HGP may not be as efficient as other qLDPC constructions in terms of spatial overhead for memory, it makes up for this difference in other ways. For instance, due to the natural product structure in the HGP, the support of logical Pauli operators on rows and columns resemble those of the classical input codewords~\cite{quintavalle2022reshape, Quintavalle_2023_fold}. This canonical logical Pauli basis is especially convenient for proving guarantees on fault-tolerant syndrome extraction~\cite{Manes_2025, Tan_2025_eff}, possible transversal gates~\cite{zhu2025fountain, Burton_2022_nogo, fu2025nogo}, as well as the analyses of logical gadgets~\cite{Quintavalle_2023_fold, Xu_2025_fast, zheng2025high, berthusen2025auto}. Similar statements for other qLDPC codes are often weaker or lacking altogether.

For an $L\times L$ (unrotated) surface code, which can be viewed as a HGP of two length-$L$ repetition codes with $L^2+(L-1)^2$ physical qubits, there exists a ``rotated'' version~\cite{Bombin_2007, anderson2011, Kovalev_2012, Tomita_2014} with only $L^2$ physical qubits and also weight-4 checks for the same code distance $L$ and performance under circuit-level noise~\cite{orourke2024}. Instead of a geometric rotation, this qubit reduction can also be viewed as a puncture, i.e.~removing physical qubits, followed by a local repair operation on the punctured stabilizer checks so that they commute~\cite{qTanner_rot_surface, Leverrier_2025_efficient}. In general, any quantum code is amenable to puncturing; however, doing so carelessly will often result in a lower code distance due to the possibility of puncturing a physical qubit in the support of a minimum-weight logical operator. Furthermore, there is also no guarantee on the new weight of the repaired checks, which may no longer be LDPC. The surface code nonetheless has a special structure that allows both its distance and LDPC property to be preserved when performing the puncture that transforms it to the rotated surface code. Due to these practical savings in addition to other reasons, this rotated version has been the code of choice for experiments with the surface code~\cite{Google_2023_SC, Google_2024_SC}. A natural question to ask then is: 

\begin{quote}
    Can the qubit reduction of the rotated surface code be generalized to generic HGP codes? If such a reduction exists, can it preserve most--if not all--of the existing HGP fault-tolerant gadgets?
\end{quote}

In this work, we present a procedure to reduce the number of physical qubits in generic HGP codes while maintaining many of their fault-tolerant gadgets. This reduction comes at a cost of increased qubit and check weights in the modified code which is at most double the original weights. Our numerical simulations show that the performance cost of the reduction under circuit-level noise depends strongly on the syndrome-extraction schedule. For a similar number of data qubits, reducing a larger code can provide a higher distance, and in our finite-size data the larger reduced block has the lower logical error rate at the smallest sampled noise strengths.

Our reduction procedure can, in some sense, be viewed as the converse of weight reduction~\cite{hastings2016weight, hastings2023weight, Sabo_2024, hsieh2025weight}. In weight reduction for classical and quantum codes, one typically introduces extra (qu)bits which partition large checks into sets of smaller checks while maintaining the code dimension and distance. In contrast, we are combining smaller checks into larger checks and discarding physical qubits.

The rest of the paper is structured as follows. In Section \ref{sec:prelims}, we review the HGP construction. In Section \ref{sec:total procedure sec}, we describe the reduction procedure and prove analytical guarantees on code parameters. In Section \ref{sec:gadget compatibility}, we demonstrate compatibility with existing fault-tolerant HGP gadgets such as distance-preserving syndrome extraction, fold-transversal gates, automorphisms and homomorphic measurements. In Section \ref{sec:examples}, we show how to perform qubit reduction on example HGP codes built from random, quasi-cyclic and cycle input codes. In Section \ref{sec:numerics}, we perform numerical QEC memory simulations under circuit-level depolarizing noise to benchmark the performance cost of qubit reduction. Finally, we give concluding remarks and future directions in Section \ref{sec:outlook}.

\subsection{Summary of main results}

The main results of this work are summarized informally in the following statements. A technical assumption we require is that all logical qubits reside in a single HGP sector, which is satisfied when the input parity-check matrices have full row rank. Our first result concerns the maximum qubit and check weights after overhead reduction.

\begin{claim}[Preservation of LDPC --- Claim~\ref{claim:LDPC preservation}]
    If the qubit and check weights are $w_q$ and $w_c$ respectively in the original HGP code, then the new qubit and check weights are at most $\tilde{w}_q \leq 2w_q$ and $\tilde{w}_c \leq \max\{w_c,2(w_c-1)\}$ in the reduced HGP code.
\end{claim}

Preserving the LDPC property of the HGP is important for maintaining a threshold under local-stochastic noise and single-ancilla syndrome extraction~\cite{Kovalev_2013_bad, gottesman2014constant}. Note that the actual qubit and check weights depend strongly on the Tanner graph structure of the classical input codes, and can often be lower than the upper bounds given above. Our next result concerns the structure of logical Pauli operators.

\begin{thm}[Preservation of canonical logical basis --- Theorem~\ref{thm:basis preservation}]
    The canonical logical Pauli basis~\cite{quintavalle2022reshape, Quintavalle_2023_fold} remains a valid logical Pauli basis on the reduced HGP code.
\end{thm}

The preservation of the canonical logical Pauli basis will allow us to port over many of the fault-tolerant gadgets of HGP codes to their reduced versions. The next result concerns the minimum distances of the reduced HGP codes.

\begin{thm}[Preservation of code distance --- Theorem~\ref{thm:distance preservation}]
    If the minimum $X$ and $Z$-distances of the original HGP code are $d_X$ and $d_Z$ respectively, then the minimum distances of the reduced HGP code are $\tilde{d}_X=d_X$ and $\tilde{d}_Z=d_Z$.
\end{thm}

The preservation of the code distance assures us that we do not lose any inherent error-correcting capabilities upon performing overhead reduction. However, due to the increased check and qubit weights, one may worry that the constant savings we achieve in practice may be diminished by a corresponding deterioration in error-correcting performance under circuit-level noise. The next result shows that the gates can be scheduled so that ancillary hook errors do not reduce the effective distance.

\begin{thm}[Distance-preserving syndrome extraction --- Theorem~\ref{thm:d_circ preservation}]
    The reduced HGP codes possess a CNOT scheduling order for single-ancilla syndrome extraction such that the effective distance is the same as the code distance.
\end{thm}

In all, the above results provide overhead savings when using a HGP code as a logical memory. A natural follow-up is to ask whether these savings are compatible with or extend to existing fault-tolerant gadgets for HGP codes. We answer this question in the affirmative for a class of transversal gadgets based on code homomorphisms~\cite{Xu_2025_fast}, where we show that the ancillary gadget block can also be reduced in an equivariant manner to the data block.

\begin{thm}[Preservation of code homomorphisms --- Theorem~\ref{thm:augmentation homomorphism} and \ref{thm:puncturing homomorphism}]
    There exist quantum code homomorphisms between reduced HGP codes that correspond to classical code homomorphisms such as puncturing and augmentation.
\end{thm}

These code homomorphisms, combined with the preservation of the canonical logical Pauli basis, enable grid-addressable logical Pauli-product measurements~\cite{Xu_2025_fast} on the reduced HGP codes. These logical Pauli-product measurements then allow us to perform logical Clifford gates via Pauli-based computation~\cite{Litinski_2019}.

\subsection{Related works}

Shortly after the breakthroughs on asymptotically good LDPC codes, Leverrier and Z\'emor~\cite{Leverrier_2025_efficient} discovered a link between their quantum Tanner codes~\cite{qTanner_codes} and the lifted product codes by Panteleev and Kalachev~\cite{Panteleev_2022_good}, which also extends to the case of HGP codes and is the main motivator for our work. It was later pointed out by Breuckmann and Leverrier~\cite{qTanner_rot_surface} that this general procedure encompasses the rotated surface code. However, their transformation is restricted to input LDPC codes with a bipartite Tanner graph structure, and we generalize their approach to arbitrary input LDPC codes.

Kovalev and Pryadko~\cite{Kovalev_2012} previously proposed schemes to improve the rate of HGP codes by modifying the input codes but still retaining the physical HGP structure. In their symmetric square-matrix method, the logical qubits are split evenly between the \Bsectorname{} and \Csectorname{} sectors. In comparison, all of the logical qubits in our formulation reside solely in the \Bsectorname{} sector. For their approach to apply generically, one needs to symmetrize a non-symmetric LDPC matrix via $H^\transp H$, which can lead to a quadratic increase in the LDPC weights. In contrast, our approach can only double the LDPC weights. However, our reduction does not maintain the physical HGP structure and so certain HGP gadgets, while naturally compatible with their scheme, may not be compatible with our scheme.

Berthusen et al.~\cite{Berthusen_2025_adaptive} introduced a local concatenation scheme to lower the spacetime overhead of HGP codes in terms of circuit volume, which can be thought of as dynamical overhead reduction in the time direction. In comparison, our scheme lowers the static spatial overhead. It would be interesting to combine the two strategies to potentially achieve an even lower spacetime overhead.


\section{Preliminaries}\label{sec:prelims}

\subsection{Classical and quantum codes}

\begin{figure}
    \centering
    \includegraphics[width=1\linewidth]{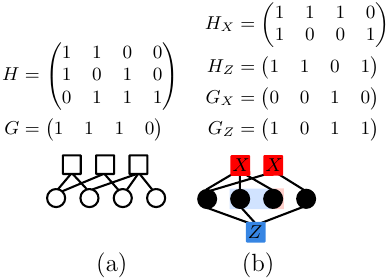}
    \caption{\textbf{(a)} Tanner graph for a classical linear code with parity check matrix $H$ and generator matrix $G$. The circular vertices are bits and squares classical checks. Corresponding rows and columns to checks and bits are ordered left-to-right. \textbf{(b)} Tanner graph for a quantum stabilizer CSS code with parity check matrices $H_X$, $H_Z$ and generator matrices $G_X$, $G_Z$. The filled circular vertices are qubits, red squares $X$ checks, and blue squares $Z$ checks. Corresponding rows and columns to checks and qubits are ordered left-to-right. The support of the logical $\bar{X}$ and $\bar{Z}$ operators are highlighted in red and blue respectively, and can be seen to overlap on qubit 3.}
    \label{fig:tanner graphs}
\end{figure}

An $[n,k,d]$ classical linear code encodes $k$ logical bits in $n$ physical bits. The logical bits span a $k$-dimensional subspace $\mathcal{C} \simeq \mathbb{F}^k_2$, called the codespace, of the vector space of all length-$n$ bitstrings, $\mathbb{F}^n_2$. The elements of this subspace are called codewords. The minimum distance, or code distance $d$, is the minimum nonzero Hamming weight among all codewords. A complete basis for the codespace can be specified by the rows of a generator matrix $G \in \mathbb{F}^{k\times n}_2$ such that $\im G^\transp = \rs G = \mathcal{C}$. Under this basis, a logical message $\mathbf{y} \in \mathbb{F}^k_2$ is translated into its corresponding codeword $\mathbf{x} \in \mathbb{F}^n_2$ via $\mathbf{x} = G^\transp \mathbf{y}$. For a given generator matrix $G$, there exists a dual matrix called a parity-check matrix $H \in \mathbb{F}^{m\times n}_2$ whose rows span $\ker G$; i.e. $HG^\transp=0$. This suggests that all codewords (rows of $G$) satisfy all of the parity checks (rows of $H$). The rank-nullity theorem implies that $m \geq n-k$. A graphical representation of a classical code is given by the code's Tanner graph $\mathcal{G} = \mathpzc{T}(B, C, E)$, a bipartite graph with vertex set $B\cup C$ and edge set $E$ such that the vertices in $B$ represents bits, $C$ parity checks, and the edges connect checks to the bits that they check; see Figure~\ref{fig:tanner graphs}.

An $\llbracket n,k,d \rrbracket$ quantum code encodes $k$ logical qubits in a $2^k$-dimensional subspace, or codespace, of the physical $2^n$-dimensional Hilbert space on $n$ qubits. The code distance, or minimum distance, $d$ is the minimum number of local errors needed to nontrivially affect its codespace, or corrupt its encoded information. A stabilizer code is a quantum error-correcting code defined by its stabilizer group---an abelian subgroup of the Pauli group on $n$ physical qubits~\cite{gottesman1997thesis}. A generating set for the stabilizer group is often specified, and its elements are called check operators, or colloquially ``stabilizers''. The codespace of a stabilizer code is the mutual +1 eigenspace of all check operators and thereby also all Pauli operators in the stabilizer group. A Calderbank-Shor-Steane (CSS) stabilizer code is a stabilizer code whose stabilizer group can be generated by a set of $X$-type and $Z$-type Pauli operators~\cite{Calderbank_1996, Steane_1996_1, Steane_1996_2}. For a CSS code, it is customary to specify two binary parity-check matrices $H_X \in \mathbb{F}^{m_X\times n}_2$ and $H_Z \in \mathbb{F}^{m_Z\times n}_2$ whose rows specify the support of $X$-type and $Z$-type check operators respectively. The $X$-type and $Z$-type stabilizer subspaces are given by $\mathcal{S}_X=\rs{H_X}$ and $\mathcal{S}_Z=\rs{H_Z}$ respectively. The stabilizer commutation condition is then equivalent to the orthogonality condition $H^{}_X H^\transp_Z = 0$. Since each independent $X$-check and $Z$-check constrains the Hilbert space, the number of logical qubits is
\begin{align}\label{eq:CSS k}
    k = n - \rank{H_X} - \rank{H_Z} \, .
\end{align}
For a CSS code, we also have two types of logical operators corresponding to each Pauli type. Since elements of $\mathcal{S}_X$ act trivially on the codespace, logical $\bar{X}$ operators are grouped according to equivalence classes under $\mathcal{S}_X$. The group of distinct logical $\bar{X}$ actions is isomorphic to $\ker{H_Z}/\rs{H_X}$. Note that this group is distinct from the subspace of all nontrivial $\bar{X}$ operators $\ker{H_Z}\backslash\rs{H_X} \subset \mathbb{F}^n_2$.
For a CSS code, we also have two types of minimum distances corresponding to each Pauli type. The $X$-type minimum distance, or $X$-distance, is defined as the minimum Hamming weight of any nontrivial logical $\bar{X}$ operator:
\begin{align}\label{eq:CSS d_X}
    d_X := \min_{\mathbf{x}\in\ker{H_Z}\backslash\rs{H_X}} \abs{\mathbf{x}} \, ,
\end{align}
where $\abs{\cdot}$ denotes the Hamming weight. The $Z$-distance $d_Z$ is defined analogously upon swapping the roles of $X$ and $Z$. 

An algebraic formulation for classical linear and quantum CSS codes in terms of chain complexes is reviewed in Appendix \ref{app:codes algebraic}.

\subsection{Hypergraph product codes}\label{subsec:HGP review}

We now briefly review the HGP construction as well as recite some relevant results for our purposes. For a more in-depth introduction with derivations, we refer the reader to both Appendix \ref{app:HGP} and the original paper by Tillich and Z\'emor~\cite{HGP_codes}.

Given two input classical linear codes with parity-check matrices $H_1 \in \mathbb{F}^{m_1\times n_1}_2$ and $H_2 \in \mathbb{F}^{m_2\times n_2}_2$, their corresponding HGP code $\text{HGP}(H_1, H_2)$ has CSS parity-check matrices
\begin{subequations}\label{eq:HGP H_X,H_Z}
\begin{align}
    H_X &= \big(\, H_1\otimes\ident_{n_2} \,\mid\, \ident_{m_1} \otimes H^\transp_2 \,\big) \, , \label{eq:HGP H_X} \\
    H_Z &= \big(\, \ident_{n_1} \otimes H_2 \,\mid\, H^\transp_1 \otimes \ident_{m_2} \,\big) \, .  \label{eq:HGP H_Z}
\end{align}
\end{subequations}
The CSS orthogonality condition is satisfied since $H^{}_X H^\transp_Z = 2H^{}_1\otimes H^\transp_2 = 0$ over $\mathbb{F}_2$. Geometrically, the HGP resembles a Cartesian graph product between the Tanner graphs $\mathcal{G}_1$, $\mathcal{G}_2$ of the classical input codes; see Figure \ref{fig:HGP} for an illustration. The number of physical or data qubits is the number of columns in \eqref{eq:HGP H_X,H_Z}:
\begin{align}
    n = \underbrace{n_1n_2}_{\mathscr{B}\text{-type}} + \underbrace{m_1m_2}_{\mathscr{C}\text{-type}} \, ,
\end{align}
where $n_1n_2$ is the number of \Bsectorname{} ($\mathscr{B}$-type) qubits, and $m_1m_2$ is the number of \Csectorname{} ($\mathscr{C}$-type) qubits, corresponding to the left and right sides of \eqref{eq:HGP H_X,H_Z} with respect to the vertical partition. It will be useful to label the data qubits according to the 2D layout of a graph product: the left and right sides of the tensor products in \eqref{eq:HGP H_X,H_Z} index rows and columns respectively (Figure \ref{fig:HGP}).

\begin{figure}[t]
  \centerline{\includegraphics[width=0.35\textwidth]{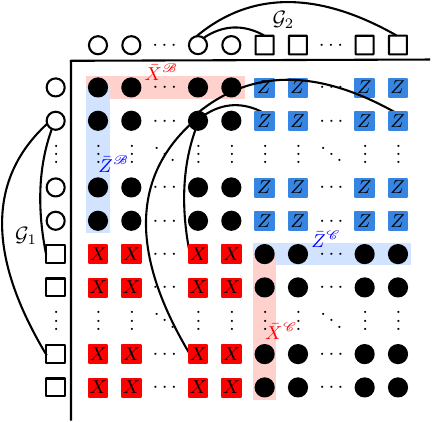}}
  \caption{General structure of the hypergraph product. The top-left section of qubits are the \Bsectorname{} qubits, and the bottom-right section of qubits are the \Csectorname{} qubits. Some edges are explicitly shown; the rest of the edges (not shown in the figure) inherit the structure from the classical input codes. The rows/columns that contain the support of the \Bsectorname{} logical $\bar{X}^\mathscr{B}$, $\bar{Z}^\mathscr{B}$, and \Csectorname{} logical $\bar{X}^\mathscr{C}$, $\bar{Z}^\mathscr{C}$ operators are shown.}
  \label{fig:HGP}
\end{figure}

For an $[n, k, d]$ classical code with parity-check matrix $H \in \mathbb{F}^{m\times n}_2$, let $[m,k^\transp,d^\transp]$ denote the parameters of the transpose code with parity-check matrix $H^\transp$. From this, the number of logical qubits in the HGP code \eqref{eq:HGP H_X,H_Z} is
\begin{align}\label{eq:HGP k}
    k = \underbrace{k^{}_1k^{}_2}_{\mathscr{B}\text{-type}} + \underbrace{k^\transp_1 k^\transp_2}_{\mathscr{C}\text{-type}} \, .
\end{align}
As hinted by the underbraces in \eqref{eq:HGP k}, our logical qubits can be partitioned into those that ``live'' on the \Bsectorname{} and \Csectorname{} qubits. Nearly all HGP codes of interest, and especially those from random classical LDPC codes, have $k^{}_1k^{}_2 \gg k^\transp_1k^\transp_2$ (often $k^\transp_1k^\transp_2=0$), and so we will mainly be interested in the $k_1k_2$ \Bsectorname{} logical qubits. A canonical basis of logical Pauli operators for the $k_1k_2$ \Bsectorname{} logical qubits is given by~\cite{quintavalle2022reshape}
\begin{subequations}\label{eq:HGP G_X,G_Z B}
    \begin{align}
        G_X^\mathscr{B} &= \big(\, E_1^\mathscr{B} \otimes G_2 \,\mid\, \mathbf{0} \,\big) \, ,  \label{eq:HGP G_X} \\
        G_Z^\mathscr{B} &= \big(\, G_1 \otimes E_2^\mathscr{B} \,\mid\, \mathbf{0} \,\big) \, ,  \label{eq:HGP G_Z}
    \end{align}
\end{subequations}
where $G_i$ ($i=1,2$) is a generator matrix for $H_i$, and $E^\mathscr{B}_i \in \mathbb{F}^{k_i\times n_i}_2$ is a matrix whose rows are unit vectors chosen outside $\rs H_i$ and form a basis for the complementary subspace; i.e. $\mathbb{F}^{n_i}_2 = \rs H_i \oplus \rs E^\mathscr{B}_i$, and is chosen so that $E_i^\mathscr{B}G_i^\transp=\ident_{k_i}$. If $G_i$ is in canonical form $G_i=(\ident_k\,|\,A)$, then $E^\mathscr{B}_i = (\ident_k \,|\,\mathbf{0})$. A canonical basis for the $k^\transp_1k^\transp_2$ \Csectorname{} logical qubits can be constructed similarly by exchanging the left and right sides of \eqref{eq:HGP G_X,G_Z B} as well as using the transpose codes. From the canonical basis \eqref{eq:HGP G_X}, we see that our canonical \Bsectorname{} logical $\bar{X}$ operators resemble $k_1$ copies of codewords of the second classical input code supported on $k_1$ rows of the \Bsectorname{} qubits. Likewise, our canonical \Bsectorname{} logical $\bar{Z}$ operators resemble codewords of the first input code supported on \Bsectorname{} columns. The structure on the \Csectorname{} qubits follows similarly; see Appendix \ref{app:HGP} for details.

Since the canonical Pauli basis for our logical qubits mimic the codewords of the classical input codes, we obtain the following upper bounds on the $X$ and $Z$ distances of the HGP code:
\begin{subequations}\label{eq:HGP d upper bounds}
\begin{align}
    d_Z &\leq \min\left(d^{}_1,d^\transpose_2\right) \\
    d_X &\leq \min\left(d^\transpose_1,d^{}_2\right) \, .
\end{align}
\end{subequations}
We now recite a useful result, first shown by Quintavalle and Campbell (\cite{quintavalle2022reshape}, Prop. 2), concerning the structure of logical operators within the \Bsectorname{} and \Csectorname{} sectors. For a vector $\mathbf{x} \in \mathbb{F}^n_2$, let $\mathbf{x}^\mathscr{B}$ and $\mathbf{x}^\mathscr{C}$ denote its restriction to the \Bsectorname{} and \Csectorname{} qubits respectively. Furthermore, let $\abs{\cdot}_{\rm r}$ and $\abs{\cdot}_{\rm c}$ denote the number of supported rows and columns respectively.
\begin{prop}[HGP logical sector weights~\cite{quintavalle2022reshape}]
\label{prop:HGP logical sector weights}
    For any nontrivial logical $\bar{X}$ with $X$-support $\mathbf{x} \in \ker{H_Z}\setminus \rs{H_X}$, we have
    \begin{align}\label{eq:X sector weights}
        \abs{\mathbf{x}^\mathscr{B}}_{\rm c} \geq d_2 \quad\text{or}\quad \abs{\mathbf{x}^\mathscr{C}}_{\rm r} \geq d_1^\transp\, .
    \end{align}
    Likewise, for any nontrivial logical $\bar{Z}$ with $Z$-support $\mathbf{z} \in \ker{H_X}\setminus \rs{H_Z}$,
    \begin{align}\label{eq:Z sector weights}
        \abs{\mathbf{z}^\mathscr{B}}_{\rm r} \geq d_1 \quad\text{or}\quad \abs{\mathbf{z}^\mathscr{C}}_{\rm c} \geq d_2^\transp\, .
    \end{align}
\end{prop}
Every representative of a \Bsectorname{} logical $\bar X$ ($\bar Z$) has support on at least $d_2$ columns ($d_1$ rows) of the \Bsectorname{} data qubits. Stabilizer deformations therefore cannot reduce these row or column counts below the lower bounds. The same follows for the \Csectorname{} logical operators.

An immediate consequence of Proposition~\ref{prop:HGP logical sector weights} is that
\begin{subequations}
\begin{align}
    d_Z &\geq \min\left(d^{}_1,d^\transpose_2\right)  \\
    d_X &\geq \min\left(d^\transpose_1,d^{}_2\right) \, .
\end{align}
\end{subequations}
Combined with the upper bounds \eqref{eq:HGP d upper bounds}, we arrive at the expression for the minimum distances of an HGP code:
\begin{subequations}
\begin{align}
    d_Z &= \min\left(d^{}_1,d^\transpose_2\right)  \label{eq:HGP d_Z}  \\
    d_X &= \min\left(d^\transpose_1,d^{}_2\right) \, . \label{eq:HGP d_X}
\end{align}
\end{subequations}
For our purposes, we are only interested in the $k_1k_2$ \Bsectorname{} logical qubits; recall that this is a mild restriction since we typically have $k^{}_1k^{}_2 \gg k^\transp_1k^\transp_2$. This restriction can be satisfied explicitly by enforcing either $H_1$ or $H_2$ to have full row rank so that $k^\transp_1k^\transp_2=0$. Alternatively, we can simply ignore all \Csectorname{} logical qubits and treat them as gauge degrees of freedom, akin to a subsystem code. Proposition~\ref{prop:HGP logical sector weights} assures us that this gauging does not affect the \Bsectorname{} sector.


\section{Qubit reduction procedure}\label{sec:total procedure sec}

We motivate this section with a simple observation. Recall that we are only storing logical information in the \Bsectorname{} logical sector. Since Proposition~\ref{prop:HGP logical sector weights} tells us that the minimum weight of logical operators in the \Bsectorname{} sector is a stabilizer-invariant, we can freely remove/puncture the check-type qubits without worrying about lowering the minimum distances of the HGP code. However, these \Csectorname{} qubits were needed to ensure commutation between $X$-checks and $Z$-checks. Thus, we desire a procedure that can simultaneously remove \Csectorname{} qubits and maintain stabilizer commutation. Ideally, such a procedure can also preserve the locality of the checks and hence the LDPC property of the code.

The main idea of removing \Csectorname{} qubits is to pick a check-type qubit and a Pauli type ($X$ or $Z$), take linear combinations of stabilizers of that Pauli type that are connected to that check-type qubit to remove support from it, and discard support of the other Pauli type. This removes all support off that check-type qubit and it can thus be discarded. It is enough to clean only one Pauli family from a qubit before deleting it, since, for any $X$-check $x$ and $Z$-check $z$, deleting a check-type qubit at coordinate $q$ changes their overlap parity by $x_qz_q$. If every transformed $X$-check has $x_q=0$, or every transformed $Z$-check has $z_q=0$, this contribution vanishes for every pair and commutation is unchanged.

One naive approach for combining stabilizer checks might be to (\emph{i}) take any \Csectorname{} qubit that you want to remove, (\emph{ii}) take linear combinations of the stabilizers connected to that qubit to remove all support off of it, (\emph{iii}) discard that \Csectorname{} qubit, and (\emph{iv}) repeat for every \Csectorname{} qubit, thereby serially removing the \Csectorname{} qubits. However, in most cases, combining any two checks will increase the weight of the merged check on the \Bsectorname{} qubits. If we combine any single check too many times, we risk losing the LDPC property of the code. This leads us to our first restriction for the procedure. 
\begin{quote} 
\textbf{Restriction \#1:} Each new check is the product of at most two original checks, and each original check enters at most two new checks.
\end{quote}
We combine a check at most once since, as we will see later, this enables us to construct distance-preserving syndrome extraction circuits.

Our second restriction comes from noting the fact that this ``parallel reduction'' does not generically work when \Csectorname{} qubits share stabilizers. The issue is that we can only safely remove a \Csectorname{} qubit when we clean all $X$-checks or $Z$-checks off of it; otherwise we risk breaking stabilizer commutation. If we simultaneously try to clean a \Csectorname{} qubit that shares, for instance, an $X$-check with another \Csectorname{} qubit, we risk creating a new $X$-check with support on that second qubit. Removing this second qubit may cause the new $X$-check to fail to commute with its neighboring $Z$-checks. This obstruction leads to our second restriction for the procedure.
\begin{quote}
\textbf{Restriction \#2:} The procedure must work on a carefully--selected subset of \emph{disjoint} stabilizers that do not share support on any qubits. 
\end{quote}

Specifically, the way we select this subset is as follows. If we partition the \Csectorname{} qubits into groups that do not share any stabilizer support, we can combine stabilizers in those groups in parallel. At a high level, we choose this partition by coloring the check nodes of the classical input codes' Tanner graphs such that two classical checks that have the same color do not share any bits in common. Such colorings can be found by coloring codes' check-adjacency graphs, defined below. Finding a minimal coloring of a graph $G = (V, E)$ is NP-hard, but greedy heuristics can be used to find a valid (not necessarily minimal) coloring in $\mathcal{O}(|V| + |E|)$ time. In our case, fewer colors will equate to more \Csectorname{} qubits removed.

\begin{figure}[t]
    \centerline{\includegraphics[scale=1]{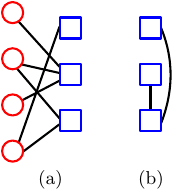}}
    \caption{ \textbf{(a)} Tanner graph $\mathcal{G} = \mathpzc{T}(B, C, E)$; \textbf{(b)} $\mathcal{G}$'s check-adjacency graph $\mathcal{G}_C = (C, E_C)$.}
    \label{bex}
\end{figure}

\begin{defn}[Check-adjacency graph]\label{def:check-adjacency graph}
    Given a Tanner graph $\mathcal{G} = \mathpzc{T}(B, C, E)$, define the check-adjacency graph as the simple graph
    \begin{align}
        \mathcal{G}_C = (C, E_C) \label{ce}
    \end{align}
    where an edge $(c_1, c_2) \in E_C$ exists iff there exists a bit $b \in B$ such that $(b, c_1) \in E$ and $(b, c_2) \in E$. In other words, edges in the check-adjacency graph only exist if the checks share a mutual bit in the Tanner graph. See Figure~\ref{bex} as an example.
\end{defn}

The result is that if we can $\chi_1$-color the checks of $H_1$ and $\chi_2$-color the checks of $H_2$, the \Csectorname{} qubits of the HGP code will be divided into $\chi_1 \chi_2$ product-color groups, within which there are no shared ($X$ or $Z$) stabilizer support. Due to the restriction, we can only choose to combine $X$ ($Z$) stabilizers in one color group per row (column). Additionally, let us restrict ourselves to never combine $X$ and $Z$ in the same color group. Doing so would be suboptimal since that would clearly equate to less \Csectorname{} qubits removed.

Finally, we note that combining $X$ ($Z$) stabilizers may change the color groups in the same row (column), which is demonstrated with a minimal example in Figure~\ref{ce2}. Thus, if we wish to perform more than one round of reduction, we have to re-color at the end of each round. However, we just concern ourselves with one round here as, which we will see later, one round guarantees distance-preserving syndrome extraction.

\begin{figure}[t]
  \centerline{\includegraphics[width=0.5\textwidth]{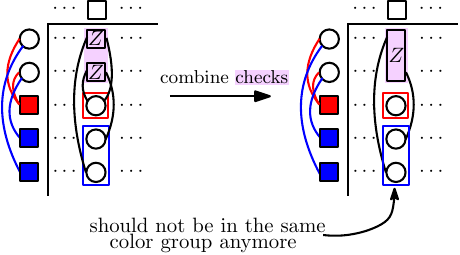}}
  \caption{Combining $Z$ stabilizers to remove the \Csectorname{} qubit boxed in red. As a side effect, this nullifies the previous color group of the \Csectorname{} qubits boxed in green.}
  \label{ce2}
\end{figure}

We can mathematically represent this single-round, parallel reduction using linear algebra. Let $\tilde{n}$ denote the number of data qubits after the reduction. Define the matrices $W_X \in \F_2^{\tilde{m}_X\times m_X}$, $W_Z \in \F_2^{\tilde{m}_Z\times m_Z}$ whose rows label the linear combinations of $X$ and $Z$ stabilizers taken, or equivalently, rows in $H_X$ and $H_Z$ respectively. In addition, define the matrix $V\in \F_2^{n \times \tilde{n}}$ whose columns select the qubits that are kept. We now have the following relations between the old parity-check matrices $H_X$, $H_Z$ and the new parity-check matrices $\widetilde{H}_X$, $\widetilde{H}_Z$:
\begin{subequations}\label{eq:tildeH,H relation}
    \begin{align}
        \widetilde{H}_X &= W_X H_X V  \label{eq:tildeHx} \\
        \widetilde{H}_Z &= W_Z H_Z V \, .  \label{eq:tildeHz}
    \end{align}
\end{subequations}
Note that $V$ is of the form (upon appropriate rearranging of its columns based on qubit labeling)
$V \sim \big( \ident_{\tilde{n}} \,\mid\, \mathbf{0}_{(n-\tilde{n})\times \tilde{n}} \big)^\transp$
to indicate that we select $\tilde n$ qubits after combining and discard the rest.

With the above motivation in mind, we now introduce how we construct $W_X,W_Z$ and $V$ in the next section.

\subsection{General procedure} \label{sec:procedure}


    Let us come up with some more precise language for this coloring procedure. For an integer $n$, let $[n] = \{1, \dots, n\}$. Suppose we have colorings of the checks of the classical Tanner graphs as $\Gamma^{(i)} = \{\Gamma^{(i)}_j\}_{j = 1}^{\chi_i}$ ($i = 1, 2$) with $|\Gamma^{(i)}| = \chi_i$. In other words, the Tanner graphs can be $\chi_i$-colored into sets $\Gamma^{(i)}_j$ for $j \in [\chi_i]$. This coloring will sector the \Csectorname{} qubits of the HGP code into a product coloring $\Gamma = \Gamma^{(1)} \times \Gamma^{(2)} = \{\Gamma_{i, j}\}_{i = 1,\,j = 1}^{\chi_1,\,\chi_2}$; i.e. the coloring contains $|\Gamma|=\chi_1\chi_2$ color groups of the form $\Gamma_{i,j}$ for $i, j \in [\chi_1] \times [\chi_2]$, each with $|\Gamma_{i, j}|$ \Csectorname{} qubits. We next need to decide which color groups we choose to combine $X$ and $Z$ stabilizers within.
    

    

    \begin{figure}[t]
      \centerline{\includegraphics[width=0.5\textwidth]{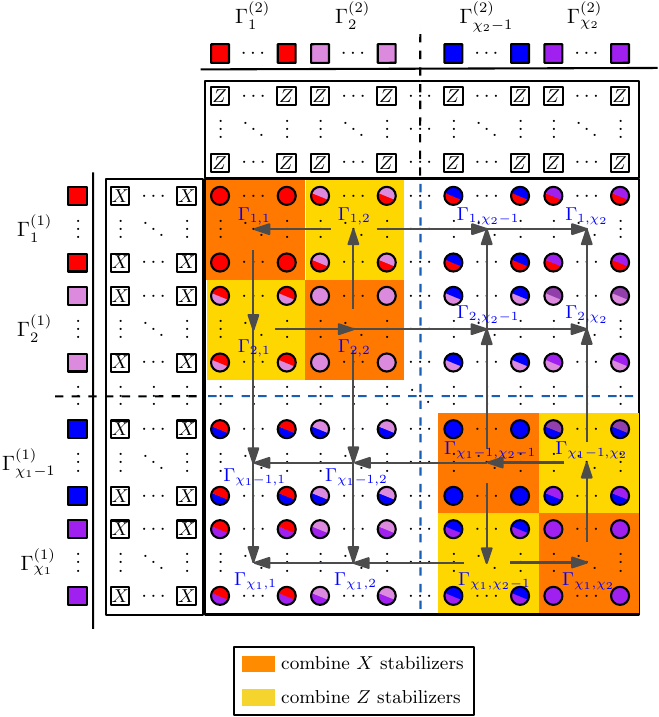}}
      \caption{Example of parallel qubit reduction based on check-colorings of the classical Tanner graphs. Here we assume that $\chi_1 = \chi_2$ and that they are both even, and choose to combine $X$ checks along the main diagonal color groups $\{\Gamma_{i, i}\}_{i=1}^{i=\min\{\chi_1, \chi_2\}}$ and $Z$ in color groups $\{\Gamma_{i, i+1}, \Gamma_{i+1, i}\}_{i=1,3,5,\dots}^{i=\min\{\chi_1, \chi_2\}}$. Stabilizer support from combining in a particular color group is pushed into the color groups indicated by the outgoing gray arrow starting not in the center of the color group.}
      \label{blockcolor}
    \end{figure}

    \subsubsection{Choosing the color groups}\label{optimization section}
    
    Considering the fact that not all color groups have the same number of \Csectorname{} qubits, we optimize the number of \Csectorname{} qubits removed by reducing the problem in polynomial time to maximum-weight matching on a bipartite graph $\mathpzc{G} = (V, E)$, which is solvable in $O(|V|^3)$ time via the Hungarian algorithm~\cite{Edmonds_1972, Tomizawa_1971}.

\begin{figure*}
    \centering
    \includegraphics[width=1\linewidth]{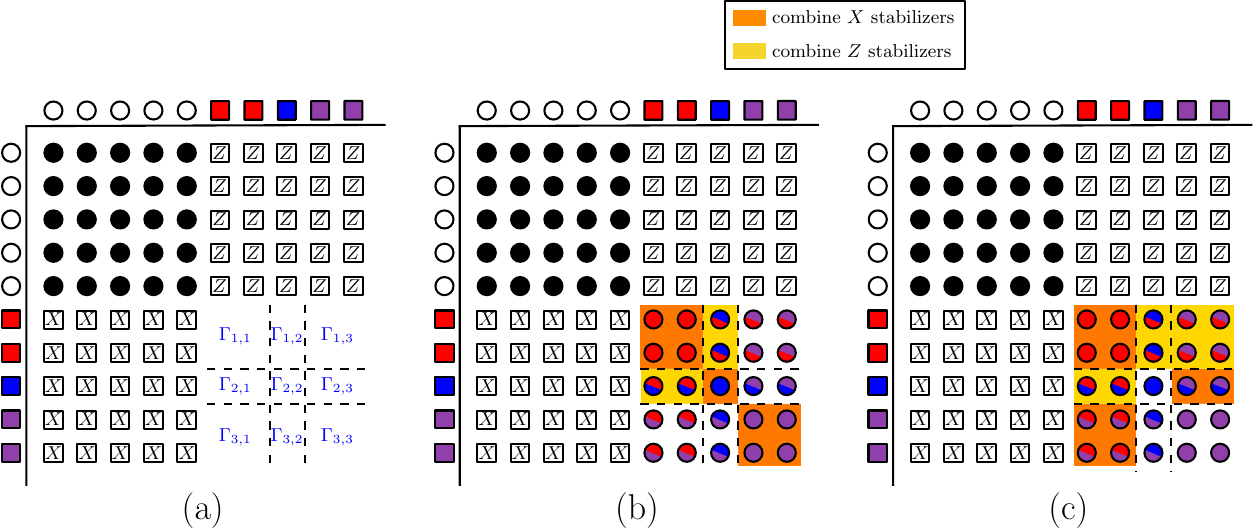}
    \caption{Example HGP graph with colorings of \Csectorname{} qubits denoted. Since the classical checks are 3-colorable ($\chi_1=\chi_2=3$), the \Csectorname{} qubits are split into 9 color groups. \textbf{(a)} The labeling of the 9 color groups $\Gamma_{i, j}$ for $i, j\in \{1, 2, 3\}$. \textbf{(b)} The (unoptimal) example schedule for combining stabilizers introduced in Figure~\ref{blockcolor}. We combine $X$ stabilizers on color groups $\Gamma_{1, 1}$, $\Gamma_{2, 2}$, and $\Gamma_{3, 3}$; and combine $Z$ stabilizers on color groups $\Gamma_{1, 2}$ and $\Gamma_{2, 1}$. This schedule gets rid of 13 \Csectorname{} qubits. \textbf{(c)} An optimal schedule to remove the most \Csectorname{} qubits. By combining $X$ stabilizers on color groups $\Gamma_{1, 1}$, $\Gamma_{2, 3}$, $\Gamma_{3, 1}$ and combining $Z$ stabilizers on color groups $\Gamma_{2, 1}$, $\Gamma_{1, 2}$, $\Gamma_{1, 3}$, we remove 18 \Csectorname{} qubits.}
    \label{hgpextboptimized}
\end{figure*}

In line with restriction \#2, all we need to enforce for our schedule is to combine $X$ ($Z$) stabilizers in at most one color group per row (column) so that we don't combine any $X$ ($Z$) that share \Csectorname{} qubits in common. We present a polynomial reduction of this problem to a maximum-weight matching problem on a bipartite graph.

\begin{lem}\label{max-wm proof}
    Solving the stabilizer combination problem proposed at the end of Section~\ref{optimization section} reduces to finding a maximum-weight matching on a bipartite graph.
\end{lem}

\begin{proof}
    Recall that $\chi_i$-colorable classical checks on $\mathcal{G}_i$ ($i=1, 2$) create $\chi_1\chi_2$ color groups on the \Csectorname{} qubits. Due to restriction $\#2$, we cannot combine $X$ ($Z$) stabilizers on more than $\chi_1$ ($\chi_2$) color groups, which is a maximum of $\chi_1 + \chi_2$ color groups. Construct a bipartite graph $\mathpzc{G} = (V_L\cup V_R, E)$ as follows:
    \begin{enumerate}
        \item \textbf{Left vertex set.} Let the left vertex set $V_L$ have $\chi_1+\chi_2$ nodes. We are choosing up to this many color groups to combine both $X$ and $Z$ stabilizers, so we will end up matching the vertices of this set.
        \item \textbf{Right vertex set.} Let the right vertex set $V_R$ have $\chi_1\chi_2$ nodes. These vertices represent the color group candidates that we can choose for each vertex in $V_L$, so let us label them as the same thing:
        \begin{equation}
            V_R = \{\Gamma_{i, j}\}_{i,j =1}^{i = \chi_1,\,j=\chi_2}\, . 
        \end{equation}
        \item \textbf{Edge set.} Arrange the first $\chi_1$ ($\chi_2$) vertices of $V_L$ into sets $\mathcal{X}$ ($\mathcal{Z}$) such that \begin{equation}
            V_L = \mathcal{X} \cup \mathcal{Z} = \{\mathcal{X}_i\}_{i=1}^{i=\chi_1} \cup \{\mathcal{Z}_j\}_{j=1}^{j=\chi_2}\, .
        \end{equation}
        Let the $i^\text{th}$ ($j^\text{th}$) element of $\mathcal{X}$, $\mathcal{X}_i$ ($\mathcal{Z}_j$), have edges connected to the vertices in $V_R$ corresponding to the color groups of the $i^\text{th}$ row ($j^\text{th}$ column). Mathematically, the edge set is 
        \begin{equation}
        E=\Big\{(\mathcal{X}_i, \Gamma_{i, j}), (\mathcal{Z}_j, \Gamma_{i, j})\Big\}_{i, j = 1}^{i=\chi_1,\,j=\chi_2}\, ,\label{reductionedges}
        \end{equation}
        which is bipartite between $V_L$ and $V_R$. These connections represent the possible candidates for color group to combine on the $i^\text{th}$ row ($j^\text{th}$ column).
        \item \textbf{Edge weights.} Edges $(v_{i, j}\in V_L, \Gamma_{i, j} \in V_R) \in E$ have weight $w_{i, j} = |\Gamma_{i, j}|$. If $v_{i, j} = \mathcal{X}_i$ ($\mathcal{Z}_j$), this represents the number of \Csectorname{} qubits that one can remove by choosing to combine $X$ ($Z$) stabilizers in the color group $\Gamma_{i, j}$.
    \end{enumerate}
    \begin{figure*}
        \centering
        \includegraphics[width=.75\linewidth]{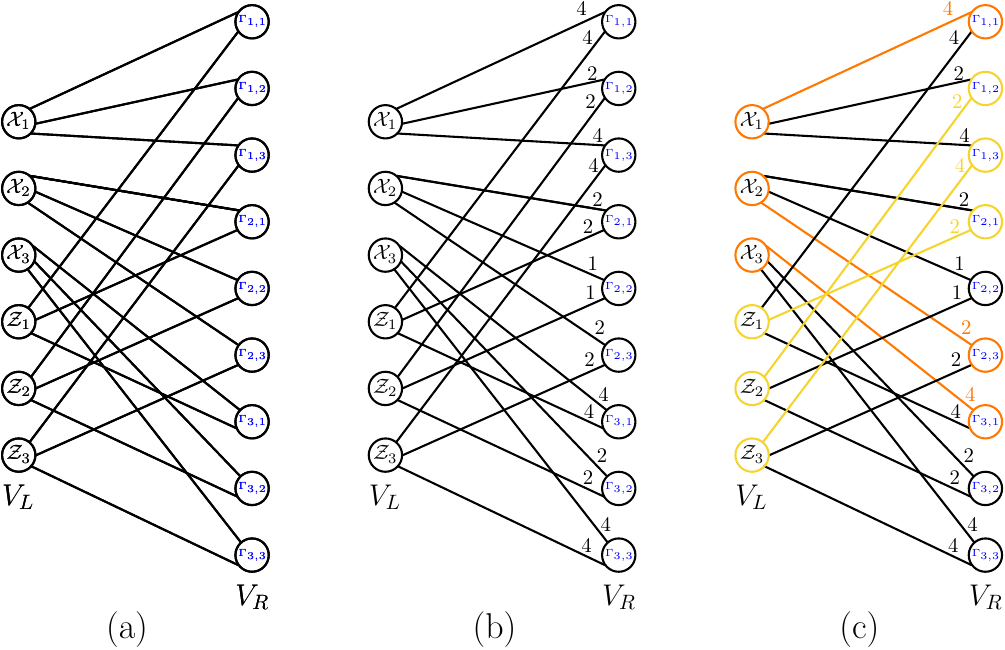}
        \caption{Example bipartite graph $\mathpzc{G}$ constructed from the HGP in Figure~\ref{hgpextboptimized}. \textbf{(a)} Organization of the vertex and edge orientation of the graph. An edge that ends at vertex $\Gamma_{i, j} \in V_R$ exists if the color group $\Gamma_{i, j}$ is a candidate for being combined by the vertex that the edge starts at, consistent with \eqref{reductionedges}. \textbf{(b)} Weights of the edges in the graph of (a). The weight of edge ending at vertex $\Gamma_{i, j} \in V_R$ is the number of \Csectorname{} qubits inside the color group $\Gamma_{i, j}$. \textbf{(c)} A maximum-weight matching of the graph in (b), denoted by the vertices and edges in orange and yellow. Note that this perfectly corresponds to the optimal stabilizer combining schedule in Figure~\ref{hgpextboptimized}(c), where the same color groups that we combine $X$ ($Z$) stabilizers in, colored in orange (yellow), correspond to the same right-side vertices in the matching. If we sum the weights of the matching, this corresponds to the 18 \Csectorname{} qubits removed.}
        \label{hgpoptreduction}
    \end{figure*}
    This bipartite graph construction leads us to the following claim for matching on the graph, which is sufficient in proving the lemma:
    \begin{claim}
        \textbf{Feasibility:} A stabilizer combination schedule is valid iff it corresponds to a matching on $\mathcal{G}$. \textbf{Optimality:} A maximum-weight matching on $\mathcal{G}$ corresponds to the optimal number of \Csectorname{} qubits removed.
    \end{claim}
    \begin{proof}
        \textbf{Feasibility.}
        
        ($\Longrightarrow$) Suppose we have a valid stabilizer combination schedule $S = (\Gamma^{(X)}, \Gamma^{(Z)})$ where $\Gamma^{(X)}$ ($\Gamma^{(Z)}$) $\subseteq \{\Gamma_{i, j}\}_{i = 1,\, j = 1}^{i = \chi_1,\, j = \chi_2}$ contains the color groups that we combine $X$ ($Z$) within. The correspondent matching is 
        \begin{equation}
            M(S)=\Big\{ (\mathcal{X}_i, \Gamma_{i, j}) \mid \Gamma_{i, j} \in \Gamma^{(X)} \Big \} \cup \Big\{ (\mathcal{Z}_j, \Gamma_{i, j}) \mid \Gamma_{i, j} \in \Gamma^{(Z)} \Big\} \,.
        \end{equation}
        This is a valid matching because we meet the following constraints.
        \begin{itemize}
            \item \emph{No vertex in $V_R$ is matched to two vertices in $V_L$.} Since $S$ is valid, we must have $\Gamma^{(X)} \cap \Gamma^{(Z)} = \emptyset$ as we don't combine $X$ and $Z$ on the same color groups.
            \item \emph{No vertex in $V_L$ is matched to two vertices in $V_R$.} Since $S$ is valid, we choose at most one color group in row $i \in [\chi_1]$ (column $j \in [\chi_2]$) to combine $X$ ($Z$) stabilizers. Hence no $\mathcal{X}_i$ ($\mathcal{Z}_j$) $\in V_L$ is incident to two chosen edges.
        \end{itemize}
        Thus $M(S)$ is a matching on $\mathpzc{G}$.
        
        ($\Longleftarrow$) Suppose we have any matching $M$ of $\mathpzc{G}$. We then have $S(M) = (\Gamma^{(X)}, \Gamma^{(Z)})$ with
        \begin{equation}
            \begin{aligned}
                \Gamma^{(X)} &= \Big\{\Gamma_{i, j} \mid (\mathcal{X}_i, \Gamma_{i, j}) \in M\Big\}\,, \\
                \Gamma^{(Z)} &= \Big\{\Gamma_{i, j} \mid (\mathcal{Z}_j, \Gamma_{i, j}) \in M\Big\}
            \end{aligned}
        \end{equation}
        as our corresponding stabilizer combination schedule. This is a valid schedule because we meet the following constraints.
        \begin{itemize}
            \item \emph{We cannot combine $X$ and $Z$ stabilizers in the same color group:} Since $M$ is a matching, it will not be the case that the specific color group $\Gamma_{i, j} \in V_R$ is matched with more than one vertex in $V_L$.
            \item \emph{We cannot combine $X$ ($Z$) stabilizers on color groups in the same row (column):} Since $M$ is a matching, each $\mathcal{X}_i$ ($\mathcal{Z}_j$) is matched with at most one $\Gamma_{i, \cdot}$ ($\Gamma_{\cdot, j}$), corresponding to one color group per row (column).
        \end{itemize}
        Thus $S(M)$ is a valid stabilizer combination schedule.
        
        \textbf{Optimality.} First we show that for \textit{any} matching $M$ and its corresponding schedule $S$, we show that the weight of that matching is the number of \Csectorname{} qubits removed in schedule $S$; i.e.
        \begin{equation}\label{eq:optimality}
            \operatorname{wt}M = \sum_S|\Gamma_{i, j}|\; .
        \end{equation}
        We can do this as follows. For any schedule $S = (\Gamma^{(X)}, \Gamma^{(Z)})$, define the total number of \Csectorname{} qubits removed $R(S)$ as
        \begin{equation}
            R(S) = \sum_{\Gamma_{i, j} \in \Gamma^{(X)}}|\Gamma_{i, j}| + \sum_{\Gamma_{i, j} \in \Gamma^{(Z)}}|\Gamma_{i, j}|\, .
        \end{equation}
        Let $M$ be a matching of $\mathpzc{G}$, and let $S(M)$ be the corresponding schedule. By construction of the edge weights, each $(\mathcal{X}_i, \Gamma_{i, j})$ or $(\mathcal{Z}_j, \Gamma_{i, j})$ in $M$ contributes $w_{i, j}$ to $\operatorname{wt}(M)$. But since $w_{i, j} = |\Gamma_{i, j}|$, this is the exact number of \Csectorname{} qubits we remove from combining in that group. Since $M$ is a matching, no color group is counted twice and so
        \begin{equation}
            R(S(M)) = \sum_S |\Gamma_{i, j}|= \sum_{M} w_{i, j} = \operatorname{wt} M\, .
        \end{equation} 
        This proves \eqref{eq:optimality}. The argument for optimality is now straightforward: since we have a maximum-weight matching, this must be an optimal stabilizer combination schedule.
    \end{proof}

    Now, by reducing the stabilizer combination problem to maximum weight matching on $\mathcal{G}$ and solving that optimally, we solve the stabilizer combination problem optimally.
\end{proof}



\begin{algorithm}[t]\label{alg:max_color_groups}
    \DontPrintSemicolon
    \caption{Choosing combination schedule by maximum-weight bipartite matching}
    \Input{The product-color groups $\Gamma=\{\Gamma_{i,j}\}_{i=1,\,j=1}^{i=\chi_1,\, j=\chi_2}$.}
    \Return{The sets $\Gamma^{(X)}$ and $\Gamma^{(Z)}$ of color groups where $X$-checks and $Z$-checks are combined, respectively.}

    Construct
    $V_L=\mathcal{X}\cup\mathcal{Z}
    =\{\mathcal{X}_i\}_{i=1}^{\chi_1}\cup\{\mathcal{Z}_j\}_{j=1}^{\chi_2}$\;
    Construct
    $V_R=\{\Gamma_{i,j}\}_{i=1,\,j=1}^{i=\chi_1,\,j=\chi_2}$\;

    Initialize a weighted adjacency matrix
    $W\in\mathbb{Z}_{\ge 0}^{(\chi_1+\chi_2)\times (\chi_1\chi_2)}$
    with all entries equal to $0$\;
    
    \For{$i=1$ \KwTo $\chi_1$}{
        \For{$j=1$ \KwTo $\chi_2$}{
            \If{$|\Gamma_{i,j}|>0$}{
                $W_{\mathcal{X}_i,\Gamma_{i,j}}\gets |\Gamma_{i,j}|$\;
                $W_{\mathcal{Z}_j,\Gamma_{i,j}}\gets |\Gamma_{i,j}|$\;
            }
        }
    }

    $M\gets \texttt{maximum\_weight\_bipartite\_matching}(W)$\;
    
    $\Gamma^{(X)}\gets \emptyset$\;
    $\Gamma^{(Z)}\gets \emptyset$\;

    \For{$i=1$ \KwTo $\chi_1$}{
        \If{there exists $x_i\in[\chi_2]$ such that $(\mathcal{X}_i,\Gamma_{i,x_i})\in M$}{
            $\Gamma^{(X)}\gets \Gamma^{(X)}\cup \{\Gamma_{i,x_i}\}$\;
        }
    }

    \For{$j=1$ \KwTo $\chi_2$}{
        \If{there exists $z_j\in[\chi_1]$ such that $(\mathcal{Z}_j,\Gamma_{z_j,j})\in M$}{
            $\Gamma^{(Z)}\gets \Gamma^{(Z)}\cup \{\Gamma_{z_j,j}\}$\;
        }
    }

    \Return{$\Gamma^{(X)},\Gamma^{(Z)}$}
\end{algorithm}

An example of the maximal qubit reduction problem is shown in Figure~\ref{hgpextboptimized}, and the corresponding graph $\mathpzc{G}$ on which we perform maximum-weight matching is shown in Figure~\ref{hgpoptreduction}.  
Additionally, the procedure is succinctly summarized in pseudocode in Algorithm \ref{alg:max_color_groups}.

\subsubsection{Cleaning stabilizers off \Csectorname{} qubits}

Now we discuss how the stabilizers are actually combined in each color group. For a given \Csectorname{} qubit, let $\Delta$ denote the number of $X$ (or $Z$) checks that it participates in. Suppose we choose to combine $X$ or $Z$ checks connected to that \Csectorname{} qubit. The relation between the new (combined) and old checks is represented by the linear transformation $w \in \F^{(\Delta-1)\times\Delta}_2$ with nonzero entries
\begin{align}\label{eq:1D rep PCM}
    w = \begin{pmatrix}
        1 & 1 &&&& \\
        &1&1&&& \\
        &&&\ddots&&
    \end{pmatrix} \, ,
\end{align}
which is equivalent to the parity-check matrix of the 1D repetition code of length $\Delta$. In other words, for each \Csectorname{} qubit that we want to remove, we have that the first new stabilizer is the combination of the first old stabilizer and the second old stabilizer, the second new stabilizer is the combination of the second old stabilizer and the third old stabilizer, etc. So for each putative \Csectorname{} qubit, the number of stabilizers per \Csectorname{} qubit decreases by one.

Because we only perform this reduction on disjoint checks, the entire reduction can be described by a matrix $W$ which, upon reordering its columns, is of block-diagonal form with either $w$ or $\ident$ as the blocks. In other words, we have $W \sim w \oplus w \oplus \ident \oplus \cdots$.
Since each local matrix $w$ has full (row) rank, so does the entire matrix $W$.

\begin{claim}[LDPC weights of reduced HGP] \label{claim:LDPC preservation}
    If the original HGP code was $(w_q,w_c)$-LDPC, then the reduced HGP code is $(\tilde{w}_q, \tilde{w}_c)$-LDPC with $\tilde{w}_q \leq 2w_q$ and $\tilde{w}_c \leq \max\{w_c,2(w_c-1)\}$.
\end{claim}

\begin{proof}
    Let $\abs{A}_{\rm r}$ and $\abs{A}_{\rm c}$ denote the maximum row and column weight respectively of a binary matrix $A$. Since the entries are non-negative, we have $\abs{A}_{\rm r} = \norm{A}_\infty$ and $\abs{A}_{\rm c} = \norm{A}_1$, where $\norm{\cdot}_p$ denotes the vector-induced $p$-norm and we have implicitly embedded its argument over the real numbers. Submultiplicativity of $\norm{\cdot}_p$ then implies that $\abs{AB}_{\rm r} \leq \abs{A}_{\rm r} \abs{B}_{\rm r}$ and $\abs{AB}_{\rm c} \leq \abs{A}_{\rm c} \abs{B}_{\rm c}$.

    We first show that $\tilde{w}_c \leq \max\{w_c,2(w_c-1)\}$. Consider any row $\tilde{h}$ of $\widetilde{H}_X$. If $\tilde{h}$ is an uncombined check, puncturing cannot increase its weight, so $|\tilde{h}|\leq w_c$. Otherwise, $\tilde{h}=(h_1+h_2)V$ for two original checks $h_1,h_2$ incident on the \Csectorname{} qubit being removed. Since that qubit is removed by $V$, we have $|h_1V|,|h_2V|\leq w_c-1$, and therefore
$
    |\tilde{h}| \leq |h_1V|+|h_2V| \leq 2(w_c-1).
$
Thus,
$
    \tilde{w}_c \leq \max\{w_c,2(w_c-1)\}.
$

    Now we show that $\tilde{w}_q \leq 2w_q$. By the same argument as for the row weight, we have $\abs{W_X}_{\rm c}\leq2$. Since $V$ is a restriction, $\abs{V}_{\rm c} = 1$. Thus, $\tilde{w}_q = \abs{W_XH_XV}_{\rm c} \leq \abs{W_X}_{\rm c} \abs{H_X}_{\rm c} \abs{V}_{\rm c} \leq 2\abs{H_X}_{\rm c} = 2w_q$. The same arguments hold using $\widetilde{H}_Z$.
\end{proof}

\subsubsection{Qubit-count and check-weight tradeoff}\label{sec:weight tradeoff}

The maximum-weight matching above optimizes the number of removed qubits. If a smaller check weight is preferred, one can instead impose a cap $w_{\max} \ge w_c$. Each matching edge is oriented: $(\mathcal X_i,\Gamma_{i,j})$ selects $X$-cleaning, whereas $(\mathcal Z_j,\Gamma_{i,j})$ selects $Z$-cleaning. For each oriented edge, we choose the local ordering in~\eqref{eq:1D rep PCM} that minimizes the largest weight of the adjacent pairwise sums in the corresponding Pauli family, and discard that edge if any resulting check exceeds $w_{\max}$. The unchanged checks must also satisfy the cap. Since the degree $\Delta$ is bounded for an LDPC family, the local ordering search has constant cost. A maximum-weight matching on the remaining oriented edges then gives the largest reduction among schedules that satisfy the cap.

More generally, any subset of a valid schedule remains valid. One may therefore restore selected \Csectorname{} qubits and their original checks. The cap controls the maximum resulting check weight, while the matching weight $|\Gamma_{i,j}|$ can be replaced by an additive objective that also penalizes the sum of the resulting check weights. This gives a direct tradeoff between spatial overhead and syndrome-extraction cost.

\subsection{Parameter preservation} \label{sec:param preservation}

Taking products of stabilizers preserves commutation before puncturing, but we must also check that deleting coordinates preserves it. Let $R$ denote the removed coordinates. Since $V$ selects the complementary coordinates, $VV^\transp=\ident_n+\sum_{q\in R}\mathbf e_q\mathbf e_q^\transp$ over $\F_2$. Therefore,
{
\begin{align}
\widetilde H_X\widetilde H_Z^\transp
&=W_XH_XVV^\transp H_Z^\transp W_Z^\transp \notag\\
&=\sum_{q\in R}(W_XH_X\mathbf e_q)(W_ZH_Z\mathbf e_q)^\transp=0\, . \label{eq:reduced commutation}
\end{align}
}
The term containing $H_XH_Z^\transp$ vanishes by the original CSS condition. For each $q\in R$, the construction cleans either every transformed $X$-check or every transformed $Z$-check from $q$, so one factor in the corresponding outer product is zero. Thus, the reduced matrices define a valid CSS code.

Let $\llbracket \tilde{n}, \tilde{k}, \tilde{d} \rrbracket$ be the code parameters of the reduced HGP code $\widetilde{\mathcal{Q}}$, and ${\widetilde{H}_X} \in \F_2^{\tilde{m}_X \times \tilde{n}}$ and ${\widetilde{H}_Z} \in \F_2^{\tilde{m}_Z \times \tilde{n}}$ the reduced $X$ and $Z$ parity-check matrices. As mentioned at the end of Section~\ref{subsec:HGP review}, we will utilize the following assumption:
\begin{quote}
    \textbf{Assumption:} The classical input parity-check matrices $H_1$ and $H_2$ have full row rank. It follows from \eqref{eq:HGP H_X,H_Z} that $H_X$ and $H_Z$ for the HGP code also have full row rank and so $k_1^\transp k_2^\transp = 0$; i.e. there are no \Csectorname{} logical qubits.
\end{quote}

We first show that the code dimension, or number of logical qubits, is unchanged by the reduction.

\begin{thm}[Logical dimension preservation]\label{thm: k preservation}
    $\tilde{k}=k$. In other words, the number of \Bsectorname{} logical qubits stays the same after applying the qubit reduction procedure in Section~\ref{sec:procedure}.
\end{thm}

\begin{proof}
    Let $n_{\rm rem} := n-\tilde{n}$ be the number of removed \Csectorname{} qubits from the reduction procedure. By construction, we also have $(m_X+m_Z) - (\tilde{m}_X + \tilde{m}_Z) = n_{\rm rem}$ since we are removing one stabilizer per removed qubit. The restriction $V$ retains the entire bit-type sector. Thus, $H_XV$ contains the block $H_1\otimes\ident_{n_2}$ and $H_ZV$ contains the block $\ident_{n_1}\otimes H_2$. These blocks have full row rank under our assumption, and hence so do $H_XV$ and $H_ZV$. The matrices $W_X$ and $W_Z$ also have full row rank. Indeed, $\mbf uWA=0$ implies $\mbf uW=0$ and then $\mbf u=0$ whenever $A$ and $W$ have full row rank. It follows that $\rank W_XH_XV=\tilde m_X$ and $\rank W_ZH_ZV=\tilde m_Z$.

    In all, this gives 
    \begin{align}
        \tilde{k} &= \tilde{n} - \rank({W_XH_XV}) - \rank({W_ZH_ZV}) \notag \\ 
        &= n - n_{\rm rem} - \tilde{m}_X - \tilde{m}_Z \notag \\ 
        &= n - n_{\rm rem} + n_{\rm rem} - m_X - m_Z \notag \\ 
        &= n - \rank H_X - \rank H_Z \notag \\ 
        &= k. \label{k'}
    \end{align}
\end{proof}

Recall that the original logical operators are defined by the rows of $G_X^\mathscr{B}$ and $G_Z^\mathscr{B}$ given in \eqref{eq:HGP G_X,G_Z B}. 
Let $\tilde{Q}$ be the sector of qubits that are preserved after the reduction, and let $H_X|_{\tilde{Q}}$ and  $H_Z|_{\tilde{Q}}$ represent the restriction of $H_X$ and $H_Z$ to the qubits in $\tilde{Q}$, so that the number of columns match those of ${\widetilde{H}_X}$ and ${\widetilde{H}_Z}$. Likewise, let $G_X^\mathscr{B}|_{\tilde{Q}}$ and $G_Z^\mathscr{B}|_{\tilde{Q}}$ represent the restriction of $G_X^\mathscr{B}$ and $G_Z^\mathscr{B}$ to the qubits in $\tilde{Q}$.

We first derive a useful relation between the restricted stabilizer subspaces of the original and reduced HGP codes. Recall that $\widetilde{H}_X = W_XH_XV$ and $\widetilde{H}_Z = W_ZH_ZV$ from \eqref{eq:tildeH,H relation}. Since $V$ is identity on the \Bsectorname{} qubits, $\widetilde{H}_X|_\mathscr{B} = W_XH_X|_\mathscr{B}$ and $\widetilde{H}_Z|_\mathscr{B} = W_ZH_Z|_\mathscr{B}$. It follows that
\begin{subequations}\label{eq:rs containment}
\begin{align}
    \rs{\widetilde{H}_X}|_\mathscr{B} &\subset \rs{H_X}|_\mathscr{B}  \label{eq:X rs containment} \\
    \rs{\widetilde{H}_Z}|_\mathscr{B} &\subset \rs{H_Z}|_\mathscr{B} \, . \label{eq:Z rs containment}
\end{align}
\end{subequations}

We now show that the original HGP logical basis \eqref{eq:HGP G_X,G_Z B} is still a valid logical basis for the reduced HGP code upon restricting to the qubits in $\tilde{Q}$.

\begin{thm}[Canonical logical basis preservation] \label{thm:basis preservation}
    $\widetilde{G}^\mathscr{B}_X := G_X^\mathscr{B}|_{\tilde{Q}}$ and $\widetilde{G}^\mathscr{B}_Z := G_Z^\mathscr{B}|_{\tilde{Q}}$ form a valid logical Pauli basis for the \Bsectorname{} qubits in the reduced HGP code with parity-check matrices $\widetilde{H}_X$ and $\widetilde{H}_Z$.
\end{thm}

\begin{proof}
    First, we note that the restrictions that define $\widetilde{G}^\mathscr{B}_X$ and $\widetilde{G}^\mathscr{B}_Z$ only remove coordinates in the \Csectorname{} sector where the entries are all zero, and so the $k$-qubit Pauli algebra $\widetilde{G}^\mathscr{B}_X (\widetilde{G}^\mathscr{B}_Z)^\transp = {G}^\mathscr{B}_X ({G}^\mathscr{B}_Z)^\transp = \ident_k = \ident_{\tilde k}$ is still satisfied.
    
    We will now show that the original $\bar{X}^\mathscr{B}$ logicals are still valid under the transformed codes by showing that both
    \begin{enumerate}
        \item rows of $\widetilde{G}^\mathscr{B}_X$ are contained in $\ker{\widetilde{H}_Z}$
        \item rows of $\widetilde{G}^\mathscr{B}_X$ are not contained in $\rs{\widetilde{H}_X}$
    \end{enumerate}
    and claim that the same holds for $\bar{Z}^\mathscr{B}$ \Bsectorname{} logicals.
    We have that
    \begin{enumerate}
        \item Since $\widetilde{G}_X^\mathscr{B}=G_X^\mathscr{B}V$ and the removed coordinates have zero support in $G_X^\mathscr{B}$, we have $VV^\transp(G_X^\mathscr{B})^\transp=(G_X^\mathscr{B})^\transp$. Therefore,
        \begin{align}
            \widetilde{H}_Z(\widetilde{G}_X^\mathscr{B})^\transp
            &=W_ZH_ZVV^\transp(G_X^\mathscr{B})^\transp \notag\\
            &=W_ZH_Z(G_X^\mathscr{B})^\transp=0,
        \end{align}
        so the rows of $\widetilde{G}_X^\mathscr{B}$ are contained in $\ker{\widetilde{H}_Z}$.
        \item Each row of $\widetilde G_X^\mathscr{B}$ anticommutes with its paired row of $\widetilde G_Z^\mathscr{B}$. If it belonged to $\rs\widetilde H_X$, it would be an $X$-stabilizer and would commute with every $Z$ logical, which is a contradiction. Thus, no row of $\widetilde G_X^\mathscr{B}$ is a stabilizer.
    \end{enumerate}
    The same argument holds for $\widetilde{G}_Z^\mathscr{B}$ with $X$ and $Z$ exchanged. Since the $k=\tilde{k}$ logical pairs have the identity pairing above, they form a complete logical Pauli basis for the reduced code.
\end{proof}

We now show that the minimum distances of the reduced HGP code are the same as those of the original HGP code, making use of both the logical sector weights (Proposition~\ref{prop:HGP logical sector weights}) and the preserved canonical logical basis (Theorem \ref{thm:basis preservation}).

\begin{thm}[Minimum distance preservation]\label{thm:distance preservation}
    $\tilde{d}_X=d_X$ and $\tilde{d}_Z=d_Z$. In other words, the minimum $X$ and $Z$ distances stay the same after applying the qubit reduction procedure in Section \ref{sec:procedure}.
\end{thm}

\begin{proof}
    Let $(d_X, d_Z)$ and $(\tilde{d}_X,\tilde{d}_Z)$ denote the minimum distances of the original and reduced HGP codes respectively. We will prove $\tilde{d}_X=d_X$, and the argument for $\tilde{d}_Z=d_Z$ is nearly identical. 

    Under the full-row-rank assumption, $d_1^\transp$ is undefined/infinity, so~\eqref{eq:HGP d_X} gives $d_X=d_2$. By Theorem~\ref{thm:basis preservation}, every nontrivial reduced $X$-logical class is represented by a nonzero linear combination $\mbf{x}$ of the rows of $G_X^\mathscr{B}$. Any representative of that class has the form $\mbf{x}+\tilde{\mbf{s}}_X$ with $\tilde{\mbf{s}}_X\in\rs\widetilde H_X$. Write $\tilde{\mbf{s}}_X=\mbf{u}W_XH_XV$ and lift it to the original stabilizer $\mbf{s}_X=\mbf{u}W_XH_X$. Since $V$ leaves the bit-type sector unchanged, $(\mbf{x}+\tilde{\mbf{s}}_X)^\mathscr{B}=(\mbf{x}+\mbf{s}_X)^\mathscr{B}$. This lift is a representative of the same nontrivial \Bsectorname{} logical class in the original HGP code. Proposition~\ref{prop:HGP logical sector weights} therefore gives $\abs{(\mbf{x}+\tilde{\mbf{s}}_X)^\mathscr{B}}_{\rm c}\geq d_2$, so every reduced representative has weight at least $d_2$. Hence, $\tilde d_X\geq d_2$.

    For the reverse inequality, choose a minimum-weight codeword $\mbf{c}\in\ker H_2\setminus\{\mathbf 0\}$ and a unit vector $\mbf{e}_i\notin\rs H_1$. Then $(\mbf{e}_i\otimes\mbf{c}\mid\mathbf 0)$ represents a \Bsectorname{} logical $\bar X$ of weight $d_2$. The reduction removes only \Csectorname{} coordinates, so the same representative has weight $d_2$ in the reduced code. Thus, $\tilde d_X=d_2=d_X$. Exchanging the two inputs and $X$ with $Z$ gives $\tilde d_Z=d_1=d_Z$, and hence $\tilde d=d$.
\end{proof}

\subsection{Special case: Quantum Tanner transformation}\label{qTanner transformation}

\begin{figure}[t]
  \centering
  \includegraphics[width=0.4\textwidth]{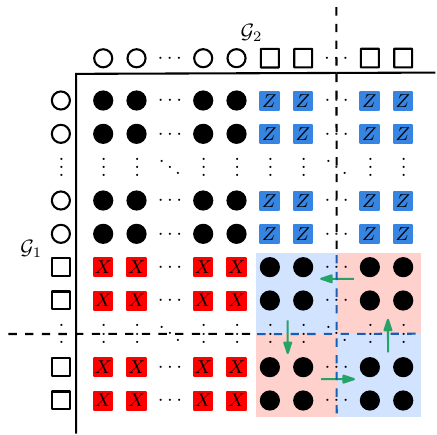}
  \caption{When the checks of the classical codes can be 2-colored, the \Csectorname{} qubits are sectored into four quadrants that do not share any stabilizer support. Then, after combining $Z$ stabilizers in blue quadrants (quadrants I and IV) and $X$ stabilizers in red quadrants (quadrants II and III), their supports are moved in the direction of the green arrows. One can then remove all \Csectorname{} qubits.}
  \label{fig:trans}
\end{figure}


We now describe how the qubit reduction procedure of Section \ref{sec:procedure} encompasses the quantum Tanner transformation~\cite{Leverrier_2025_efficient} for HGP codes, in the special case where the classical input codes are built from simple bipartite graphs.

Suppose we can define the classical input codes $\mathcal{C}_1$, $\mathcal{C}_2$ as Tanner codes $\mathcal{T}_1\big(\mathpzc{G}_1,\hat{\mathcal{C}}_1\big)$, $\mathcal{T}_2\big(\mathpzc{G}_2,\hat{\mathcal{C}}_2\big)$ (see Appendix~\ref{app:Tanner codes}) on $\Delta_1$-regular, $\Delta_2$-regular, bipartite base graphs $\mathpzc{G}_1 = (V^{(1)}_\ell \cup V^{(1)}_r, E_1)$, $\mathpzc{G}_2 = (V^{(2)}_\ell \cup V^{(2)}_r, E_2)$ where edges do not exist between vertices within $V^{(i)}_\ell$ or $V^{(i)}_r$ ($    i= 1, 2$). Note that one can create a $\Delta$-regular, bipartite graph by taking the bipartite double cover an arbitrary $\Delta$-regular graph (see Appendix~\ref{app:double cover}). For each Tanner code, each vertex of the bipartite graph hosts the parity-checks of the local code $\hat{\mathcal{C}}$, and each edge is a physical bit. Now, rather than directly coloring the parity checks, we color each vertex cluster of parity checks. Since the underlying graph is bipartite, we only need two colors. The HGP code now has a similar coloration structure to Figure \ref{fig:trans}, but each \Csectorname{} ``qubit'' is now a cluster of qubits corresponding to the cluster of checks in the classical Tanner codes.


Additionally, defining the classical codes in this way allows us to easily understand how we can combine stabilizers. Let $\hat{g}_1$, $\hat{g}_2$ and $\hat{h}_1, \hat{h}_2$ be the generator and parity-check matrices of the local vertex codes which satisfy
\begin{align}
    \hat{h}_1\hat{g}_1^\transpose = \hat{h}_2\hat{g}_2^\transpose = 0 \, . \label{localtrans}
\end{align}
In the transpose Tanner codes $\mathcal{T}_1^\transpose\big(\mathpzc{G}_1,\hat{\mathcal{C}}_1\big)$, $\mathcal{T}_2^\transpose\big(\mathpzc{G}_2,\hat{\mathcal{C}}_2\big)$, where the roles of checks and bits are swapped, local check support looks like $\hat{h}_1^\transpose$, $\hat{h}_2^\transpose$, whose structure is mimicked on the \Csectorname{} qubits of the HGP code. Noting that \eqref{localtrans} is equivalent to $\hat{g}_1\hat{h}_1^\transpose = \hat{g}_2\hat{h}_2^\transpose = 0$ after transposing both sides, we can combine $Z$ stabilizers according to the rows of $\hat{g}_1$ and combine $X$ stabilizers according to the rows of $\hat{g}_2$, which will shift support on \Csectorname{} qubits like in Figure~\ref{fig:trans}. When each vertex only hosts a single parity check, i.e. the local code is a single parity-check code, then the Tanner code is simply a cycle code on the bipartite graph, and the local transformation \eqref{localtrans} reduces to the repetition-code transformation \eqref{eq:1D rep PCM} of the general procedure in Section \ref{sec:procedure}.

Since \eqref{localtrans} is a linear transformation, the entire quantum Tanner transformation is still described in terms of matrices \eqref{eq:tildeH,H relation}. As such, all of the parameter preservation results in Section \ref{sec:param preservation} follow through as well.


\section{Compatibility with existing fault-tolerant gadgets} \label{sec:gadget compatibility}

In this section, we show how the HGP physical qubit reduction procedure of Section~\ref{sec:procedure} is compatible with several existing fault-tolerant gadgets for HGP codes. In subsection~\ref{sec:hook errors}, we show how the reduction procedure admits distance-preserving stabilizer measurement schedules that mitigate hook errors. In subsection~\ref{sec:fold-transversal}, we show how the reduction procedure can be modified to preserve fold-transversal gates~\cite{Quintavalle_2023_fold, Quintavalle_2023_fold}. In subsection~\ref{sec:automorphisms}, we show how the reduction procedure can be adapted to HGP automorphism gates~\cite{SHYPS_codes, HGP_aut_gadgets}. In subsection~\ref{sec:GPPM}, we show that the reduction procedure can be adapted to HGP code homomorphisms which allow homomorphic grid-Pauli-product measurement (GPPM) gadgets~\cite{Xu_2025_fast}.

The fault-tolerant gadgets, in the order of most compatible to least compatible, are the distance-preserving syndrome extraction and homomorphic gadgets, followed by the fold-transversal and lastly the automorphism gates. The first two are directly compatible with any reduction scheduling, whereas the last two restrict the possible schedules and so are weaker in their compatibilities.

\subsection{Distance-preserving syndrome extraction}\label{sec:hook errors}


The error syndrome of a stabilizer code is the measurement outcomes of all its stabilizer checks. For a qLDPC code with low-weight stabilizers, it is customary to perform single-ancilla syndrome extraction: introduce an ancillary ``syndrome'' qubit for each stabilizer check, couple each syndrome qubit to its corresponding data qubits, and finally measure all syndrome qubits. For a stabilizer check with weight $w$, the middle coupling step involves $w$ CNOT gates. From the definition of the code distance, we know that we need at least $d$ single-qubit errors on the data qubits to be undetectable by the syndrome. If we think about repeated syndrome measurements in the context of a quantum memory, then potential faults in one round of syndrome extraction can lead to additional data errors in the next round. If we also consider errors on the syndrome qubits, then this radius of detection can decrease, since a single error on a syndrome qubit during the syndrome extraction circuit can propagate to $\lfloor w/2\rfloor$ data ``hook'' errors under stabilizer equivalence; see Figure \ref{fig:logi_hooks}(a) for an illustration with an $X$-check measurement. So in the worst case, an undetectable error (during the next QEC cycle) may be caused by only $\lceil d/\lfloor w/2\rfloor \rceil$ elementary hook errors if they happen to align in the direction of a minimum-weight logical operator. This observation motivates the following notion of \emph{effective} distance.

\begin{defn}[Effective distance~\cite{Manes_2025}]\label{def:d_eff}
    For a $\llbracket n,k,d \rrbracket$ stabilizer code with single-ancilla syndrome-extraction circuit $U_{\rm SE}$, the effective distance $d_{\rm eff}(U_{\rm SE})$ is defined as the minimum number of single-qubit errors, on both data and syndrome qubits at any point in $U_{\rm SE}$, such that the ensuing data error following $U_{\rm SE}$ is a nontrivial logical operator (modulo stabilizers). If $d_{\rm eff}(U_{\rm SE})=d$, then we say that $U_{\rm SE}$ is distance-preserving.
\end{defn}

Note that the effective distance is different from the circuit-level fault distance since the latter also includes undetectable timelike (measurement) errors.
The freedom in constructing a single-ancilla syndrome-extraction circuit $U_{\rm SE}$ lies in the ordering or scheduling of the CNOT gates. For a weight-$w$ stabilizer check, there are $w!$ different ways to schedule the CNOTs, which will affect the potential patterns for hook errors. It has been previously shown that any CNOT-scheduling of $U_{\rm SE}$ for 2D~\cite{Manes_2025} and higher-dimensional~\cite{Tan_2025_eff} HGP codes are distance-preserving. The argument in 2D follows straightforwardly from the logical sector weights of Proposition \ref{prop:HGP logical sector weights}: since each $X$-check has support on at most one column (row) in the \Bsectorname{} (\Csectorname{}) sector, any single ancillary hook error is only supported within a single column (row). At the same time, each logical $\bar{X}$ operator is supported on at least $d_X$ columns (rows). Hence, at least $d_X$ errors on both data and syndrome qubits are required to effect a logical operation and be undetectable by the next round of QEC.

In the reduced HGP code, the combined stabilizer checks will have support on two rows or columns on the \Bsectorname{} qubits. So a generic CNOT schedule for these checks may result in hook errors that have support on more than one row or column, leading to a possible decrease in the effective distance. For example, in the rotated surface code, a ``bad'' syndrome measurement schedule leads to a 50\% decrease in the effective distance. To maintain a distance-preserving syndrome extraction, the CNOTs are zigzagged so that hook errors lie in an orthogonal direction to the relevant logical operators~\cite{Tomita_2014}. The solution for the rotated surface code can be straightforwardly generalized to our reduced HGP codes. The idea is that for each combined $Z$-check ($X$-check), we first apply CNOTs to one entire row (column) before proceeding to the other row (column); see Algorithm \ref{alg:split SE} for the detailed procedure for $X$-checks; for $Z$-checks, the ancilla is initialized and measured in the $Z$-basis, and the directions of the CNOTs are reversed. In this manner, any single $Z$ ($X$) hook error is restricted to just one row (column).

\begin{algorithm}[t]\label{alg:split SE}
    \DontPrintSemicolon
    \caption{Split $X$-syndrome extraction}
    \Input{The reduced code's $X$-Tanner graph $T_X=(Q\cup C_X,E_X)$; the subset $C_X^{\mathrm{comb}}\subseteq C_X$ of combined $X$-checks; for each $c\in C_X^{\mathrm{comb}}$, a partition of its incident edges $E(c)=E_1(c)\sqcup E_2(c) \sqcup E_{\rm rest}(c)$ into the two \Bsectorname{} data-qubit columns and the rest.}
    \Return{The measurement outcome of all $X$-check generators in $C_X$.}
    $E_{X,1}\gets \bigcup_{c\in C_X\setminus C_X^{\mathrm{comb}}} E(c) \,\cup\, \bigcup_{c\in C_X^{\mathrm{comb}}} E_1(c)$\;
    $E_{X,2}\gets \bigcup_{c\in C_X^{\mathrm{comb}}} E_{\rm rest}(c)$\;
    $E_{X,3}\gets \bigcup_{c\in C_X^{\mathrm{comb}}} E_2(c)$\;
    Construct subgraphs $T_{X,1}=(Q\cup C_X,E_{X,1})$, $T_{X,2}=(Q\cup C_X^{\mathrm{comb}},E_{X,2})$, and $T_{X,3}=(Q\cup C_X^{\mathrm{comb}},E_{X,3})$\;
    
    Compute minimum edge colorings $\chi^{}_{1}, \chi^{}_{2}, \chi^{}_{3}$ of $T_{X,1}, T_{X,2}, T_{X,3}$ respectively.\;
    
    Prepare an ancilla in $\ket{+}$ for each check $c\in C_X$.\;
    \For{$\alpha\in\{1,2,3\}$}{
    \For{color $\xi\in \chi^{}_{\alpha}$}{
      Simultaneously apply all gates $\mathrm{CNOT}_{i\rightarrow j}$ from the $i$th ancilla
      to the $j$th data qubit supported on an edge $\{i,j\}\in E_{X,\alpha}$ with color $\xi$.\;
    }}
    Measure each ancilla in the $X$ basis.\;
\end{algorithm}

We briefly comment on the syndrome extraction depth. For $P\in\{X,Z\}$, let $T_{P,\alpha}$ where $\alpha = 1, 2, 3$ be the three bipartite subgraphs in the split schedule. By K\"{o}nig's line-coloring theorem, a minimum edge coloring of each subgraph has $\Delta(T_{P,\alpha})$ colors. The CNOT depth used by Algorithm~\ref{alg:split SE} is therefore
\begin{align}
    D_P^{\rm split}=\sum_{\alpha=1}^3\Delta(T_{P,\alpha})\, . \label{eq:split depth}
\end{align}
If $\Delta(T_P)$ is the depth of an unconstrained edge coloring of the full $P$-Tanner graph, then
\begin{align}
    \Delta(T_P)\leq D_P^{\rm split}\leq3\Delta(T_P)\, . \label{eq:split depth bounds}
\end{align}
Each of the three blocks is depth-optimal, but their concatenation need not be globally optimal, since layers from different blocks can be interleaved when the order $E_{P,1}\prec E_{P,2}\prec E_{P,3}$ is retained at every combined-check ancilla. Finding the shortest such schedule is generically a more complex problem, which we leave to future work.

\begin{figure}[t]
  \centerline{\includegraphics[width=0.4\textwidth]{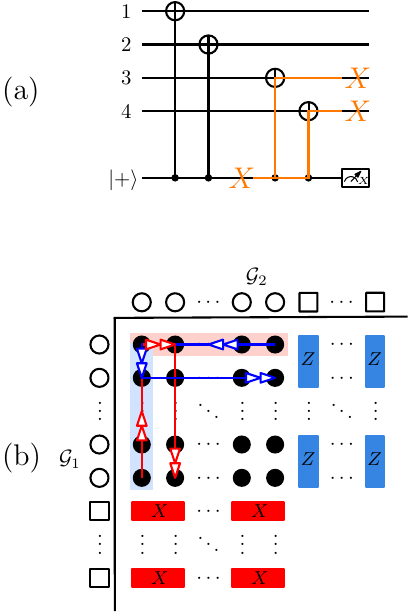}}
  \caption{\textbf{(a)} A CNOT measurement schedule of an $XXXX$ stabilizer check. A single ancilla $X$ error in the middle of the circuit propagates to two data qubit $X$ ``hook'' errors. If a minimum-weight logical operator has support on those two data qubits, then the code's effective distance decreases, as we can't reduce this logical error by stabilizer equivalence. \textbf{(b)} A \Bsectorname{} logical $\bar{X}^{\mathscr{B}}$ operator will only have support on the row in red; a $\bar{Z}^{\mathscr{B}}$ logical operator will only have support on the column in blue. Possible distance-preserving measurement schedules for a combined $X$-check that spans columns 1 and 2 (red arrows) and for a combined $Z$-check that spans rows 1 and 2 (blue arrows). This schedule will ensure that propagated errors for an $X$ ($Z$) stabilizer stay contained within one column (row), up to stabilizer equivalence.}
  \label{fig:logi_hooks}
\end{figure}

\begin{thm}[Effective-distance preservation under the reduction] \label{thm:d_circ preservation}
    The reduced HGP code after the reduction in Section~\ref{sec:procedure} admits a distance-preserving stabilizer measurement schedule given by Algorithm \ref{alg:split SE} and its $Z$-version.
\end{thm}

\begin{proof} 
    First, recall from Proposition \ref{prop:HGP logical sector weights} that \Bsectorname{} logical $\bar{X}$ ($\bar{Z}$) operators will have column (row) weight at least $d_2$ ($d_1$). For a combined $X$-check $c$, let $E(c) = E_1(c) \sqcup E_2(c) \sqcup E_{\rm rest}(c)$ be the partition of checks according to the input of Algorithm \ref{alg:split SE}, and let $B_1(c), B_2(c), B_{\rm rest}(c)$ be their data-qubit supports respectively. Recall that $B_1(c)$ and $B_2(c)$ are individually supported in single columns of the \Bsectorname{} sector, and $B_{\rm rest}(c)$ is in a row of the \Csectorname{} sector. By definition of the HGP checks \eqref{eq:HGP H_X,H_Z}, we have $\abs{B_1(c)}=\abs{B_2(c)}$. The CNOTs are scheduled in the order of $E_1(c)$, then $E_{\rm rest}(c)$, and then finally $E_2(c)$. An ancilla fault propagates to a suffix of this ordered support. Multiplication by the check being measured replaces that suffix by its complementary prefix. If the fault occurs during $E_1(c)$, choose the prefix representative, whose \Bsectorname{} support lies in $B_1(c)$. If it occurs during $E_{\rm rest}(c)$ or $E_2(c)$, choose the suffix representative, whose \Bsectorname{} support lies in $B_2(c)$. In every case, the \Csectorname{} support lies in one row. Thus, each ancilla fault, and also each single-qubit data fault, has a representative supported on at most one \Bsectorname{} column and one \Csectorname{} row.
    A combination of $t$ such faults therefore has a representative supported on at most $t$ \Bsectorname{} columns. Suppose that the resulting $X$ error is a nontrivial logical operator of the reduced code. Since the canonical \Bsectorname{} basis spans its logical operators, this error can be written as $(\mbf aG_X^{\mathscr B}+\mbf uW_XH_X)V$, with $\mbf a\neq\mathbf 0$. The vector $\mbf aG_X^{\mathscr B}+\mbf uW_XH_X$ is a representative of the same nontrivial logical operator in the original HGP code, and its \Bsectorname{} restriction is unchanged by $V$. Proposition~\ref{prop:HGP logical sector weights} then gives $t\geq d_2$. The analogous argument for $Z$ errors gives $t\geq d_1$.
    Any nontrivial Pauli logical operator has a nontrivial $X$ or $Z$ component, so one of these two lower bounds applies.
    Conversely, a minimum-weight canonical logical operator can be applied by placing one data fault on each qubit of its support after the last CNOT layer. Thus the effective distance is at most $\min\{d_1,d_2\}=d$. Together with the preceding lower bounds, this proves $d_{\rm eff}=d$.
\end{proof}

\subsection{Fold-transversal gates}\label{sec:fold-transversal}

When the classical input parity-check matrices to the HGP are the same, the corresponding HGP code exhibits a $ZX$-duality~\cite{Breuckmann_2024_ZX} given by a reflection along the ``main diagonal''~\cite{Quintavalle_2023_fold}. In other words, the $X$-type and $Z$-type stabilizer checks map onto one another across the diagonal fold. This fold symmetry gives rise to a Hadamard-type and a phase-type fold-transversal gate, which induces logical Hadamard-SWAP and logical CZ-S gates respectively~\cite{Quintavalle_2023_fold}.

\begin{figure*}
    \centering
    \includegraphics[width=.75\linewidth]{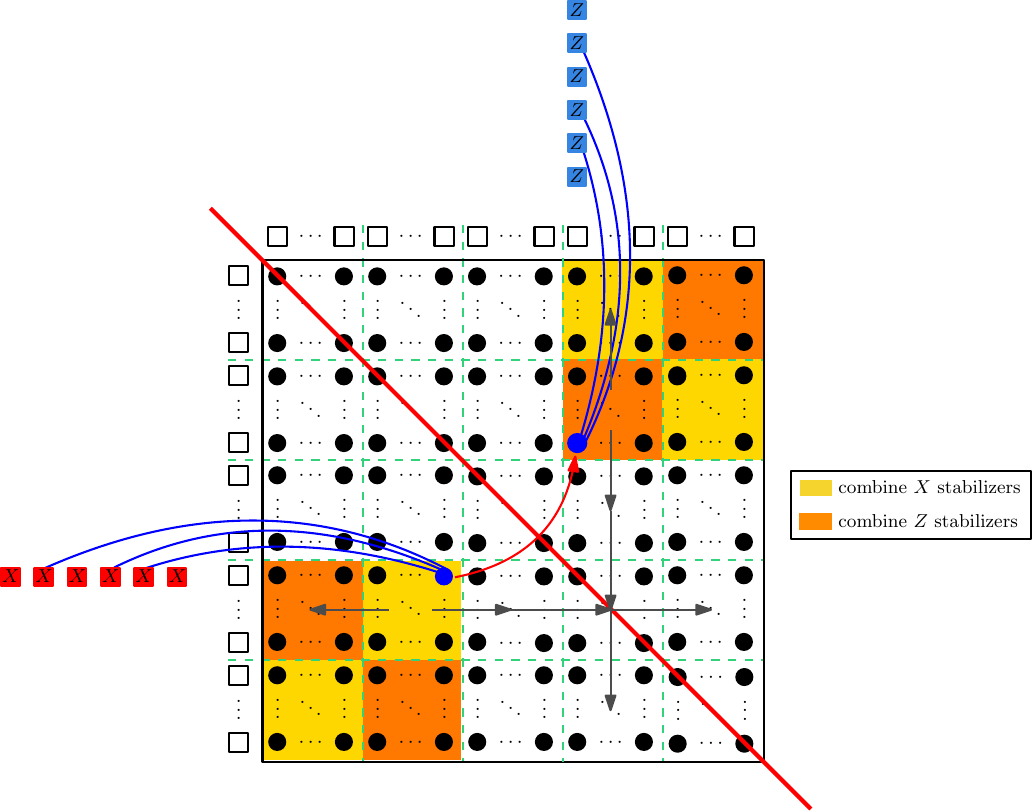}
     \caption{Example \Csectorname{} qubits with 25 color groups showing the main diagonal upon which we are ``fold-symmetric'' highlighted in red. ``Mirror'' qubits are those upon the reflection of this diagonal. (An example of two corresponding mirror qubits are in blue. It is also shown that the $X$ ($Z$) support of one qubit is the same as the $Z$ ($X$) support of its mirror qubit, since the HGP code is symmetric.) To enact a logical $\bar{H}$ + logical $\overline{\text{SWAP}}$, we first apply physical $H$ on every physical qubit then do physical SWAPs between each pair of mirror qubits. To enact a logical $\overline{\text{CZ}}$ + logical $\bar{S}$, we first apply physical $S$ on the qubits on the left of the diagonal, physical $S^\dag$ on the qubits on the right of the diagonal, and then physical CZ between mirror qubits. The baseline way to combine stabilizers to preserve fold symmetry is also shown, where we combine along the off-diagonal and alternate so that color group $\Gamma_{i, j}$ gets $X$ ($Z$) stabilizers combined and color group $\Gamma_{j, i}$ gets $Z$ ($X$) stabilizers combined, for some $i, j\in 1, \dots, 5,\ i\neq j$. This will push support in the direction of the outgoing gray arrows from the example color groups, which can be seen to push in symmetric directions such that the $ZX$-duality is preserved.}
    \label{fig:foldtransversal}
\end{figure*}

Specifically, the reason as to why these transversal gates arise is that the $X$ ($Z$) support of one qubit is exactly the same as the $Z$ ($X$) support of its mirror qubit across the diagonal---shown in Figure~\ref{fig:foldtransversal}---which is a consequence of using the same input parity-check matrix in the HGP. The qubits directly along the main diagonal do not have mirror qubits, and so we should not combine on any color groups on the main diagonal (those of the form $\Gamma_{i, i}$ for $i \in [\chi]$) as that would ruin the symmetry. Hence, the naive optimization that we present in Section \ref{sec:procedure}, which may combine along the main diagonal, would generically destroy this fold-symmetry.

We can, however, preserve fold-symmetry if we respect it in our reduction schedule.

\begin{prop}[Fold-symmetry preservation] \label{prop:fold preservation}
    Let $\mathcal{Q}=\mathrm{HGP}(H,H)$ be a symmetric HGP code, and let $\Gamma=\{\Gamma_{i,j}\}_{i, j \in [\chi]}$ be the product coloring of the \Csectorname{} qubits. Then any reduction protocol with the following properties:
    \begin{enumerate}
        \item No diagonal color group $\Gamma_{i,i}$ is used for stabilizer combination.
        \item Whenever $X$-checks are combined on some off-diagonal color group
        $\Gamma_{i,j}$, $Z$-checks are combined on its mirror color group $\Gamma_{j,i}$.
    \end{enumerate}
    retains the same fold-symmetry as $\mathcal{Q}$. Consequently, the fold-transversal Hadamard-type and phase-type logical gates of the symmetric HGP code are preserved under the reduction.
\end{prop}

\begin{proof}
    Let $B$ and $C$ denote the bit and check nodes of the Tanner graph of $H$.
    Since $\mathcal{Q}=\mathrm{HGP}(H,H)$ is symmetric, there is a diagonal reflection of qubits
    \begin{align}
    \rho_Q:(B\times B)\cup(C\times C)\to(B\times B)\cup(C\times C)
    \end{align}
    defined by
    \begin{align}
    \rho_Q(b_1,b_2)=(b_2,b_1) \quad,\quad
    \rho_Q(c_1,c_2)=(c_2,c_1)
    \end{align}
    with $b_i \in B$ and $c_i \in C$ for $i = 1, 2$. (i.e., $\rho_Q$ will swap two bits and checks.) There is likewise a reflection between $X$- and $Z$-checks:
    \begin{align}
    \rho_S:C\times B \longleftrightarrow B\times C \quad,\quad
    \rho_S(c,b)=(b,c)
    \end{align}
    with $b\in B$ and $c\in C$. By symmetry of the HGP construction \eqref{eq:HGP H_X,H_Z} when $H_1=H_2=H$, $\rho_Q$ exchanges $X$- and $Z$-supports:
    for every $X$-check $s_X$ and every data qubit $q$,
    \begin{align}
    q\in \operatorname{supp}(s_X)
    \quad\Longleftrightarrow\quad
    \rho_Q(q)\in \operatorname{supp}(\rho_S(s_X)) \, ,
    \end{align}
    and the analogous statement holds with $X$ and $Z$ interchanged, which is just a restatement of the usual $ZX$-duality used to define the fold-transversal gates in the unreduced HGP code.
    
    Because the same $\chi$-coloring is used on both input codes, the induced
    product coloring is fold-symmetric itself:
    \[
    \rho_Q(\Gamma_{i,j})=\Gamma_{j,i}
    \qquad
    \text{for all } i,j\in[\chi].
    \]
    Hence the diagonal color groups $\Gamma_{i,i}$ are exactly the fixed points of
    the fold. If one were to combine stabilizers on any $\Gamma_{i,i}$, then the
    resulting reduction would single out qubits lying on the fold axis and destroy
    the mirror pairing required for the fold-transversal symmetry. Therefore
    all diagonal groups must be excluded.
    Since condition 1 preserves the fold axis, and condition 2 removes off-axis qubits only in mirror pairs, this reduction schedule preserves the $ZX$-fold-symmetry.
\end{proof}

Note that Proposition~\ref{prop:fold preservation} only gives generic constraints on fold-symmetric reduction scheduling but does not specify any specific schedule. However, observe that the net effect of the two restrictions is that we now only have freedom to schedule our reduction in the upper-triangular color sectors, with the lower-triangular part completely determined by the fold symmetry. Therefore, we can run Algorithm \ref{alg:max_color_groups} on the upper-triangular product-color groups $\{\Gamma_{i,j>i}\}$ to obtain a maximal fold-symmetric qubit reduction.
    
The physical fold-transversal circuits~\cite{Quintavalle_2023_fold} are defined entirely in terms of this diagonal reflection: qubits on the fold axis are fixed, and off-axis qubits are paired with their mirrors. Since Theorem~\ref{thm:basis preservation} tells us that the canonical logical basis on the \Bsectorname{} sector is unchanged by the reduction, these fold-transversal circuits induce the same logical Hadamard-type and phase-type actions on $\widetilde{\mathcal{Q}}$ as on $\mathcal{Q}$.

\subsection{Automorphism gates}\label{sec:automorphisms}

A (permutation) automorphism of a classical or quantum code is a relabeling of its data bits or qubits which preserves its codespace. Automorphisms can expand the group of logical actions from transversal gates by enabling the physical support of the gates to act on different sets of data qubits related by a code automorphism~\cite{Grassl_2013, gong2024reed, SHYPS_codes, phantom_codes}. More general stabilizer-code automorphisms, including local Clifford transformations, and algorithms for finding them are reviewed in \cite{Sayginel_2025_AutQEC}. We focus below on the permutation automorphisms inherited from the input Tanner graphs. For a classical linear code with parity-check matrix $H \in \mathbb{F}^{m\times n}_2$, an automorphism is a permutation matrix $\sigma \in \mathrm{S}_n$ (in its defining representation) such that
\begin{align}\label{eq:aut condition}
    A H &= H \sigma \, ,
\end{align}
for invertible matrix $A \in \mathrm{GL}(m,\mathbb{F}_2)$. The collection of $\{\sigma\}$ that satisfy \eqref{eq:aut condition} form a group called the automorphism group of the code. When $A$ is also a permutation matrix, then $\sigma$ and $A$ describe an automorphism of the Tanner graph induced by $H$, with a corresponding Tanner-graph automorphism subgroup. For a quantum CSS code, there are two automorphism conditions \eqref{eq:aut condition}: one for the $X$-type sector and one for the $Z$-type sector. A subgroup of the automorphism group is the logical automorphism group, whose elements represent distinct logical actions that arise from the physical permutations. For classical linear and quantum CSS codes, these logical actions are restricted to the affine class of reversible gates~\cite{aaronson2015reversible}; i.e. those generated by CNOTs.

For HGP codes, it has been shown that Tanner-graph automorphisms of the classical codes lift to Tanner-graph automorphisms of the quantum HGP code~\cite{HGP_aut_gadgets}. Specifically, given classical Tanner-graph automorphisms
\begin{subequations}
\begin{align}
    \tau_1 H_1 &= H_1\sigma_1  \\
    \tau_2 H_2 &= H_2\sigma_2
\end{align}
\end{subequations}
for $\tau_i \in \mathrm{S}_{m_i}$, we obtain a HGP Tanner-graph automorphism
\begin{subequations}\label{eq:HGP TGraph Aut}
\begin{align}
    T_X H_X &= H_X \Sigma  \label{eq:HGP TGraph Aut X} \\
    T_Z H_Z &= H_Z \Sigma  \label{eq:HGP TGraph Aut Z}
\end{align}
\end{subequations}
with right (physical) permutation
\begin{align}
    \Sigma = (\sigma_1\otimes\sigma_2) \oplus (\tau_1\otimes\tau_2)
\end{align}
and left (stabilizer) permutations
\begin{subequations}
\begin{align}
    T_X &= \tau_1 \otimes \sigma_2  \\
    T_Z &= \sigma_1 \otimes \tau_2 \, .
\end{align}
\end{subequations}
The physical permutations of the first input code become row permutations in the HGP code, and the physical permutations of the second input code become column permutations. The group $\mathcal{A}_{\rm T} \simeq \langle \Sigma \rangle$ of inherited HGP Tanner-graph automorphisms is a direct product of those of the input codes. The logical Tanner-graph automorphism group $\bar{\mathcal{A}}_{\rm T}$ is also a direct product of those of the input codes as a result of the canonical logical basis \eqref{eq:HGP G_X,G_Z B}.

\begin{prop}[Tanner-graph symmetry preservation conditions]
    Suppose a reduced HGP code is obtained by \eqref{eq:tildeH,H relation}
    \[
    \widetilde H_X = W_X H_X V \, ,
    \qquad
    \widetilde H_Z = W_Z H_Z V \, .
    \]
    Suppose for some $g\in \mathcal{A}_{\rm T}$, there exist permutation matrices $\widetilde T_X(g), \widetilde T_Z(g), \widetilde \Sigma(g)$ such that
    \begin{subequations}\label{eq:equivariance push-through}
    \begin{align}
    V\,\widetilde \Sigma(g) &= \Sigma(g) V  \label{eq:qubit-equivariance}  \\
    \widetilde T_X(g)\,W_X &= W_X\,T_X(g)  \label{eq:X-equivariance}  \\
    \widetilde T_Z(g)\,W_Z &= W_Z\,T_Z(g) \, .  \label{eq:Z-equivariance} 
    \end{align}
    \end{subequations}
    Then,
    \begin{equation}
    \widetilde T_X(g)\,\widetilde H_X = \widetilde H_X\,\widetilde \Sigma(g) \, ,
    \quad\;
    \widetilde T_Z(g)\,\widetilde H_Z = \widetilde H_Z\,\widetilde \Sigma(g) \, ,
    \label{eq:reduced-TG-aut}
    \end{equation}
    and so \(g\in \mathcal{A}_{\rm T}\) descends to a Tanner-graph automorphism of the reduced HGP code.
\end{prop}

\begin{proof}
    Direct computation gives (dropping $g$ label)
    \begin{align}
        \widetilde T_X \widetilde H_X &= \widetilde T_XW_XH_XV  \notag \\
        &= W_XT_XH_XV  \notag \\
        &= W_XH_X\Sigma V  \notag \\
        &= W_XH_XV \widetilde\Sigma  \notag \\
        &= \widetilde H_X \widetilde\Sigma \, ,
    \end{align}
    where we have used \eqref{eq:X-equivariance}, \eqref{eq:HGP TGraph Aut X}, \eqref{eq:qubit-equivariance}, and \eqref{eq:tildeHx} in the second, third, fourth, and fifth lines respectively.
    Similarly, using \eqref{eq:HGP TGraph Aut Z} and \eqref{eq:Z-equivariance}, we obtain $\widetilde T_Z \widetilde H_Z = \widetilde H_Z \widetilde\Sigma$.
\end{proof}

Let $\mathcal A_{\rm desc}\subseteq\mathcal A_{\rm T}$ contain the original Tanner-graph automorphisms that satisfy the push-through relations~\eqref{eq:equivariance push-through}. Restriction to the kept qubits defines a homomorphism $\rho:\mathcal A_{\rm desc}\rightarrow\operatorname{Aut}_{\rm T}(\widetilde{\mathcal Q})$. The inherited reduced group is its image, $\widetilde{\mathcal A}_{\rm T}:=\rho(\mathcal A_{\rm desc})\simeq\mathcal A_{\rm desc}/\ker\rho$. This need not be the full automorphism group of the reduced code. Automorphisms of the original code that do not preserve the reduction are lost, while the reduced parity-check matrices may acquire automorphisms that do not lift to the original HGP code. Thus, there is no general subgroup relation between the two full automorphism groups. For the full stabilizer-code automorphism group, including local Clifford transformations, compatible original automorphisms again give only an inherited subgroup. The full group of the reduced code must instead be computed from its stabilizer matrix and may contain additional elements.

Let us now explain graphically what the automorphism preservation conditions mean. Let $B_i$ and $C_i$ denote the bit and check nodes of the Tanner graph of $H_i$. The HGP qubits decompose as
\begin{align}
Q = (B_1\times B_2) \sqcup (C_1\times C_2),
\end{align}
where, as a reminder, $B_1\times B_2$ are the \Bsectorname{} qubits and $C_1\times C_2$ are the
\Csectorname{} qubits. Under $g\in \mathcal{A}_{\rm T}$, the induced qubit permutation $\Sigma(g)$ maps
\begin{subequations}
\begin{align}
    (b_1,b_2) &\xmapsto{\Sigma} (\sigma_1b_1,\sigma_2b_2)  \\
    (c_1,c_2) &\xmapsto{\Sigma} (\tau_1c_1,\tau_2c_2) \, .
\end{align} 
\end{subequations}
So \eqref{eq:qubit-equivariance} can be interpreted as demanding that the set of kept qubits $\widetilde{Q}$ is $\mathcal{A}_{\rm T}$-invariant.

Likewise, the \(X\)-checks indexed by \(C_1\times B_2\) and transform as
\begin{align}
(c_1,b_2) \xmapsto{T_X} (\tau_1c_1,\sigma_2b_2) \, ,
\end{align}
while the \(Z\)-checks are indexed by \(B_1\times C_2\) and transform by
\begin{align}
(b_1,c_2) \xmapsto{T_Z} (\sigma_1b_1,\tau_2c_2) \, .
\end{align}
Hence \eqref{eq:X-equivariance} and \eqref{eq:Z-equivariance} can be interpreted as demanding that the families of $X$-type and $Z$-type stabilizer combinations used in the reduction are $\mathcal{A}_{\rm T}$-invariant.

In particular, since $V$ removes only \Csectorname{} qubits, a sufficient combinatorial formulation of the automorphism preservation conditions (push-through relations) \eqref{eq:equivariance push-through} is:
\begin{enumerate}
\item The removed qubits form a union of $\mathcal{A}_{\rm T}$-orbits in $C_1\times C_2$.
\item The $X$-check combinations are taken on $\mathcal{A}_{\rm T}$-invariant subsets of $C_1\times B_2$.
\item The $Z$-check combinations are taken on $\mathcal{A}_{\rm T}$-invariant subsets of $B_1\times C_2$.
\end{enumerate}

Since our canonical logical basis \eqref{eq:HGP G_X,G_Z B} is preserved under the reduction as promised by Theorem \ref{thm:basis preservation}, the logical actions of $\widetilde{\mathcal{A}}_{\rm T}$ are the same as those on the unreduced code.

As an example, the quantum Tanner transformation used in~\cite{Hong_2024_Quantinuum} satisfies all three conditions above. Note that we can technically be a bit looser and only demand that $\widetilde T_X$ and $\widetilde T_Z$ be invertible matrices instead of permutations, in order to satisfy the automorphism condition \eqref{eq:aut condition} instead of the stricter Tanner-graph automorphism condition.

\subsection{Code homomorphisms}\label{sec:GPPM}

In this section, we will show how HGP code homomorphisms can be adapted to our reduced HGP codes. Homomorphic gadgets can be used to selectively measure logical Pauli products within a single HGP block, dubbed grid-Pauli-product measurement (GPPM) gadgets~\cite{Xu_2025_fast}. In particular, we will show that both the data and ancilla code blocks can be reduced in an equivariant way that respects their homomorphic chain map.

The main feature that enables these homomorphic gadgets is the existence of a chain homomorphism~\cite{Huang_2023} (Definition~\ref{def:chain homomorphism}) between the chain complex of the HGP data code $\mathcal{Q}$ and a suitably modified HGP ancilla code $\mathcal{Q}'$, which enables a ``homomorphic'' (transversal) CNOT gate between corresponding logical qubits of the codes. For HGP codes, the ancilla code will be another HGP code with data qubits removed or added. Because this ancilla code keeps the same product structure as that of the data code, the homomorphic CNOT can be realized by applying physical CNOTs between corresponding pairs of physical qubits (and by doing nothing to physical qubits without a pair). In all, this enables us to do logical measurements.

We now review code homomorphisms, how homomorphic CNOTs arise from them, and how they appear in our HGP codes. We refer the reader to Appendix \ref{app:codes algebraic} for background on the relevant mathematical terminologies such as chain complexes and their relation to CSS codes. Recall that a quantum CSS code $\mathcal{Q}$ is associated with a 3-term (co)chain complex
\begin{align}
    \mathcal{Q}: \qquad S_X \xlongrightarrow{H^\transp_X} Q \xlongrightarrow{H^{}_Z} S_Z \, ,
\end{align}
where $S_X \in \F^{m_X}_2$, $Q\in\F^{n}_2$ and $S_Z\in\F^{m_Z}_2$ are the binary vector spaces of $X$-syndromes, data qubits and $Z$-syndromes respectively.

\begin{defn}[CSS chain homomorphism]\label{def:chain homomorphism}
    Suppose we have two CSS codes $\mathcal{Q}, \mathcal{Q}'$ corresponding to two 3-term chain complexes of the form in \eqref{eq: CSS complex}. A CSS chain homomorphism $\Gamma$ is given by three linear maps $\Gamma^{}_X : S^{}_X \rightarrow S'_X$, $\Gamma_Q : Q \rightarrow Q'$ and $\Gamma^{}_Z : S^{}_Z \rightarrow S'_Z$ such that the following diagram commutes:
    \begin{equation}
        \begin{tikzcd}
    	{\mathcal{Q}:} & {S_X} & Q & {S_Z} \\
    	{\mathcal{Q}':} & {S'_X} & {Q'} & {S'_Z}
        \arrow[from=1-1, to=2-1]
    	\arrow["{H^\transp_X}", from=1-2, to=1-3]
    	\arrow["{\Gamma^{}_X}"', from=1-2, to=2-2]
    	\arrow["{H^{}_Z}", from=1-3, to=1-4]
    	\arrow["\Gamma"', from=1-3, to=2-3]
    	\arrow["{\Gamma^{}_Z}"', from=1-4, to=2-4]
    	\arrow["{H^{\prime\transp}_X}", from=2-2, to=2-3]
    	\arrow["{H'_Z}", from=2-3, to=2-4]
        \end{tikzcd}
    \end{equation}
    In other words, $\Gamma H^\transp_X = H^{\prime\transp}_X \Gamma^{}_X$ and $\Gamma^{}_Z H^{}_Z = H'_Z \Gamma$.
\end{defn}

The arrow going from $\mathcal{Q} \rightarrow \mathcal{Q}'$ corresponds to the fact that the homomorphic CNOT is controlled on a logical qubit in $\mathcal{Q}'$ targeting a logical qubit in $\mathcal{Q}$ and vice versa~\cite{Huang_2023}.

For HGP codes, chain (complex) homomorphisms on the classical input codes $\mathcal{C}_1$, $\mathcal{C}_2$ rise to a chain homomorphisms on the corresponding HGP code $\mathcal{Q}$~\cite{Xu_2025_fast}. There are two well-known modifications of classical LDPC codes which lead to classical code homomorphisms: augmenting and puncturing. We will show how each of them is compatible with our qubit reduction scheme.

\subsubsection{Augmentation} \label{subsec:augmenting}

\begin{figure}[t]
    \centering
    \includegraphics[width=0.4\textwidth]{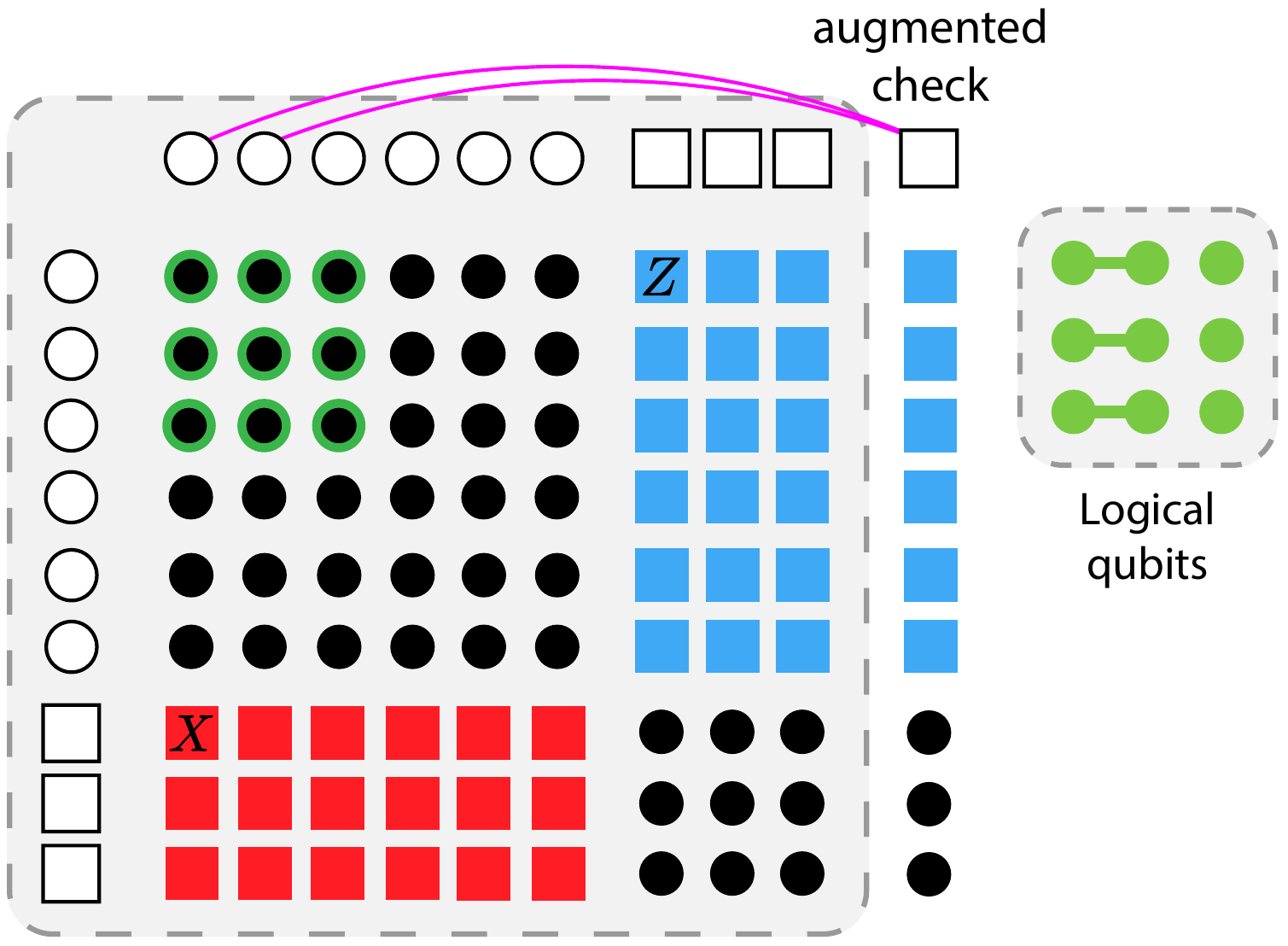}
    \caption{Augmenting the second input classical code results in a new column of $Z$-checks and \Csectorname{} qubits in the HGP code (HGP edges not shown). Qubits (black circles) with green highlights correspond to informational bits in the classical input codes. Since the new check acts on the first two information bits, the first two columns of logical qubits are fused pairwise, resulting in $9\rightarrow 6$ independent logical qubits.}
    \label{fig:augmenting}
\end{figure}

Augmentation is a procedure where a new parity check is added to a linear code, thereby enlarging the parity-check matrix $H$ by a single row. If this new check is linearly independent from the others, then the rank of $H$ increases by 1, and so as a consequence the number of logical bits decreases by 1. Without loss of generality, recall from \eqref{Gcanonical} that the canonical form of a classical generator matrix $G \in \F_2^{k \times n}$ is $\begin{pmatrix} \ident_k & A^\transpose \end{pmatrix}$ for some matrix $A$. We call the first $k$ bits of the matrix in this form informational bits. We can rearrange the columns of $H \in \F_2^{m \times n}$ to respect this same bit ordering, where the $k$ informational bits come first, to get the form of \eqref{ej}: $H = \begin{pmatrix} A & \ident_{n-k} \end{pmatrix}$. Clearly these satisfy $HG^\transpose = \mathbf{0}$. By choosing the augmented checks to only act on the $k$ informational bits, we can controllably alter the properties of the modified code to our desire. 

As an example, suppose we have an $[n,k,d]$ code $\mathcal{C}$ whose generator matrix is in canonical form $G = \begin{pmatrix} \ident_k & A^\transpose \end{pmatrix}$ with $k \geq 2$. Then suppose we augment a new parity check with support on the first two informational bits; i.e. we add the row $\begin{pmatrix} 1 & 1 & \mathbf{0}_{1\times(n-2)} \end{pmatrix}$ to $H = \begin{pmatrix} A & \ident_{n-k} \end{pmatrix}$. Of course we need a new generator matrix $G'$ to satisfy $H'G'^\transpose=0$ for 
\begin{align}
H' = \begin{pmatrix} A & \ident_{n-k} \\ \begin{matrix} 1 & 1\end{matrix} & \mathbf{0}\end{pmatrix} \, .    
\end{align}
However, since we augmented on the informational bits, we can predict what this matrix looks like. First, note that only the first two logical bits (rows of $G$) will need to be changed, since the new parity check only checks the first two bits of each codeword. While the first two rows of $G$ do not satisfy the augmented parity check, their linear combination over $\F_2$ does (since the new parity check enforces that the first two bits of codewords must be the same). So we simply replace the first two rows of $G$ with their combination; in other words, we ``fuse'' the first two logical bits. The result is a new $(k-1)\times n$ matrix $G'$, which is a valid generator matrix for the augmented code $\mathcal{C}'$. Since the augmented code is a subcode of the original code, its minimum distance cannot decrease, so $d' > d$.


Suppose the input codes of an HGP are in canonical form and thus create a $k_1k_2$ sector of ``logical'' qubits on the top-left of the \Bsectorname{} qubits, as exemplified by the highlighted green qubits in Figure~\ref{fig:augmenting}. Suppose we applied the example augmentation procedure above to the second (horizontal) input code of a HGP code, as illustrated in Figure~\ref{fig:augmenting}. The result is a modified HGP code where the logical qubits of the first two columns are horizontally ``merged'' pairwise. For each fused pair, the new logical $\bar{X}$ operator is the product of the two individual $\bar{X}$ operators, and the new logical $\bar{Z}$ operator can be either of the original $\bar{Z}$ operators, which are now stabilizer-equivalent due to the new augmented $Z$-checks. If we augmented the first input code instead, then we fuse the first two rows of logical qubits, with the roles of $X$ and $Z$ interchanged above.

Our qubit reduction scheme in Section~\ref{sec:procedure} is straightforwardly compatible with input code augmentation. To see why, observe that the addition of new checks does not change the connectivity of the original \Csectorname{} qubits, and so any qubit reduction on the original HGP code is also a valid qubit reduction on the modified HGP code. In particular, any coloring of the \Csectorname{} qubits prior to augmentation also remains a valid coloring after augmentation. If desired, one can additionally reduce the new \Csectorname{} qubits to further lower the spatial overhead.

Let $\mathcal{Q}$ denote the original HGP code and $\mathcal{Q}'$ its augmented version. Similarly, let $\widetilde{\mathcal{Q}}$ be the reduced HGP code after the procedure in Section~\ref{sec:procedure} and $\widetilde{\mathcal{Q}}'$ its augmented version. It has been shown that a chain homomorphism (Definition \ref{def:chain homomorphism}) exists between $\mathcal{Q}$ and $\mathcal{Q}'$~\cite{Xu_2025_fast}. We will now show the existence of a chain homomorphism between $\widetilde{\mathcal{Q}}$ and $\widetilde{\mathcal{Q}}'$.

\begin{figure}[t]
    \centering
    \includegraphics[width=0.3\textwidth]{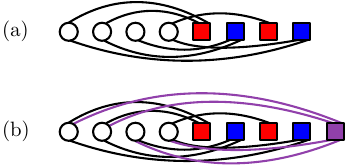}
    \caption{\textbf{(a)} Example check-coloring of a Tanner graph. \textbf{(b)} Augmenting a check (purple) does not change the coloring of the original nodes of the Tanner graph.}
    \label{augmentedcolors}
\end{figure}

\begin{lem}[Augmented check coloring]\label{lem:augmented coloring}
    Let $H \in \F^{m\times n}_2$ be a parity-check matrix associated with a Tanner graph $\mathcal{G}$ of $m$ checks and $n$ bits, and let $\chi$ be a coloring of the check nodes based on bits checks in common. Suppose $H' \in \F^{m'\times n}_2$ is an augmented version of $H$ with $m'-m>0$ additional parity checks. Then $\chi$ remains a valid coloring of the original $m$ check nodes on the augmented Tanner graph $\mathcal{G}'$.
\end{lem}

\begin{proof}
    Let $\Gamma_S = \begin{pmatrix} \ident_m \\ \mathbf{0}_{(m'-m)\times m} \end{pmatrix}$ represent the inclusion map $S \hookrightarrow S'$ from the original checks (that span $S$) to the new checks (that span $S'$).
    Then $\Gamma^\transp_S H' \in \F^{m\times n}_2$ is the restriction of $H'$ to the original $m$ rows corresponding to $H$. By construction, $\Gamma^\transp_S H' = H$, and so we have $HH^\transp = (\Gamma^\transp_S H')(\Gamma^\transp_S H')^\transp = \Gamma^\transp_S (H'H^{\prime\transp}) \Gamma^{}_S$, which implies that the check-adjacency graph of $\mathcal{G}$ and that of $\mathcal{G}'$ are identical when restricted to the original $m$ check nodes. Thus $\chi$ is a valid check-coloring of $\mathcal{G}'$ on the original $m$ check nodes. This can be easily understood intuitively since an augmented check will not alter the structure of the original checks; see Figure~\ref{augmentedcolors}.
\end{proof}

Now, let $\mathcal{Q}$ be a HGP code with input parity-check matrices $H_1,H_2$, and let $Q$ denote its data qubits. Let $\mathcal{Q}'$ be its augmented variant upon augmenting the horizontal code $H_2$ to $H'_2$. Let $\widetilde{\mathcal{Q}}$ denote the reduced version of $\mathcal{Q}$ with data qubits $\tilde{Q}$, and let $\widetilde{\mathcal{Q}}'$ be the reduced version of $\mathcal{Q}'$ using the same reduction scheme used from $\mathcal{Q}$ to $\mathcal{Q}'$.

\begin{thm}[Reduced HGP chain homomorphism for augmentation]\label{thm:augmentation homomorphism}
    Let $\widetilde{\mathcal{Q}}'$ be the vertically augmented version of $\widetilde{\mathcal{Q}}$. There exist sparse maps $\widetilde{\Gamma}_X, \widetilde{\Gamma}_Q, \widetilde{\Gamma}_Z$ such that the following diagram commutes:
    \begin{equation}\label{eq:aug diagram}
        \begin{tikzcd}
        \widetilde{\mathcal{Q}}: & \widetilde{S}_X \arrow[r, "\widetilde{H}^\transp_X"] & \tilde{Q} \arrow[r, "\widetilde{H}_Z"] & \widetilde{S}_Z \\
        \arrow[u] \widetilde{\mathcal{Q}}': & \widetilde{S}'_X \arrow[r, "\widetilde{H}^{\prime\transp}_X"] \arrow[u, "\widetilde{\Gamma}_X"] & \tilde{Q}' \arrow[r, "\widetilde{H}'_Z"] \arrow[u, "\widetilde{\Gamma}"] & \widetilde{S}'_Z \arrow[u, "\widetilde{\Gamma}_Z"]
        \end{tikzcd}
    \end{equation}
\end{thm}

\begin{proof}
    Recall that since we have augmented $H_2$, we have additional $Z$-checks and \Csectorname{} qubits; the number of $X$-checks is unchanged. First we will show that $\widetilde{\mathcal{Q}}'$ is a valid CSS code. Let $\chi_1$ and $\chi_2$ denote the check-colorings of $H_1$ and $H_2$ respectively. By Lemma \ref{lem:augmented coloring}, $\chi_2$ is still a valid check-coloring of $H'_2$. Since $\widetilde{H}'_X$, $\widetilde{H}_X$ and $\widetilde{H}'_Z$, $\widetilde{H}_Z$ are identical when restricted to the \Csectorname{} qubits in $\tilde{Q}$, $\chi_1 \times \chi_2$ is also a valid product coloring of the \Csectorname{} qubits for qubit reduction of $\widetilde{\mathcal{Q}}'$. Hence, because $\widetilde{\mathcal{Q}}$ is a valid CSS code by assumption, $\widetilde{\mathcal{Q}}'$ is also a valid CSS code; i.e. $\widetilde{H}'_Z\widetilde{H}^{\prime\transp}_X = 0$.
    
    Let $\widetilde{\Gamma}$ be the transpose of the inclusion map associated with $\tilde{Q} \hookrightarrow \tilde{Q}'$. The relation between $\widetilde{H}_X$ and $\widetilde{H}'_X$ can be written as $\widetilde{H}'_X\Gamma^\transp = \widetilde{H}_X$. Let $\widetilde{\Gamma}_X = \ident_{\tilde{m}_X}$ since we have $\widetilde{S}^{}_X \simeq \widetilde{S}'_X$. Then we have $\widetilde{H}_X = \tilde\Gamma^\transp_X \widetilde{H}^{}_X = \widetilde{H}'_X\Gamma^\transp$, or equivalently,
    \begin{align}\label{eq:aug comm 1}
        \widetilde{H}^\transp_X \tilde\Gamma_X = \Gamma \widetilde{H}^{\prime\transp}_X \, .
    \end{align}
    Now let $\widetilde{\Gamma}_Z$ be the transpose of the inclusion map associated with $\widetilde{S}^{}_Z \hookrightarrow \widetilde{S}'_Z$. Since the augmentation did not change the structure of the original $Z$-checks acting on the original data qubits $\tilde{Q}$, we have $\tilde\Gamma^{}_Z \widetilde{H}'_Z \tilde\Gamma^\transp = \widetilde{H}_Z$, or equivalently,
    \begin{align}\label{eq:aug comm 2}
        \tilde\Gamma^{}_Z \widetilde{H}'_Z = \widetilde{H}_Z \Gamma \, .
    \end{align}
    Together, \eqref{eq:aug comm 1} and \eqref{eq:aug comm 2} show that the diagram \eqref{eq:aug diagram} commutes.
\end{proof}

Since the qubit reductions preserve the canonical logical bases of both $\tilde{\mathcal{Q}}$ and $\tilde{\mathcal{Q}}'$ individually by Theorem \ref{thm:basis preservation}, the logical action of \eqref{eq:aug diagram} is the same as that of the unreduced version.

\subsubsection{Puncturing}

\begin{figure}[t]
    \centering
    \includegraphics[width=0.38\textwidth]{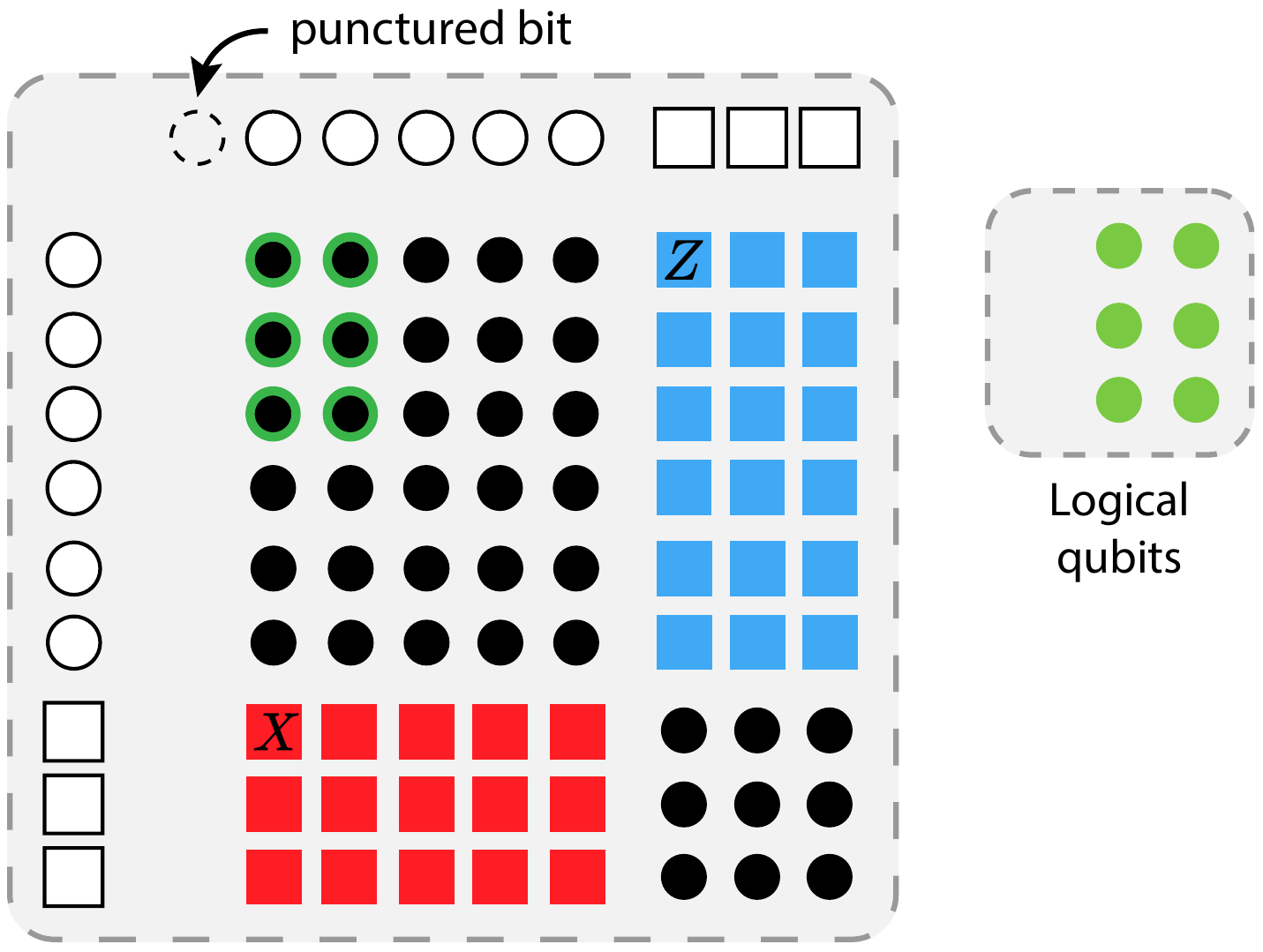}
    \caption{Puncturing the second input classical code results in the deletion of a column of \Bsectorname{} qubits and $X$-checks. In this case, since the first informational bit has been punctured, the first column of logical qubits are deleted in the HGP code.}
    \label{fig:puncturing}
\end{figure}

Puncturing is a procedure where bits are deleted from a linear code, thereby reducing the length of the code or number of columns in the parity-check and generator matrices. Similar to the case of augmentation, by puncturing only the informational bits, we maintain control over the modified code's parameters. Suppose we have a $[n,k,d]$ code $\mathcal{C}$ with canonical generator matrix $G = \begin{pmatrix} \ident_k & A^\transpose \end{pmatrix}$, and we puncture $n-n'\leq k$ columns in first the $k$ columns (informational bits). Let $H \in \F^{m\times n}_2$ be a parity-check matrix for $G$ such that $HG^\transp=0$ and $\rank H = n-k$. Let $H' \in \F^{m\times n'}_2$ denote the punctured parity-check matrix.

\begin{lem}[Puncturing informational bits does not alter $\rank H$]
\label{lem:cannot puncture check}
    $\rank H' = \rank H = n-k$.
\end{lem}

\begin{proof}
    Suppose, on the contrary, that $\rank H' < \rank H$ (clearly the rank cannot increase by removing columns). Then it implies the existence of a new linear dependency $\sum_i \mathbf{c}'_i = 0$ among some subset of punctured checks $\{\mathbf{c}'_i\} \subset H'$ that did not exist in the unpunctured code; i.e. $\sum_i \mathbf{c}_i \neq 0$. Let $\hat{\mathbf{c}} \equiv \sum_i \mathbf{c}_i$. Then by construction, $\hat{\mathbf{c}}$ is only supported on the punctured bits. Since the punctured bits were informational, we must then have $\hat{\mathbf{c}} \notin \rs{H}$, which contradicts its definition: $\hat{\mathbf{c}} = \sum_i \mathbf{c}_i \in \rs{H}$.
\end{proof}

Since $\rank H$ remains invariant, the number of logical bits then decreases by the number of punctured informational bits. Since the standard generator matrix takes the form of $\ident_k$ on the informational columns, if whenever we puncture a column we also remove (expurgate) the corresponding row (corresponding to the 1 in that column), then the resulting modified generator matrix is a valid generator matrix, also in standard form, for the punctured code. The support of the remaining codewords remains unchanged.

Suppose we punctured the first informational bit in the second input code of a HGP code, as illustrated in Figure~\ref{fig:puncturing}. The result is a modified HGP code where the data qubits, logical qubits and $X$-checks in the first column are deleted. If we punctured the first input code instead, then we would delete the first row of data qubits, logical qubits and $Z$-checks. Puncturing additional informational bits results in the deletion of more data qubits, logical qubits and checks.

Let $\mathcal{Q}$ denote the original HGP code and $\mathcal{Q}'$ its punctured version. Similar to the case for augmentation (Section \ref{subsec:augmenting}), there exists a chain homomorphism between $\mathcal{Q}$ and $\mathcal{Q}'$~\cite{Xu_2025_fast}. Let $\tilde{Q}$ denote the reduced version of $\mathcal{Q}$ after the procedure in Section \ref{sec:procedure}. We will now construct a reduced, punctured HGP code $\widetilde{\mathcal{Q}}'$ such that there exists a chain homomorphism between $\widetilde{\mathcal{Q}}$ and $\widetilde{\mathcal{Q}}'$.

Since the reduction procedure is independent for each \Csectorname{} qubit of interest, it suffices to analyze the case for a single \Csectorname{} qubit. We will focus on the procedure for combining $X$-checks, as the case for $Z$-checks is analogous. Recall that for a reducible \Csectorname{} qubit, we combine its $\Delta$ neighboring $X$-checks in pairs according to the local reduction matrix $w \in \F^{(\Delta-1)\times\Delta}_2$ \eqref{eq:1D rep PCM}, which is the parity-check matrix of the 1D repetition code of length $\Delta$. If some of the $X$-checks are deleted as a result of an input code puncture, we simply combine the remaining $\Delta'<\Delta$ $X$-checks in the same fashion; see Figure \ref{fig:punctured reduction} for an illustration. The local reduction $w' \in \F^{(\Delta'-1)\times\Delta'}_2$ is now the parity-check matrix of the 1D repetition code of length $\Delta'$. If $\Delta'=1$, then we simply just remove the remaining $X$-check. Note that we cannot have $\Delta'=0$ because it would imply that we have punctured the entire support of a parity check in the second input code, which would reduce $\rank H'_2$ and hence contradict Lemma \ref{lem:cannot puncture check}; a more general proof with Tanner codes is given in Appendix \ref{app:Tanner codes} (Lemma \ref{lem:Tanner code max puncture}). At the end of this procedure, we can then remove the same \Csectorname{} qubits as those in the unpunctured reduced code $\widetilde{\mathcal{Q}}$ to obtain $\widetilde{\mathcal{Q}}'$.

\begin{figure}[t]
    \centering
    \includegraphics[width=0.35\textwidth]{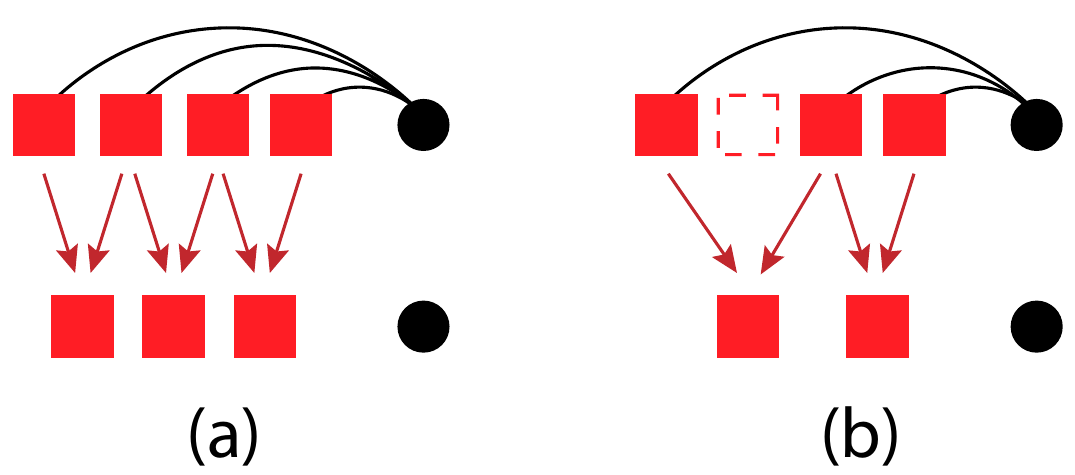}
    \caption{\textbf{(a)} The reduction procedure from Section \ref{sec:procedure} applied to the $X$-checks (red squares) supported on a \Csectorname{} qubit (black circle). Four $X$-checks are combined pairwise such that the three new $X$-checks have no support on the \Csectorname{} qubit. \textbf{(b)} The modified reduction procedure when an $X$-check is deleted as a result of an input code puncture. The three remaining $X$-checks are still combined pairwise to form two new $X$-checks that have no support on the \Csectorname{} qubit.}
    \label{fig:punctured reduction}
\end{figure}

\begin{thm}[Reduced HGP chain homomorphism for puncturing]\label{thm:puncturing homomorphism}
    Let $\widetilde{\mathcal{Q}}'$ be the vertically punctured version of $\widetilde{\mathcal{Q}}$. There exist sparse maps $\widetilde{\Gamma}_X, \widetilde{\Gamma}, \widetilde{\Gamma}_Z$ such that the following diagram commutes:
    \begin{equation}\label{eq:punc diagram}
    \begin{tikzcd}
	{\widetilde{S}_X} & {\tilde{Q}} & {\widetilde{S}_Z} \\
	{\widetilde{S}'_X} & {\tilde{Q}'} & {\widetilde{S}'_Z}
	\arrow["{\widetilde{H}^\transp_X}", from=1-1, to=1-2]
	\arrow["{\widetilde{\Gamma}_X}"', from=1-1, to=2-1]
	\arrow["{\widetilde{H}_Z}", from=1-2, to=1-3]
	\arrow["{\widetilde{\Gamma}}"', from=1-2, to=2-2]
	\arrow["{\widetilde{\Gamma}_Z}"', from=1-3, to=2-3]
	\arrow["{\widetilde{H}^{\prime\transp}_X}", from=2-1, to=2-2]
	\arrow["{\widetilde{H}'_Z}", from=2-2, to=2-3]
\end{tikzcd}
    \end{equation}
\end{thm}

\begin{proof}
    From the chain homomorphism for punctured HGP codes~\cite{Xu_2025_fast}, we have
    \begin{align}\label{eq:H_X,H'_X relation}
        H^{}_X \Gamma^\transp = \Gamma^\transp_X H'_X \, .
    \end{align}
    Recall that the global $X$-check reduction $W_X$ comprises direct sums of $w_X$ and $\ident_\Delta$, where $w_X \in \F^{(\Delta-1)\times\Delta}_2$ is the local reduction that satisfies $\rs w_X = \mathcal{C}^\perp_{\rm rep}$. Since the rows of the modified reduction $w'_X$ also span a dual-repetition code but with shorter length, we have $\rs w'_X \subset \mathcal{C}^\perp_{\rm rep} = \rs w^{}_X$ upon restoring the deleted columns as 0 columns. Hence, there exists a local map $\tilde{\gamma}^{}_X \in \F^{(\Delta'-1)\times(\Delta-1)}_2$ such that $\tilde{\gamma}^{}_X w^{}_X = w'_X \gamma^{}_X$, where $\gamma^{}_X$ is an inclusion map for the restored 0 columns. The global map $\widetilde{\Gamma}_X$ then follows by direct-summing the local maps $\tilde{\gamma}^{}_X$; the sparsity of $\widetilde{\Gamma}_X$ follows from that of $\tilde{\gamma}^{}_X$. Direct-summing all of the local relations $\tilde{\gamma}^{}_X w^{}_X = w'_X \gamma^{}_X$ then implies the global relation
    \begin{align}\label{eq:W_X,W'_X relation}
        \widetilde{\Gamma}_X W_X = W'_X \Gamma_X \, .
    \end{align}
    Now, from \eqref{eq:tildeH,H relation}, we have $\widetilde{H}_X = W_XH_XV$ for the unpunctured code and $\widetilde{H}'_X = W'_XH'_XV'$ for the punctured code. Since we remove the same \Csectorname{} qubits in both the data and ancilla code blocks, we have
    \begin{align}\label{eq:V',V relation}
        V' = \Gamma V \widetilde{\Gamma}^\transp
    \end{align}
    for inclusion map $\widetilde{\Gamma}^\transp : \tilde{Q}' \hookrightarrow \tilde{Q}$. Now we can write
    \begin{align}\label{eq:tilde H'_X derivation}
        \widetilde{\Gamma}^\transp_X \widetilde{H}'_X &= \widetilde{\Gamma}^\transp_X W'_X H'_X V'  \notag \\
        &= W^{}_X \Gamma^\transp_X H'_X \Gamma V \widetilde{\Gamma}^\transp  \notag \\
        &= W_X H_X \Gamma^\transp \Gamma V \widetilde{\Gamma}^\transp  \notag \\
        &= W_X H_X V \widetilde{\Gamma}^\transp  \notag \\
        &= \widetilde{H}_X \widetilde{\Gamma}^\transp \, ,
    \end{align}
    where in the second line we used \eqref{eq:W_X,W'_X relation} and \eqref{eq:V',V relation}; in the third line we used \eqref{eq:H_X,H'_X relation}; and in the last line we used \eqref{eq:tildeH,H relation}. Transposing both sides gives us
    \begin{align}\label{eq:tilde H_X, tilde H'_X relation}
        \widetilde{H}^{\prime\transp}_X \widetilde{\Gamma}^{}_X = \widetilde{\Gamma} \widetilde{H}^\transp_X \, .
    \end{align}
    An analogous calculation to \eqref{eq:tilde H'_X derivation} for the $Z$-sector gives us $\widetilde{H}^{}_Z \widetilde{\Gamma}^\transp = \widetilde{\Gamma}^\transp_Z \widetilde{H}'_Z$, or equivalently,
    \begin{align}\label{eq:tilde H_Z, tilde H'_Z relation}
        \widetilde{\Gamma}^{}_Z \widetilde{H}^{}_Z = \widetilde{H}'_Z \widetilde{\Gamma} \, .
    \end{align}
    Together, \eqref{eq:tilde H_X, tilde H'_X relation} and \eqref{eq:tilde H_Z, tilde H'_Z relation} show that the diagram \eqref{eq:punc diagram} commutes.
\end{proof}

Similar to the case for augmentation, since the qubit reductions preserve the canonical logical bases of both $\tilde{\mathcal{Q}}$ and $\tilde{\mathcal{Q}}'$ individually by Theorem \ref{thm:basis preservation}, the logical action of \eqref{eq:punc diagram} is the same as that of the unreduced version.


\section{Examples}\label{sec:examples}

In this section, we provide examples of HGP codes that can be reduced and assess the savings performance to set up the memory simulations in Section~\ref{sec:numerics}.

\subsection{Rotated surface code}\label{sec:rotated surface example}

Let us start by discussing how the rotated surface code is perhaps the simplest instance of this reduction procedure. Let both classical inputs be the length-$L$ repetition code, with parity-check matrix $H_{\rm rep}\in\F_2^{(L-1)\times L}$. The corresponding HGP code is the distance-$L$ unrotated surface code; see Figure~\ref{fig:surface footprint}(a). The \Csectorname{} qubits form an $(L-1)\times(L-1)$ array, which are the purple qubits in the figure. Its check-adjacency graph is bipartite, and the reduction can remove every \Csectorname{} qubit by alternating the $X$- and $Z$-check combinations. Specifically, for the distance-3 surface code, the new parity-check matrices are

\begin{subequations}
    \begin{align}
        \widetilde{H}_X
        &= \underbrace{\begin{pmatrix}
            1 & 1 & 0 & 0 & 0 & 0 \\
            0 & 0 & 1 & 0 & 0 & 0 \\
            0 & 0 & 0 & 1 & 0 & 0 \\
            0 & 0 & 0 & 0 & 1 & 1
        \end{pmatrix}}_{W_X} \notag \\
        &\quad\times \underbrace{\left(\begin{array}{ccccccccc|cccc}
            1&0&0&1&0&0&0&0&0 & 1&0&0&0 \\
            0&1&0&0&1&0&0&0&0 & 1&1&0&0 \\
            0&0&1&0&0&1&0&0&0 & 0&1&0&0 \\
            0&0&0&1&0&0&1&0&0 & 0&0&1&0 \\
            0&0&0&0&1&0&0&1&0 & 0&0&1&1 \\
            0&0&0&0&0&1&0&0&1 & 0&0&0&1
        \end{array}\right)}_{H_X} \notag \\
        &\quad\times \underbrace{\begin{pmatrix} \ident_9 \\ \mathbf{0}_{4\times 9} \end{pmatrix}}_{V} \notag \\
        &= \begin{pmatrix}
            1&1&0&1&1&0&0&0&0 \\
            0&0&1&0&0&1&0&0&0 \\
            0&0&0&1&0&0&1&0&0 \\
            0&0&0&0&1&1&0&1&1
        \end{pmatrix} \label{Wx_ex}
    \end{align}
    \begin{align}
        \widetilde{H}_Z
        &= \underbrace{\begin{pmatrix}
            1 & 0 & 0 & 0 & 0 & 0 \\
            0 & 1 & 0 & 1 & 0 & 0 \\
            0 & 0 & 1 & 0 & 1 & 0 \\
            0 & 0 & 0 & 0 & 0 & 1
        \end{pmatrix}}_{W_Z} \notag \\
        &\quad\times \underbrace{\left(\begin{array}{ccccccccc|cccc}
            1&1&0&0&0&0&0&0&0 & 1&0&0&0 \\
            0&1&1&0&0&0&0&0&0 & 0&1&0&0 \\
            0&0&0&1&1&0&0&0&0 & 1&0&1&0 \\
            0&0&0&0&1&1&0&0&0 & 0&1&0&1 \\
            0&0&0&0&0&0&1&1&0 & 0&0&1&0 \\
            0&0&0&0&0&0&0&1&1 & 0&0&0&1
        \end{array}\right)}_{H_Z} \notag \\
        &\quad\times \underbrace{\begin{pmatrix} \ident_9 \\ \mathbf{0}_{4\times 9} \end{pmatrix}}_{V} \notag \\
        &= \begin{pmatrix}
            1&1&0&0&0&0&0&0&0 \\
            0&1&1&0&1&1&0&0&0 \\
            0&0&0&1&1&0&1&1&0 \\
            0&0&0&0&0&0&0&1&1
        \end{pmatrix} \label{Wz_ex}
    \end{align}
\end{subequations}

The resulting code is the distance-$L$ rotated surface code; see Figure~\ref{fig:surface footprint}(b). Thus, the number of data qubits $n_{\rm data}$ changes from $L^2+(L-1)^2$ to $L^2$.
A single ancilla per independent stabilizer gives $n_{\rm anc}=n_{\rm data}-1$ on both sides. The full data-and-ancilla footprint $n_{\rm data}+n_{\rm anc}$ therefore changes from $(2L-1)^2$ to $2L^2-1$, which is a $2(L-1)^2$ qubit savings.

\begin{figure}[t]
    \centering
    \includegraphics[width=0.43\textwidth,pagebox=cropbox]{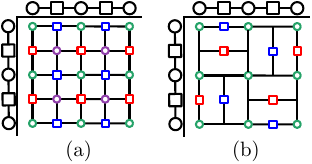}
    \caption{The distance-$L$ unrotated and rotated surface-code layouts obtained before and after the reduction. Data qubits are shown at vertices and stabilizer-measurement ancillas at plaquettes. The footprint counts include both types of physical qubits.}
    \label{fig:surface footprint}
\end{figure}

\subsection{Random LDPC codes}

\begin{table*}[t]
\renewcommand{\arraystretch}{1.35}
\begin{tabular}{c|c|c|ccc|ccc}
    \;\textbf{Input degrees}\; &
    \;\textbf{Random LDPC}\; &
    \;\,$\chi$\;\, &
    \textbf{HGP} &
    $N_{\rm 2q}$ &
    $\big(w_q, w_c\big)$ &
    \textbf{Reduced HGP} &
    $\tilde{N}_{\rm 2q}$ &
    $\big(\tilde{w}_q, \tilde{w}_c\big)$ \\
    \hline
    $(3, 4)$ & $[12, 3, 6]$ & $3$ & $\llbracket 225, 9, 6\rrbracket$ & $1218$ & $(4, 7)$ & $\llbracket 160, 9, 6\rrbracket$ & $1051$ & $(6, 12)$ \\
    $(3, 5)$ & $[20, 8, 6]$ & $5$ & $\llbracket 544, 64, 6\rrbracket$ & $3840$ & $(5, 8)$ & $\llbracket 472, 64, 6\rrbracket$ & $4344$ & $(8, 14)$ \\
    $(5, 6)$ & $[30, 5, 13]$ & $9$ & $\llbracket 1525, 25, 13\rrbracket$ & $14\,960$ & $(6, 11)$ & $\llbracket 1325, 25, 13 \rrbracket$ & $17\,630$ & $(10, 20)$
\end{tabular}
\caption{Examples of random LDPC codes, their corresponding HGP codes, and their reduced HGP variants. $\chi$ is the number of check colors in the input code (note that this is not the minimal number of colors). $N_{\rm 2q}$ ($\tilde{N}_{\rm 2q}$) is the total number of 2-qubit gates required for syndrome extraction before (after) the reduction. In $(w_q, w_c)$ ($(\tilde{w}_q, \tilde{w}_c)$), the first number is the maximum check-degree of the data qubits before (after) the reduction, and the second number is the maximum $X$ or $Z$ stabilizer weight before (after) the reduction. In implementing the reduction, we utilize the optimization presented in Section~\ref{optimization section} but do not restrict ourselves to not combine upon specific color groups to preserve fold-transversal logical operators as detailed in Section~\ref{sec:fold-transversal}.}
\label{tab:random code HGP}
\end{table*}

\begin{figure}[t]
      \centering
      \includegraphics[width=.2\textwidth]{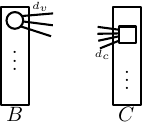}
      \caption{Using the configuration model with bit-vertex degree $d_v$ and check-vertex degree $d_c$ on a graph with mutually disjoint vertex set $V=B\cup C$ results in a $(d_v, d_c)$-LDPC code. Edges between the two sets are randomly paired. This is inherently equivalent to randomly adding $d_c$ ($d_v$) 1s to each row (column) of an $|C| \times |B|$ parity-check matrix. Note that $d_v|B|=d_c|C|$ must be satisfied so that all edges can be paired.}
      \label{configmodel}
\end{figure}

One can define a classical random LDPC code using the bipartite configuration model, which will return a random bipartite graph $G=(B \cup C, E)$, where edges do not exist between vertices within $B$ or $C$, that we can interpret as a Tanner graph with bit-vertex set $B$ and check-vertex set $C$. If we specify the number of checks each bit participates in ($d_v$), and the number of bits that each check checks ($d_c$), we can get a random classical $(d_v, d_c)$-LDPC code; see Figure~\ref{configmodel}. We also take note of the number of colors in the check-adjacency graphs (Section~\ref{def:check-adjacency graph}) of the random codes, as less check colors equates to more \Csectorname{} qubits removed. However, note that we do not necessarily get a minimal coloring (equivalent to the chromatic number $\chi(G)$) as we use a polynomial-time, greedy heuristic (independent set).

Let us provide an example. We start with a random $[20, 8, 6]$ $(3,5)$-LDPC code with parity-check matrix 
\begin{equation}
\resizebox{\columnwidth}{!}{$
H =
\begin{bmatrix}
0 & 0 & 0 & 1 & 0 & 0 & 0 & 0 & 1 & 1 & 0 & 0 & 0 & 1 & 0 & 0 & 0 & 1 & 0 & 0 \\
0 & 0 & 1 & 0 & 0 & 1 & 0 & 0 & 0 & 0 & 1 & 1 & 0 & 0 & 0 & 0 & 0 & 0 & 1 & 0 \\
0 & 0 & 0 & 0 & 0 & 0 & 0 & 0 & 1 & 0 & 0 & 0 & 0 & 0 & 0 & 1 & 1 & 0 & 1 & 1 \\
1 & 0 & 0 & 0 & 0 & 0 & 0 & 1 & 0 & 0 & 0 & 0 & 0 & 0 & 0 & 1 & 0 & 0 & 1 & 0 \\
0 & 0 & 1 & 0 & 0 & 0 & 1 & 1 & 0 & 0 & 0 & 0 & 0 & 1 & 0 & 0 & 0 & 0 & 0 & 0 \\
0 & 1 & 0 & 0 & 1 & 0 & 0 & 0 & 0 & 0 & 0 & 1 & 1 & 1 & 0 & 0 & 0 & 0 & 0 & 0 \\
0 & 1 & 0 & 1 & 0 & 0 & 0 & 0 & 0 & 1 & 0 & 1 & 0 & 0 & 0 & 0 & 1 & 0 & 0 & 0 \\
0 & 0 & 0 & 0 & 0 & 0 & 0 & 1 & 0 & 0 & 0 & 0 & 1 & 0 & 1 & 1 & 0 & 0 & 0 & 1 \\
1 & 0 & 0 & 0 & 0 & 1 & 0 & 0 & 1 & 0 & 1 & 0 & 1 & 0 & 0 & 0 & 0 & 0 & 0 & 0 \\
0 & 0 & 0 & 0 & 1 & 0 & 1 & 0 & 0 & 1 & 0 & 0 & 0 & 0 & 1 & 0 & 0 & 0 & 0 & 0 \\
0 & 1 & 0 & 1 & 1 & 0 & 0 & 0 & 0 & 0 & 1 & 0 & 0 & 0 & 0 & 0 & 0 & 0 & 0 & 1 \\
0 & 0 & 0 & 0 & 0 & 1 & 0 & 0 & 0 & 0 & 0 & 0 & 0 & 0 & 1 & 0 & 1 & 1 & 0 & 0
\end{bmatrix} \, .\label{randomldpcex1}
$}
\end{equation}
We then construct the check-adjacency graph of the code as described in Definition~\ref{def:check-adjacency graph} and find that we can 5-color it, meaning that there are five disjoint sets of parity checks where each check in a set checks no bits in common. The check color groups formed are
\begin{subequations}
    \begin{align}
    C_1 &= \{4, 6, 12\} \\
    C_2 &= \{5, 7, 9\} \\
    C_3 &= \{1, 2, 8\} \\
    C_4 &= \{3, 10\} \\
    C_5 &= \{11\} 
    \end{align}
\end{subequations}
where each number $i$ in a set represents the parity check in row $i$ of $H$. Note that it is clear that these sets have disjoint qubits. For example, to verify that the checks in color $C_1$ are disjoint, take the rows of $H$ for stabilizers 4, 6, 12:
\begin{subequations} \begin{align}
    4 &: \begin{bmatrix}1 & 0 & 0 & 0 & 0 & 0 & 0 & 1 & 0 & 0 & 0 & 0 & 0 & 0 & 0 & 1 & 0 & 0 & 1 & 0\end{bmatrix}\label{S3C1} \\
    6 &: \begin{bmatrix} 0 & 1 & 0 & 0 & 1 & 0 & 0 & 0 & 0 & 0 & 0 & 1 & 1 & 1 & 0 & 0 & 0 & 0 & 0 & 0 \end{bmatrix}\label{S5C1} \\
    12 &: \begin{bmatrix} 0 & 0 & 0 & 0 & 0 & 1 & 0 & 0 & 0 & 0 & 0 & 0 & 0 & 0 & 1 & 0 & 1 & 1 & 0 & 0 \end{bmatrix}\label{S11C1}
\end{align} \end{subequations}
Clearly there are no qubits checked in common between the stabilizers, and thus they have the same check color.
Taking the (symmetric) HGP of \eqref{randomldpcex1} with itself, we get a $\llbracket 544, 64, 6 \rrbracket$ quantum CSS code. Since we could 5-color the checks of the classical code, the \Csectorname{} qubits of the HGP code create 25 color groups.
We then reduce the HGP as described by the procedure in Section~\ref{sec:procedure}. We end up getting a $\llbracket \tilde{n}, \tilde{k}, \tilde{d} \rrbracket = \llbracket 472, 64, 6 \rrbracket$ code. We can also find that the maximum number of ($X$ or $Z$) checks that a data qubit participates in and the maximum number of data qubits that an ($X$ or $Z$) check checks goes from $\big(w_q, w_c\big) = (5, 8)$ to $\big(\tilde{w}_q, \tilde{w}_c\big) = (8, 14)$. Additionally, the number of two-qubit gates required for ($X$ and $Z$) syndrome extraction goes from $N_{\rm 2q} = 3840$ to $\tilde{N}_{\rm 2q} = 4344$. Note that this number will stay the same after the reduction iff $\tilde{\bar{w}}_q/\bar{w}_q = n/\tilde{n}$, decrease iff $\tilde{\bar{w}}_q/\bar{w}_q < n/\tilde{n}$, and increase otherwise, where $\bar{w}_q$ ($\tilde{\bar{w}}_q$) is the average number of ($X$ and $Z$) checks that a data qubit participates in before (after) the reduction.

Repeating this process for a variety of random LDPC codes, we get the following results in Table~\ref{tab:random code HGP}. Since we only ever combine at most two checks at a time, and since all combining is done in parallel, these reduced codes admit distance-preserving syndrome extraction schedules following the discussion of Section \ref{sec:hook errors}.

\subsection{Quasi-cyclic LDPC codes}

Quasi-cyclic LDPC (QC-LDPC) codes are a subclass of LDPC codes with an embedded periodic structure~\cite{tanner1999qc, tanner2001qc, Okamura_2003, Fossorier_2004}, and they are widely deployed for 5G wireless communications~\cite{Richardson_2018_5G, Patil_2020}. Like typical LDPC codes, they are presented by their parity-check matrices. There are two ingredients that make up the parity-check matrix of a QC-LDPC code: a protomatrix and a lift. The protomatrix is a binary matrix $H_{\rm proto} \in \mathbb{F}^{m_{\rm p}\times n_{\rm p}}_2$ which can be thought of as a parity-check matrix of a small code of length $n_{\rm p}$. For example, let us choose
\begin{align}\label{eq:H_proto}
    H_{\rm proto} = \begin{pmatrix}
        1&0&1&1 \\
        0&1&1&1 \\
        1&1&0&1
    \end{pmatrix}
\end{align}
as our protomatrix. For the second ingredient, the lift, we first choose a positive integer $\ell$ to be our lift size. Then, we promote each nonzero entry of $H_{\rm proto}$ from 1 to a polynomial with degree $<\ell$, or more formally an element of the quotient ring $\mathsf{R}_\ell \equiv \mathbb{F}_2[x]/(x^\ell-1)$. The new matrix $H_{\rm lift}(x) \in \mathsf{R}^{m_{\rm p}\times n_{\rm p}}_\ell$ is called the lifted matrix or matrix of circulants. For example, the following matrix \cite[Table III]{Xu_2025_fast}
\begin{align}\label{eq:H_lift}
    H_{\rm lift}(x) = \begin{pmatrix}
        x^4&0&x^4&x^3 \\
        0&x^3&x^3&x^4 \\
        x^3&x^4&0&x^3
    \end{pmatrix}
\end{align}
is a lifted version of \eqref{eq:H_proto} with lift size $\ell=5$. To convert $H_{\rm lift}$ to a binary matrix, we use the well-known isomorphism between the ring of polynomials of degree $<\ell$ and the ring of $\ell\times\ell$ circulant matrices: $x$ maps to the $\ell\times\ell$ shift matrix $\delta_{i,i+1}$ (mod $\ell$), which is the identity matrix with all columns shifted by one unit to the right and wrapping around at the ends; i.e.,
\begin{equation}
    x^0 = \begin{pmatrix} 1 \\ & 1 \\ & & \ddots \\ & & & 1 \end{pmatrix}, \quad x^1 = \begin{pmatrix} 0 & 1 \\ & 0 & \ddots \\ & & \ddots & 1 \\ 1 & & & 0\end{pmatrix}
\end{equation}
etc., where empty entries denote zeroes. Since each polynomial element of $H_{\rm lift}$ becomes a $\ell\times\ell$ circulant matrix, the full binary matrix $H_{\rm full}$ has dimensions $m_{\rm p}\ell \times n_{\rm p}\ell$. Let $w_\ell$ denote the maximum weight or number of monomials among all the nonzero elements of $H_{\rm lift}$. It follows that if $H_{\rm proto}$ is $(w_b,w_c)$-LDPC, then $H_{\rm full}$ is $(w_bw_\ell,w_cw_\ell)$-LDPC. In the above example where each element is a monomial (so $w_\ell=1$), \eqref{eq:H_lift} maps to a $(3,3)$-LDPC matrix of size $15\times20$ and is a parity-check matrix of a $[20,5,9]$ code. One of the benefits of the QC-LDPC construction is that the rate can be lower-bounded by the dimensions of the protomatrix, and the length can be adjusted by tuning the lift size.

\begin{table*}[t]
\renewcommand{\arraystretch}{1.35}
\begin{tabular}{c|c|c|c|ccc|ccc}
    \textbf{Input code}\; & $H_{\rm lift}$ &
    \;\,$\ell$\;\, & \;\,$\chi$\;\, &
    \textbf{HGP} &
    $N_{\rm 2q}$ &
    \;$\big({w}_q, {w}_c\big)$\; &
    \;\textbf{Reduced HGP}\; &
    $\tilde{N}_{\rm 2q}$ &
    $\big(\tilde{w}_q, \tilde{w}_c\big)$ \\
    \hline
    $[20, 5, 9]$ & \;$\begin{pmatrix} x^4&0&x^4&x^3 \\ 0&x^3&x^3&x^4 \\ x^3&x^4&0&x^3 \end{pmatrix}$\; & $5$ & $3$ & $\llbracket 625, 25, 9\rrbracket$ & $3150$ & $(3,6)$ & $\llbracket 475, 25, 9\rrbracket$ & $2775$ & $(5,9)$ \\
    \hline
    $[24, 6, 10]$ & $\begin{pmatrix} x^5&0&x^3&x^3 \\ 0&x^4&x^2&x \\ x^2&x&0&x \end{pmatrix}$ & $6$ & $3$ & $\llbracket 900, 36, 10\rrbracket$ & $4536$ & $(3,6)$ & $\llbracket 684, 36, 10\rrbracket$ & $3996$ & $(5,9)$ \\
    \hline
    $[28, 7, 11]$ & $\begin{pmatrix} x&0&x^2&x^3 \\ 0&x^5&x^6&x \\ x^4&x^5&0&x^5 \end{pmatrix}$ & $7$ & $3$ & \;\;$\llbracket 1225, 49, 11\rrbracket$\;\; & \;$6174$\; & $(3,6)$ & $\llbracket 931, 49, 11\rrbracket$ & $5439$ & $(5,9)$
\end{tabular}
\caption{Examples of type-I QC-LDPC codes, their corresponding HGP codes, and their reduced HGP variants. All three listed examples are built from the same protomatrix \eqref{eq:H_proto}. $\chi$ is the number of check colors in the input code. $N_{\rm 2q}$ ($\tilde{N}_{\rm 2q}$) is the number of 2-qubit gates required for syndrome extraction before (after) the reduction.}
\label{tab:QC-LDPC HGP}
\end{table*}

\begin{defn}[Type-I QC-LDPC~\cite{Smarandache_2004}]
    A QC-LDPC matrix $H_{\rm lift} \in \mathsf{R}^{m_{\rm p}\times n_{\rm p}}_\ell$ is of type I if all nonzero entries are monomials.
\end{defn}

The previous example \eqref{eq:H_lift} is of type I. We now show that a check-coloring of $H_{\rm full}$ can be derived from that of $H_{\rm proto}$ when $H_{\rm lift}$ is of type I.

\begin{prop}[Type-I QC-LDPC check-coloring] \label{prop:lifted coloring}
    For a type-I QC-LDPC code with protomatrix $H_{\rm proto}$ and full parity-check matrix $H_{\rm full}$, if the checks of $H_{\rm proto}$ can be $\chi$-colored, then the checks of $H_{\rm full}$ can be $\chi'$-colored with $\chi'\leq\chi$.
\end{prop}

\begin{proof}
    We will prove the statement of the proposition by constructing a valid coloring for all the lifted checks of $H_{\rm full}$. The Tanner graph of $H_{\rm full}$ can be arranged as $\ell$ stacked copies of $H_{\rm proto}$ with intercopy edges connected according to the $\mathsf{R}_\ell$ elements in $H_{\rm lift}$, which are monomials from the type-I assumption. Index the lifted bit and check nodes by $(b,t)$ and $(c,t)$ respectively, where the $b$ and $c$ ``row coordinates'' index the nodes within each copy of $H_{\rm proto}$, and the $t$ ``column coordinate'' indexes the different copies of $H_{\rm proto}$. With this convention, there exists an edge $(c,t)\sim(b,t+a_{cb})$ whenever $(H_{\rm lift})_{cb} = x^{a_{cb}}$.
    
    We will first show that all the lifted checks with a fixed $c$ coordinate, i.e. $(c,\cdot)$, do not share any lifted bits in common and can thus take the same color. For the $c$th row in $H_{\rm lift}$, each monomial labels how its corresponding $t$ edges are connected. Since each monomial corresponds to a shift permutation whose $\ell\times\ell$ matrix column weight is 1, no checks along a row $(c,\cdot)$ will share any bits in common. Accordingly, all these lifted checks can be the same color.

    Now, define the projection $\pi:(c,t)\rightarrow c$, $(b,t)\rightarrow b$. If two lifted checks $(c,t)$ and $(c',t')$ share a lifted bit $(b,t'')$, then their projections $c$ and $c'$ will share the same projected bit $b$. The contrapositive then implies that if two checks in $H_{\rm proto}$ are nonadjacent, then all of their lifted copies will also be nonadjacent in $H_{\rm full}$. This statement, in combination with the previous paragraph, then implies that the following $\chi'$-coloring is a valid coloring for $H_{\rm full}$. For a given $\chi$-coloring of $H_{\rm proto}$, take each check $c$ and color its lifted copies $(c,\cdot)$ with the same color. By construction, we have $\chi'=\chi$.
\end{proof}

Proposition~\ref{prop:lifted coloring} is particularly useful for us because it allows us to fix a desired check chromatic number $\chi$ and LDPC weights $(w_b,w_c)$ by carefully choosing a small $H_{\rm proto}$. We can then obtain larger code lengths and distances by adjusting the lift size $\ell$ and the monomials in $H_{\rm lift}$. See Table \ref{tab:QC-LDPC HGP} for three examples of type-I QC-LDPC codes with the same protomatrix. Unfortunately, it is known that a type-I QC-LDPC with protomatrix size $m_{\rm p}\times n_{\rm p}$ has a minimum distance upper-bounded by $d \leq (m_{\rm p}+1)!$ \cite[Corollary 9]{Smarandache_2012}. However, for our examples with $m_{\rm p}=3$, this upper bound is $4!=24$, which is likely enough for the practical regime. We also note that the proof of Proposition~\ref{prop:lifted coloring} applies to more generic (e.g. non-abelian) lifts with $\ell\times\ell$ permutation matrices that need not be circulant, which is a way to circumvent the upper bound.

\subsection{Bipartite cycle codes}

Given a $\Delta$-regular, bipartite, simple graph $\mathpzc{G}=(V,E)$, its cycle code $\mathcal{C}(\mathpzc{G})$ is defined as the linear code with bits on edges and a parity check on every vertex. Codewords are spanned by closed paths or cycles on $\mathpzc{G}$, since a closed path in a graph touches every vertex an even number of times and thus satisfies all parity checks. It is also known as the Tanner code on $\mathpzc{G}$ whose local vertex code is the single-parity-check code; i.e. where the parity matrix is
\begin{equation}
    \hat{h} = \overbrace{\begin{pmatrix} 1 & \overset{1}{\cdots} & 1 \end{pmatrix}}^{\Delta}\, .
\end{equation}
For a connected graph, the cycle-code dimension is $|E|-|V|+1$ and its codewords are the even subgraphs. The special case in which the graph is a single cycle gives a repetition code on its edges.

We will restrict ourselves to connected graphs $\mathpzc{G}$, where the transpose code is also known as the repetition code on $\mathcal{G}$. As such, the sum of all vertices is zero, and we can remove this linear dependency by dropping one vertex or parity check from the code. We will also be interested in bipartite graphs since they are inherently two-colorable, and so we will be able to remove all \Csectorname{} qubits in the corresponding HGP code. Note that we can convert any simple graph into a bipartite graph of twice the size by considering its bipartite double cover; see Appendix \ref{app:double cover} for more details.

\begin{table*}[t]
\renewcommand{\arraystretch}{1.35}
\begin{tabular}{c|c|c|c|ccc|ccc}
\textbf{Graph}\; & \;$|V|$\; & \;$\Delta$\; & \;\textbf{Cycle code}\; & \textbf{HGP}                     & $N_{\rm 2q}$ & $\big({w}_q, {w}_c\big)_X$ & \;\textbf{Reduced HGP}             & $\tilde{N}_{\rm 2q}$ & $\big(\tilde{{w}}_q, \tilde{{w}}_c\big)_X$ \\ \hline
$K_{3,3}$  & 6 & 3    & $[9,4,4]$                  & \;$\llbracket 106,16,4 \rrbracket$\; & 420          & $(3,5)$                 & $\llbracket 81,16,4 \rrbracket$  & 345           & $(4,6)$                   \\
Heawood    & 14 & 3        & $[21,8,6]$                 & $\llbracket 610,64,6 \rrbracket$ & \;2652\;         & \;$(3,5)$\;                 & $\llbracket 441,64,6 \rrbracket$ & 2145          & \;$(4,6)$\;           \\
\;Tutte–Coxeter\;    & 30 & 3        & $[45,16,8]$                 & \;$\llbracket 2866,256,8 \rrbracket$\; & \;12876\;         & \;$(3,5)$\;                 & $\llbracket 2025,256,8 \rrbracket$ & 10353          & \;$(4,6)$\;     
\end{tabular}
\caption{Example bipartite cycle codes, their corresponding HGP codes, and their reduced HGP variants. $|V|$ is the number of vertices of the underlying (regular) bipartite graph, and $\Delta$ is the vertex degree. $N_{\rm 2q}$ ($\tilde{N}_{\rm 2q}$) is the number of 2-qubit gates required for syndrome extraction before (after) the reduction. The three graphs are the $(3,4)$-, $(3,6)$-, and $(3,8)$-cages, which means they are the minimum sizes for their respective girths.}
\label{tab:cycle code HGP}
\end{table*}

For biregular graphs, which are bipartite graphs with regular left and right degrees, there exists a relatively simple form for the HGP code parameters prior to and following the reduction in Section \ref{sec:procedure}. Suppose $\mathpzc{G}=(V_\ell \cup V_r, E)$ is connected and biregular with uniform left and right degree $\Delta$ and $\abs{V_\ell} = \abs{V_r} = \abs{V}/2 \equiv v$. The number of edges is $\abs{E} = \Delta V_\ell = \Delta V_r$. With bits on edges and checks on vertices, the parity-check matrix has dimensions $\abs{V}\times\abs{E}$. To make it full-rank, we remove one vertex so that $H \in \F^{(\abs{V}-1)\times\abs{E}}_2$. Upon taking the HGP of $H$ with itself, the HGP code parameters are
\begin{subequations}
\begin{align}
    n &= (v\Delta)^2 + (2v-1)^2  \\
    k &= (v(\Delta-2) + 1)^2  \\
    d &= \operatorname{girth}(\mathpzc{G}) \, ,
\end{align}
\end{subequations}
where $\operatorname{girth}(\mathpzc{G})$ is the length of the shortest cycle of $\mathpzc{G}$ and obeys the Moore bound $\operatorname{girth}(\mathpzc{G}) \leq 2\log_{\Delta-1}\abs{V} + O(1)$~\cite{biggs1993}, which implies that regular cycle codes have $d = O(\log n)$ when the graph is not a ring ($\Delta>2$). The maximum check weight is $w_c = \Delta+2$. The number of two-qubit gates required for $X$-syndrome extraction is simply the number of nonzero entries in $H_X$, which is $N^{(X)}_{\rm 2q} = \Delta(2v-1)(v\Delta+2v-1)$.

After qubit reduction, the HGP code has parameters
\begin{subequations}
\begin{align}
    \tilde{n} &= (v\Delta)^2  \\
    \tilde{k} &= k = \big[v(\Delta-2)+1\big]^2 \\
    \tilde{d} &= d = \operatorname{girth}(\mathpzc{G}) \, ,
\end{align}
\end{subequations}
and the maximum check weight is now $\tilde{w}_c = 2\Delta$, an increase of $\tilde{w}_c - w_c = \Delta-2$ compared to the unmodified HGP code. Removing all \Csectorname{} qubits according to the conventions of Figure \ref{fig:trans}, the new number of two-qubit gates required for $X$-syndrome extraction can be exactly computed to be
\begin{align}
    \tilde{N}^{(X)}_{\rm 2q} = 2v^2\Delta(\Delta-1) + 2(v-1)^2\Delta(\Delta-1) + (v-1)\Delta^2 \, ,
\end{align}
which results in a net difference of
\begin{align}\label{eq:X 2q savings}
    N^{(X)}_{\rm 2q} - \tilde{N}^{(X)}_{\rm 2q} = \frac{\Delta}{2}\Big[ (4-\Delta)(2v-1)^2 - (\Delta-2) \Big] \, .
\end{align}
A similar calculation for the $Z$-sector yields
\begin{align}\label{eq:Z 2q savings}
    N^{(Z)}_{\rm 2q} - \tilde{N}^{(Z)}_{\rm 2q} = \Delta\Big[ 1+2(4-\Delta)v(v-1) \Big] \, .
\end{align}
The total two-qubit gate savings is then given by summing \eqref{eq:X 2q savings} and \eqref{eq:Z 2q savings}:
\begin{align}\label{eq:2q savings}
    N_{\rm 2q} - \tilde{N}_{\rm 2q} = \Delta(4-\Delta)(2v-1)^2 \, ,
\end{align}
which is positive when $\Delta\leq3$, zero when $\Delta=4$, and negative when $\Delta\geq5$.

We now provide two examples of reduced HGP codes based on classical cycle codes of bipartite graphs. Since the underlying graphs are bipartite, the qubit reduction procedure of Section~\ref{sec:procedure} will eliminate all \Csectorname{} qubits as described in Section~\ref{qTanner transformation}. The code parameters are summarized in Table~\ref{tab:cycle code HGP}. The resulting reduced HGP codes can also be viewed as special instances of quantum Tanner codes~\cite{qTanner_codes, Leverrier_2025_efficient}. The number of two-qubit gates $N_{\rm 2q}$ is given by the sum of the numbers of edges in the $X$ and $Z$ Tanner graphs.


Suppose we define both $\mathcal{C}_1$, $\mathcal{C}_2$ on the complete, bipartite graph $K_{3, 3}$ using single-parity check local codes with parity-check matrices $\hat{h}_1 = \hat{h}_2 = \hat{h} = \begin{pmatrix} 1 & 1 & 1 \end{pmatrix}$. From this structure, the parity-check matrices take the form

\begin{align}
    \begin{pmatrix}
        \hat{h} \otimes \ident_3 \\
        \ident_3 \otimes \hat{h}
    \end{pmatrix} \, . \label{K33pcm_gen}
\end{align}

To make the code non-redundant (i.e., $H_1$, $H_2$ full rank), we remove the last row and get
\begin{align}
    H_1 = H_2 = \begin{pmatrix}
        1 & 1 & 1 & 0 & 0 & 0 & 0 & 0 & 0 \\ 
        0 & 0 & 0 & 1 & 1 & 1 & 0 & 0 & 0 \\
        0 & 0 & 0 & 0 & 0 & 0 & 1 & 1 & 1 \\
        1 & 0 & 0 & 1 & 0 & 0 & 1 & 0 & 0 \\ 
        0 & 1 & 0 & 0 & 1 & 0 & 0 & 1 & 0         
    \end{pmatrix} \, . \label{K33pcm}
\end{align}

The above parity-check matrix corresponds to a $[9,4,4]$ $(2,3)$-LDPC code. Taking the HGP of \eqref{K33pcm} with itself, we obtain a $\llbracket 106,16,4 \rrbracket$ $(3,5)$-LDPC HGP code with $kd^2/n \approx 2.4$. After qubit reduction, we obtain a $\llbracket 81,16,4 \rrbracket$ $(4,6)$-LDPC reduced HGP (quantum Tanner) code with $kd^2/n \approx 3.2$.

Similarly, if we defined the classical input codes to be the cycle code of the 3-regular, bipartite Heawood graph  with 14 vertices and 21 edges, after removing one vertex check, we obtain a $[21,8,6]$ $(2,3)$-LDPC code. Upon taking the HGP of this code with itself, we obtain a $\llbracket 610,64,6 \rrbracket$ $(3,5)$-LDPC HGP code with $kd^2/n \approx 3.8$. After qubit reduction, we arrive at a $\llbracket 441,64,6 \rrbracket$ $(4,6)$-LDPC reduced HGP (quantum Tanner) code with $kd^2/n \approx 5.2$.

Again, for both these examples, we only ever combine at most two checks at a time, and so both of these reduced HGP codes admit distance-preserving syndrome extraction schedules following the discussion of Section \ref{sec:hook errors}.


\section{Numerical simulations}\label{sec:numerics}

To assess the near-term prospects of our qubit reduction scheme in practice, we perform circuit-level $Z$-memory simulations on a few select code instances. For a given $\llbracket n,k,d \rrbracket$ code, we initialize all data qubits in $\ket{0}^{\otimes n}$, perform $d$ rounds of ($X$ and $Z$) syndrome extraction, followed by a transversal $Z$-measurement of all data qubits. A final $Z$-syndrome and the logical $\bar{Z}$-measurements can be inferred from this final transversal $Z$-measurement. The $d+1$ rounds of $Z$-syndromes are then processed by a classical decoder, which in turn predicts which logical $\bar{Z}$-measurements should flip at the end. If the predicted $\bar{Z}$-flips do not match the observed $\bar{Z}$-flips, then a logical failure is declared.

\begin{figure*}[t]
    \centering
    \includegraphics[width=0.49\textwidth]{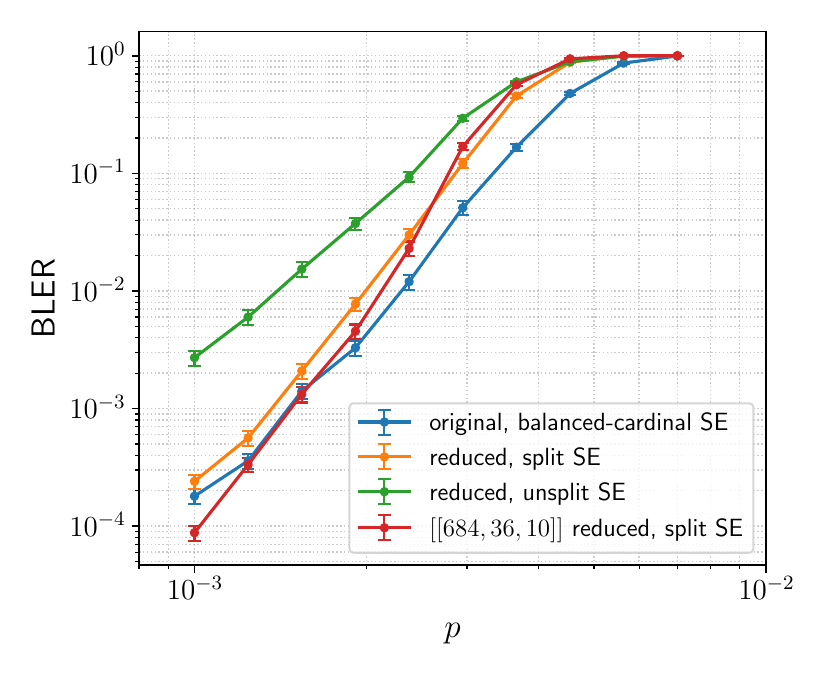} \hfill
    \includegraphics[width=0.49\textwidth]{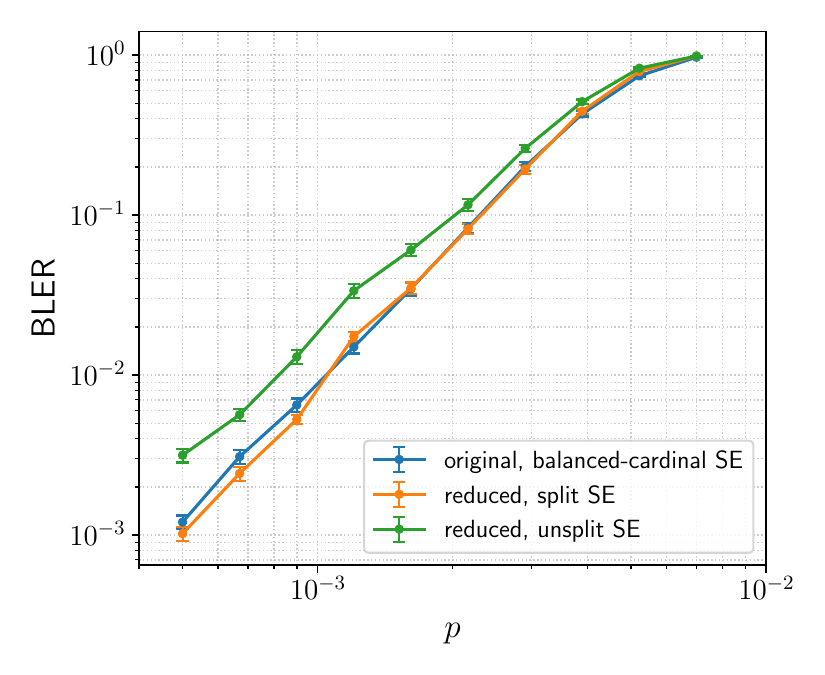}
    \caption{Block logical error rate (BLER) as a function of the noise strength $p$. The unreduced codes use depth-$8$ balanced-cardinal schedules, while the reduced codes use either an unsplit schedule or the split schedule of Algorithm~\ref{alg:split SE}. The left panel uses the $[20,5,9]$ QC-LDPC input code and also includes the reduced split schedule for the $[24,6,10]$ input code. The right panel uses the $[21,8,6]$ Heawood cycle code. Shot counts are increased at low BLER until the relative binomial standard error is at most $15\%$ for the QC-LDPC data and $10\%$ for the Heawood data. Error bars show one binomial standard error.}
    \label{fig:numerics}
\end{figure*}

For our noise model, we consider circuit-level depolarizing noise with strength $p$:
\begin{itemize}
    \item Single-qubit gates, including idling, depolarize with probability $p/10$.
    \item Two-qubit gates depolarize with probability $p$.
    \item State preparation and measurement (SPAM) flip with probability $p$.
\end{itemize}
The above noise model is chosen to closely resemble what is benchmarked in current-generation trapped-ion~\cite{Quantinuum_Helios} and neutral-atom processors~\cite{QuEra_MSD_2025}, which are two platforms capable of realizing the long-range connectivity required for high-rate HGP codes.

We choose two different HGP code examples to simulate. The input of the first HGP code is the $[20,5,9]$ QC-LDPC code from Table \ref{tab:QC-LDPC HGP}; note that the parity-check matrix listed already has full rank. The input of the second HGP code is the $[21,8,6]$ Heawood cycle code from Table \ref{tab:cycle code HGP}; we remove one row from the incidence matrix of the Heawood graph to make it full rank. For each code, we simulate both the original, unreduced HGP code as well as its reduced version. For syndrome extraction, we initialize an ancilla qubit for each parity check, couple the ancilla qubits to corresponding data qubit using CNOT gates, and finally measure the ancilla. For the unreduced code, we use the balanced-cardinal schedule of QUITS~\cite{Kang_2025_QUITS}. This schedule interleaves the $X$- and $Z$-check CNOTs in four directional blocks and has total CNOT depth 8. For the reduced code, we implement both a fixed unsplit schedule based on a minimum edge coloring and the distance-preserving split-SE scheduling of Algorithm \ref{alg:split SE}. The (Clifford) SE circuits as well as the corresponding depolarizing noise model are all performed using the \textsf{Stim} Python package~\cite{Stim}.

For one full syndrome-extraction round on the $[20,5,9]$ QC-LDPC HGP instance, the original code has total CNOT depth 12 when the $X$- and $Z$-checks are extracted separately and depth 8 with the balanced-cardinal schedule. The unsplit and split schedules on the reduced code have depths 18 and 22, respectively. For the $[21,8,6]$ Heawood HGP instance, the corresponding original-code depths are 10 and 8, while both reduced-code schedules have depth 12. The combined $X/Z$ Tanner graph of each original code has maximum degree 6, so the depth-8 balanced-cardinal schedules are within two layers of the degree lower bound. The reduced-code lower bounds are 9 and 8 for the QC-LDPC and Heawood examples, respectively.

To decode the $d+1$ rounds of $Z$-syndromes, we construct a spacetime phenomenological decoding graph by taking the HGP of the $Z$-Tanner graph with a length-($d+1$) repetition code. We use this simpler graph rather than the full circuit-level detector graph to reduce the decoding cost. Detector samples are still drawn from the full noisy syndrome-extraction circuit. For each circuit and value of $p$, we obtain its detector error model from \textsf{Stim} and project each error mechanism onto the phenomenological graph while preserving its detector syndrome and action on the canonical logical observables. The prior for each phenomenological variable is the probability that it is flipped by an odd number of these projected mechanisms. Relay-BP treats the variables independently, so the priors retain the circuit- and schedule-dependent marginal error rates but omit correlations among the projected variables. The curves therefore benchmark the chosen circuit--decoder pairs rather than optimal decoder performance.
For the decoder, we use the \textsf{Relay-BP} Rust/Python package~\cite{RelayBP_Github}, which is based off of the Relay-BP algorithm~\cite{RelayBP}. Given a parity-check (or detector) matrix $H\in\mathbb{F}^{M\times N}_2$ and a syndrome vector $\mathbf{s}=H\mathbf{e} \in \mathbb{F}^M_2$ for some error $\mathbf{e}\in\mathbb{F}^N_2$, the decoder tries to infer a correction $\hat{\mathbf{e}}\in\mathbb{F}^N_2$ such that $\mathbf{s}=H\hat{\mathbf{e}}$. In the memory simulations, decoding is successful if the logical-observable flips inferred from $\hat{\mathbf{e}}$ match those sampled by \textsf{Stim}. Belief propagation (BP) is an efficient marginalization algorithm that can be used to estimate the most likely error by passing messages along the Tanner graph associated with $H$~\cite{MacKay_1999}. However, despite its excellent performance for classical LDPC codes, BP behaves poorly for quantum LDPC codes due to short cycles and degeneracy, often resulting in an invalid correction~\cite{Murphy_1999, Roffe_2020}. Ordered-statistics (OS) postprocessing resolves the aforementioned pitfalls by forcing a valid correction at the expense of costly matrix inversion~\cite{LP_codes}. More recently, memory effects have been proposed as a fast way to converge BP without relying on matrix inversion~\cite{Kuo_2022, Chen_2025}. The Relay-BP algorithm chains consecutive BP runs with disordered memory strengths and has been shown to beat BP+OSD in both accuracy and runtime for certain families of LDPC codes and parameter regimes~\cite{RelayBP}. For our relay parameters, we used the configuration: \{gamma0=0.65, pre\_iter=80, num\_sets=100, set\_max\_iter=60, gamma\_dist\_interval=(-0.24, 0.66), stop\_nconv=5\}. The same Relay-BP hyperparameters are used for every circuit and value of $p$.


The results are plotted in Figure~\ref{fig:numerics}. For a decoder that corrects every fault set of size at most $\lfloor(d_{\rm circ}-1)/2\rfloor$, the first possible uncorrectable contribution occurs at order $p^{\lceil d_{\rm circ}/2\rceil}$~\cite{Bravyi_2013_sim}. This repeated-circuit fault distance is related, and often equal to \cite{shaw2026optimising}, the effective distance established by Theorem~\ref{thm:d_circ preservation}, which concerns propagated syndrome hook errors in one syndrome-extraction circuit. The present simulations therefore benchmark finite repeated circuits using a phenomenological graph with independent DEM-derived priors and do not numerically verify that theorem. We do not fit slopes or thresholds from these finite-size data.

At low and intermediate noise strengths, where the BLERs remain away from one, the split schedule gives a lower BLER estimate than the unsplit schedule for both code instances. For the $[20,5,9]$ QC-LDPC instance, the split and unsplit depths are $22$ and $18$, respectively; the split curve is lower than the unsplit curve, while the depth-$8$ balanced-cardinal curve remains lower than the split curve. Thus, circuit depth alone does not determine the observed BLER. For the Heawood instance, both reduced schedules have depth $12$, but the split curve is lower than the unsplit curve and closely follows the balanced-cardinal curve at low and intermediate $p$.

So far we have only compared the unreduced HGP code with its reduced version directly. Another comparison we can do is to (roughly) fix the number of data qubits and ask if there is a reduced HGP code with better performance than an unreduced code. In the left plot of Figure \ref{fig:numerics}, we also plot a $\llbracket 684,36,10 \rrbracket$ reduced HGP code with split-SE scheduling from the $[24,6,10]$ QC-LDPC input code in Table \ref{tab:QC-LDPC HGP}, which has slightly more data qubits than the unreduced $\llbracket 625,25,9 \rrbracket$ HGP code but encodes more logical qubits with a slightly larger minimum distance. The data-qubit counts differ by about $9\%$, and the blocks encode $36$ and $25$ logical qubits, so their BLERs compare different sizes and rates. At the three smallest sampled noise strengths, the larger reduced block has the lower BLER estimate, with nonoverlapping one-standard-error intervals at $p=10^{-3}$. The point-estimate curves cross between the sampled values $1.54\times10^{-3}$ and $1.91\times10^{-3}$, above which the smaller unreduced block has the lower BLER estimate. This finite-size crossover hence showcases a tradeoff among data-qubit count, code rate, distance, and circuit cost at these code lengths.


\section{Outlook}\label{sec:outlook}

We introduced a general-purpose procedure to reduce the number of physical qubits in HGP codes while maintaining code parameters such as code dimension, logical operator basis and minimum distance. We also show compatibility with a variety of existing HGP fault-tolerant gadgets such as distance-preserving syndrome extraction, fold-transversal gates and code homomorphisms. We then demonstrate that carefully-structured input codes allow the most effective reductions both in terms of qubit overhead, degrees, and circuit-level tradeoffs.

There are several possible extensions of the qubit reduction that would promote its practicality for fault-tolerant quantum computation. On the numerical front, an improvement in decoding would directly lower the overhead for fault-tolerant quantum computation with these HGP codes by enabling the usage of smaller code blocks. There is room for significant improvement because the decoder \emph{(i)} uses a phenomenological graph with independent DEM-derived priors rather than retaining the full circuit-level correlations, and \emph{(ii)} was not fine-tuned for the decoding graph which typically leads to better performance \cite{RelayBP, wills2026concatenating}. On the theoretical front, the most immediate extension would be to generalize our procedure to higher-dimensional ($D>2$) HGP codes; i.e. products of more than 2 classical codes. $D>2$ HGP codes can support single-shot QEC/initialization~\cite{Quintavalle_2021_singleshot} as well as transversal\footnote{Here ``transversal'' is a bit looser from the usual sense and means a constant-depth circuit.} non-Clifford gates~\cite{breuckmann2025cups, golowich_2025_ccz, lin2024sheaf, hsin2025non, zhu2025fountain, tan2025ccz}, and so they have been of both theoretical and practical interest. Since the qubit reduction consists only of local modifications, we also expect it to preserve the confinement~\cite{Quintavalle_2021_singleshot} and soundness~\cite{Campbell_2019} properties relevant for single-shot error correction and initialization, but leave formal verifications to future work.

There have been several recent works that introduce new HGP fault-tolerant gadgets: two based on code-switching~\cite{golowich2025switch, tan2025ccz} and two based on code surgery~\cite{zheng2025high, chang2026constant}. Another work \cite{Gu_2026_QGPU} develops a structured logical basis for lifted-product codes and a parallel product-surgery protocol that also applies to HGP codes. It would be interesting to determine which of these parallel operations survive a compatible qubit reduction. It would be of further practical relevance if our qubit reduction can be made compatible with a subset or all of these schemes. In~\cite{golowich2025switch}, addressable logical Clifford gates are performed in constant depth by code-switching between 2D and 3D HGP codes along different axes, interleaved with transversal Clifford gates. Extending the 2D qubit reduction to 3D would be the first step towards progress in this direction, then constructing chain maps between reduced 2D and 3D HGP codes to enable code-switching. In~\cite{tan2025ccz}, code-switching between structured 2D and 3D HGP codes enable logical CCZ gates in constant depth. To make qubit reduction compatible with this scheme would require finding a reduction schedule compatible with the algebraic structures that enable the transversal CCZ in the 3D code. In~\cite{zheng2025high} and~\cite{chang2026constant}, specialized HGP ancilla blocks are attached and detached to simultaneously measure extensive numbers of logical Paulis; in the first case different logical Paulis for high-rate surgery and in the second case stabilizer-equivalent logical Paulis for single-shot surgery. There are some constant-factor qubit savings when using these ``coning'' approaches compared to the homomorphic measurement scheme. Since the chain maps for coning are slightly different than for homomorphic measurements, the main challenge for showing reduction compatibility would be to prove distance-preservation for the merged data-ancilla code.

Recently, HGP codes have also been shown to achieve optimal most-likely-error (MLE) decoding under more realistic neutral-atom noise models that include atom loss, reaching the same decoding radius as that for erasures~\cite{liu2026loss}. An extension to reduced HGP codes would boost their practicality for the neutral-atom platform.

Lastly, it would be interesting if any analytical guarantees for HGP qubit reduction can be extended to more exotic product constructions such as the lifted product~\cite{LP_codes} and balanced product~\cite{Breuckmann_2021_balanced}. These constructions allow for codes with even more efficient parameters than the HGP~\cite{2BGA_codes, BB_codes, cain2026shor}. Advancing in this direction could include generalizing the K\"unneth theorem to tensor products of chain complexes based on group rings \cite{kunneth_lp}, over which the lifted and (certain) balanced products are defined.


\section*{Acknowledgments}

We thank Shi Jie Samuel Tan, Qian Xu and Guo Zheng for helpful discussions and feedback on the draft.
This material is based upon work supported in part by the Defense Advanced Research Projects Agency (DARPA) under Agreement HR00112490357, the NSF QLCI award OMA-2120757 (RQS), the NSF-funded NQVL:QSTD: Pilot: DLPQC and the DoE ASCR Quantum Testbed Pathfinder program (awards No.~DE-SC0019040 and No.~DE-SC0024220).
We also thank Victor Albert, William Gasarch and the REU-CAAR program at the University of Maryland, College Park for support during the creation and development of this project. We acknowledge the use of AI for technical discussions and some assistance in prose.


\section*{Code availability}

A Python implementation of the qubit-reduction procedure in Section~\ref{sec:procedure} and the numerical experiments done in Section~\ref{sec:numerics} are available at the following \href{https://github.com/apabla1/hgp-reduction}{GitHub repository}.


\bibliography{thebib}

\include{appendix}

\end{document}

%% file: preamble.tex
\usepackage[utf8]{inputenc}
\usepackage[T1]{fontenc}
\usepackage{graphicx}
\usepackage[usenames,dvipsnames]{xcolor}
\usepackage{mathrsfs}
\usepackage{amsmath,amssymb,graphics,amsthm}
\usepackage{bbm}
\usepackage{bm}
\usepackage{braket}
\usepackage{stmaryrd}
\usepackage{comment}
\usepackage{enumitem}
\usepackage[normalem]{ulem}

\usepackage[colorlinks=true, urlcolor=violet, linkcolor=blue, citecolor=red, hyperindex=true, linktocpage=true]{hyperref}
\allowdisplaybreaks

\newtheorem{thm}{Theorem}
\numberwithin{thm}{section}

\newtheorem{lem}[thm]{Lemma}

\newtheorem{prop}[thm]{Proposition}
\newtheorem*{prop*}{Proposition}
\newtheorem{eg}[thm]{Example}
\newtheorem{claim}[thm]{Claim}

\newtheorem{defn}[thm]{Definition}

\newtheorem{fact}[thm]{Fact}

\renewcommand{\thesection}{\arabic{section}}
\renewcommand{\thesubsection}{\thesection.\arabic{subsection}}

\makeatletter
\renewcommand{\p@subsection}{}
\renewcommand{\p@subsubsection}{}
\makeatother

\usepackage{xcolor}
\usepackage{mathtools}

\newcommand{\mbf}[1]{\mathbf{#1}}

\DeclareMathAlphabet{\mathpzc}{OT1}{pzc}{m}{it}

\providecommand{\revadd}[1]{}

\newcommand{\Bsectorname}{bit-type}
\newcommand{\Csectorname}{check-type}

\newcommand{\norm}[1]{\left\lVert#1\right\rVert}

\newcommand{\im}{\operatorname{im}}

\newcommand{\rank}{\operatorname{rank}}
\newcommand{\rs}{\operatorname{rs}}

\newcommand{\abs}[1]{\lvert{#1}\rvert}

\newcommand{\F}{\mathbb{F}}
\newcommand{\transpose}{\mathsf{T}}
\newcommand{\transp}{\mathsf{T}}

\newcommand{\ident}[0]{\mathds{1}}

\usepackage{dsfont}

\usepackage{quiver}

\usepackage[ruled,linesnumbered]{algorithm2e}

\SetCommentSty{mycommfont}
\SetKwInOut{Input}{Input}
\SetKwInOut{Return}{Return}

%% file: appendix.tex
\appendix
\renewcommand{\thesubsection}{\thesection.\arabic{subsection}}

\section{Chain complex formalism for classical and quantum codes}\label{app:codes algebraic}

In this section, we cover the algebraic formalism of classical and quantum codes which gives a connection between the linear algebraic perspective of the codes to a pictorial perspective using Tanner graphs. This section will also give the background necessary for the results of Section~\ref{sec:GPPM}.

\begin{defn}[Chain complex]\label{def:chain complex}\label{2-term complex}\label{3-term complex}
    A (based) chain complex over $\F_2$ is a sequence of vector spaces $\{A_i\}_{i\in\mathbb{Z}}$ together with linear boundary maps
    \[
        \partial_i : A_i \to A_{i-1}
    \]
    such that for every $i$,
    \begin{align}
        \partial_{i-1}\circ \partial_i = 0\, . \label{boundarycondition}
    \end{align}
    Equivalently, a chain complex is written as
    \begin{align}
        \cdots \xleftarrow{\partial_0} A_0 \xleftarrow{\partial_1} A_1 \xleftarrow{\partial_2} A_2 \xleftarrow{\partial_3} \cdots \, . \label{eq:chain complex}
    \end{align}
    The $i$th homology group is defined as $\mathcal{H}_i := \ker\partial_i/\im\partial_{i+1}$.
    The associated cochain complex on the dual spaces $\{A_i^*\}_{i\in\mathbb{Z}}$ is defined by the coboundary maps
    \[
        \delta^i := \partial_{i+1}^\transp : A_i^* \to A_{i+1}^* \, ,
    \]
    so that
    \begin{align}
        \cdots \xrightarrow{\delta^{-2}} A_{-1}^* \xrightarrow{\delta^{-1}} A_0^* \xrightarrow{\delta^0} A_1^* \xrightarrow{\delta^1} A_2^* \xrightarrow{\delta^2} \cdots \, . \label{eq:cochain complex}
    \end{align}
    Since $\delta^i = \partial_{i+1}^\transp$, the cochain condition is equivalently
    \[
        \delta^{i+1}\circ \delta^i = 0
    \]
    for every $i$. The $i$th cohomology group is defined as $\mathcal{H}^i := \ker\delta^i/\im\delta^{i-1}$.
\end{defn}

In this work, we focus on 2-term and 3-term chain complexes:
\begin{align}
    A_0 \xleftarrow{\partial_1} A_1 \, , \label{eq: 2-complex}
\end{align}
\begin{align}
    A_0 \xleftarrow{\partial_1} A_1 \xleftarrow{\partial_2} A_2 \, , \label{3-term comp}
\end{align}
with cochain complexes
\begin{align}
    A_0^* \xrightarrow{\delta^0} A_1^* \, , \label{eq: 2-complex transpose}
\end{align}
\begin{align}
    A_0^* \xrightarrow{\delta^0} A_1^* \xrightarrow{\delta^1} A_2^* \, . \label{3-term comp dual}
\end{align}

Geometrically, when $A_i=\F_2[X_i]$ is generated by a basis of $i$-cells $X_i$, the map $\partial_i$ sends a formal sum of $i$-cells to its boundary in $\F_2[X_{i-1}]$. In particular, the 2-term case corresponds to edges and vertices, while the 3-term case corresponds to faces, edges, and vertices; see Figure \ref{fig:app-3complex} for an illustration.

\begin{figure}[t]
    \centering
    \includegraphics[width=1\linewidth]{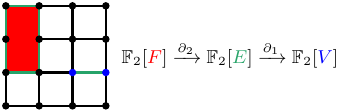}
    \caption{Geometric picture of a 3-complex as a mapping between faces to edges to 0-dimensional vertices in a 2-dimensional manifold. The generating set of faces $F$ are the nine faces shown, the generating set of edges $E$ are the twenty-four edges shown, and the generating set of vertices $V$ are the sixteen vertices shown. The highlighted face in red is the sum of the top-right and middle-right faces over $\F_2$. We denote this face as $f \in \langle F \rangle$. Applying the boundary map $\partial_2$, one gets that face's boundary -- the surrounding edges in green, which we denote $e_1 \in \langle E \rangle$. Applying the boundary map $\partial_1$ to the blue edge $e_2 \in \langle E\rangle$ on the right side, one gets that edge's boundary -- the highlighted vertices in blue, which we denote $v \in \langle V \rangle$. As equations, we have $\partial_2 f = e_1$ and $\partial_1 e_2 = v$. Note that the constraint $\partial_1 \circ \partial_2 = 0$ is satisfied here because the sum of the boundaries of edges surrounding a face in $\langle V \rangle$ is zero, because every vertex cancels each other out. In other words, the boundary of a (sum of) face(s) (a closed loop of edges) has no boundary of vertices. Indeed, $e_1$ has no boundary because summing the boundary of each green vertex cancels each other out, so $\partial_1 \partial_2 f = 0$.}
    \label{fig:app-3complex}
\end{figure}

\begin{defn}[Classical code 2-term chain complex]
    An $[n, k, d]$ classical linear code $\mathcal{C}$ is a $k$-dimensional subspace of bits $B=\F_2^n$ with minimum Hamming distance $d$. Let $H\in \F_2^{m\times n}$ be the parity-check matrix for $\mathcal{C}$ so that the code is defined as
    \begin{align}
        \mathcal{C} = \{\mathbf{g} \in B : H\mathbf{g}^\transp = \mathbf{0}\}
    \end{align}
    which equivalently means $\mathcal{C}=\ker H$. This represents the fact that the a codeword $\mathbf{g}$ satisfies all parity constraints defined by the rows of $H$.
    
    For an error $\mathbf{e}\in\F^n_2$, its error syndrome is $\mathbf{s}(\mathbf{e})=H\mathbf{e} \in \F^m_2 =: S$, where $S$ denotes the space of syndromes.
    The code $\mathcal{C}$ can then be represented by a 2-term chain complex as follows:
    \begin{align}
        \mathcal{C}: \quad B \xrightarrow{H} S\, .\label{eq: classical code 2-complex}
    \end{align}
    \begin{figure}
        \centering
        \includegraphics[width=1\linewidth]{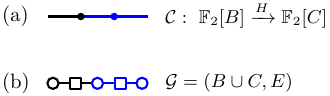}
        \caption{\textbf{(a)} Geometric picture of a 2-complex mapping between 1-dimensional edges and 0-dimensional vertices being analogous to \textbf{(b)} the bipartite Tanner graph of a classical code. The generating set of edges $B$ in (a) is analogous to the bit set $B$ in (b), and the generating set of vertices $C$ in (a) is analogous to the check set $C$ in (b). The boundary map of the complex in (a), $H : \F_2^{|B|} \rightarrow \F_2^{|C|}$, is embedded in the edges $E$ of the Tanner graph, since the entries of one row of $H$ equates to the edges coming out of one check. In both, the highlighted blue check checks the highlighted blue qubits.}
        \label{app-2complexclassicalcode}
    \end{figure}
    Note that the if you append the 0 space to both sides of \eqref{eq: classical code 2-complex} as follows:
    \begin{align}
        \mathcal{C}: \quad 0\rightarrow B \xrightarrow{H} C \rightarrow 0\, ,\label{classical 2-complex w zeros}
    \end{align} 
    where the leftmost map is the trivial map and the rightmost map is the zero map whose kernel is the entirety of $C$, the codespace $\rs G$ is defined as the first homology group of \eqref{classical 2-complex w zeros}:
    \begin{align}
       \rs G = \mathcal{H}_1(\mathcal{C}) = \ker H \, . 
    \end{align}
    Also note that we can draw the connection between the code $\mathcal{C}$ and the Tanner graph $\mathcal{G}$ by using the same edge-vertex boundary map described in Definition~\ref{2-term complex}. Such a connection is shown in Figure~\ref{app-2complexclassicalcode}.
\end{defn}

\begin{defn}[Quantum CSS code 3-term chain complex]\label{def:CSS chain complex}
    An $\llbracket n, k, d \rrbracket$ quantum CSS code $\mathcal{Q}$ is a $2^k$-dimensional subspace of the $2^n$-dimensional Hilbert space $\mathcal{H}^n$. Let $H_X \in \F^{m_X\times n}_2$ and $H_Z\in \F^{m_Z\times n}_2$ be $X$ and $Z$ parity check matrices for the code. The $2^k$-dimensional codespace is the joint $+1$ eigenspace of all $X$ and $Z$ stabilizers defined as
    \begin{subequations}
        \begin{align}
            \{X_i\}_{i=1}^{i=m_X} &= \left\{\bigotimes_{j=1}^n X^{h^{(X)}_{ij}}\right\}_{i=1}^{m_X} \, ,\\
            \{Z_i\}_{i=1}^{i=m_Z} &= \left\{\bigotimes_{j=1}^n Z^{h^{(Z)}_{ij}}\right\}_{i=1}^{m_Z}
        \end{align}
        where $h^{(X)}_{ij}$ ($h^{(Z)}_{ij}$) is the $(i,j)$-th element of $H_X$ ($H_Z$). 
    \end{subequations} 
    
    The $\bar{X}$ ($\bar{Z}$) logical groups are defined to be in the null space of $H_Z$ ($H_X$) but not in the image of $H_X^\transp$ ($H_Z^\transp$):
    \begin{subequations}\label{XZlogicals}
        \begin{align}
            \bar{\mathcal{X}}(\mathcal{Q}) = \ker{H_Z} / \rs{H_X}\label{Xlogicals} \, , \\
            \bar{\mathcal{Z}}(\mathcal{Q}) = \ker{H_X} / \rs{H_Z}\label{Zlogicals} \, .
        \end{align}
    \end{subequations}
    Equivalently, this means that the $\bar{X}$ ($\bar{Z}$) logical operators satisfy all the $Z$ ($X$) checks and cannot be generated by $X$ ($Z$) stabilizers.

    This gives us that $k = \dim \bar{\mathcal{X}}(\mathcal{Q}) = \dim \bar{\mathcal{Z}}(\mathcal{Q})$.
    
    
    Let $Q=\F_2^n$ be the space of qubits and $S_X = \F^{m_X}_2$ ($S_Z = \F^{m_Z}_2$) be the space of $X$- ($Z$-) syndromes obtained by multiplying an $Z$ ($X$) binary error vector $\mathbf{e}_Z \in Q$ ($\mathbf{e}_X \in Q$) by $H_X$ ($H_Z$) to obtain an $X$ ($Z$) syndrome $\mathbf{s}_X \in S_X$ ($\mathbf{s}_Z \in S_Z$). Equivalently, this means that $H_X\mathbf{e}_Z = \mathbf{s}_X$ ($H_Z\mathbf{e}_X = \mathbf{s}_Z$).
    
    Following Definition~\ref{2-term complex}, the code $\mathcal{Q}$ can then be represented by two 2-complexes as follows:
    \begin{subequations}
        \begin{align}
            \F_2[Q] &\xrightarrow{H_X} \F_2[S_X]\, , \\ 
            \F_2[Q] &\xrightarrow{H_Z} \F_2[S_Z] \, .
        \end{align}
    \end{subequations}
    For simplicity, we drop the $\F_2[\cdot]$ labeling, getting us
    \begin{subequations}
        \begin{align}
            Q &\xrightarrow{H_X} S_X\, , \label{sx 2-complex} \\ 
            Q &\xrightarrow{H_Z} S_Z \, . \label{sz 2-complex}
        \end{align}
    \end{subequations}
    Using the cochain complex described in \eqref{eq: 2-complex transpose} on \eqref{sx 2-complex} and the isomorphism of the dual space of $\F_2$, one gets the following 2-complexes:
    \begin{subequations}\label{dual spaces}
        \begin{align}
            Q &\xleftarrow{H_X^\transp} S_X \, \label{sx 2-complexT}\, , \\ 
            Q &\xrightarrow{H_Z} S_Z \, . \label{sz 2-complex2}
        \end{align}
    \end{subequations}
    (Recall that the $S_X$, $Q$, $S_Z$ are still a generating set for $X$-stabilizers, qubits, $Z$-stabilizers respectively and that addition is done over $\F_2$.) Combining these two, one gets the general expression of a CSS code defined as a 3-term chain complex as described in Definition~\ref{3-term complex}:
    \begin{align}
        \mathcal{Q}: \quad S_X \xrightarrow{H^\transp_X} Q \xrightarrow{H^{}_Z} S_Z \, .\label{eq: CSS complex}
    \end{align}
    The condition in \eqref{boundarycondition} is precisely the CSS commutation condition $H_ZH_X^\transp = \mathbf{0}$.

    \begin{figure}
        \centering
        \includegraphics[width=1\linewidth]{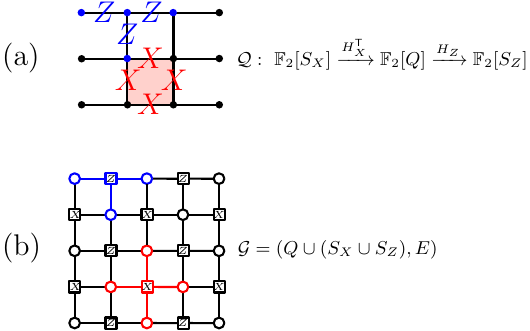}
        \caption{\textbf{(a)} Geometric picture of a 3-complex mapping between 2-dimensional faces, 1-dimensional edges, and 0-dimensional vertices being analogous to \textbf{(b)} the tripartite Tanner graph of a quantum CSS code. The generating set of faces $S_X$ in (a) is analogous to the $X$-stabilizer set $S_X$ in (b), the generating set of qubits $Q$ in (a) is analogous to the qubits $Q$ in (b), and the generating set of vertices $S_Z$ in (a) is analogous to the $Z$-stabilizer set $S_Z$ in (b). The boundary map of the complex in (a) between $S_X$ and $Q$, $H_X^\transpose : \F_2^{|S_X|} \rightarrow \F_2^{|Q|}$, implies that the boundary of faces form $X$ stabilizers. Similarly, the boundary map between $Q$ and $S_Z$, $H_Z : \F_2^{|Q|} \rightarrow \F_2^{|S_Z|}$, implies that the qubits touching vertices (that's aren't on a rough edge) forms $Z$ stabilizers. In both, an $X$ stabilizer and $Z$ stabilizer is highlighted. The typical plaquette and vertex stabilizer geometry of the surface code is apparent in (a).}
        \label{app-3complexquantumcode}
    \end{figure}
    
    One can also see that the group of $\bar{X}$ logicals as defined in \eqref{Xlogicals} is precisely the first homology group of \eqref{eq: CSS complex}
    \begin{align}\label{eq:Xlogicals hom}
        \mathcal{H}_1(\mathcal{Q}) = \ker{H_Z} / \operatorname{im}H_X^\transp = \bar{\mathcal{X}}(\mathcal{Q}) \, ,
    \end{align}
    and that the group of $\bar{Z}$ logicals as defined in \eqref{Zlogicals} is the first cohomology group of \eqref{eq: CSS complex}:
    \begin{align}\label{eq:Zlogicals cohom}
        \mathcal{H}^1(\mathcal{Q}) = \ker{H_X} / \operatorname{im}H_Z^\transp = \bar{\mathcal{Z}}(\mathcal{Q}) \, .
    \end{align}
\end{defn}

As an example, a square tessellation of the 2-torus with faces $(F)$, edges $(E)$ and vertices $(V)$ induces the 3-term chain complex
\begin{align}
    V \xleftarrow{\partial_e} E \xleftarrow{\partial_f} F
\end{align}
and cochain complex
\begin{align}
    V \xrightarrow{\delta^v} E \xrightarrow{\delta^e} F \, .
\end{align}
The 2D toric is obtained by assigning $X$-syndromes to $V$, qubits to $E$, and $Z$-syndromes to $F$.


\section{Hypergraph product codes}\label{app:HGP}

Here we re-review the basic definitions of HGP codes and derive their main properties using the algebraic formalism of Appendix~\ref{app:codes algebraic}. 

Consider two classical linear codes $\mathcal{C}_1$, $\mathcal{C}_2$ with parity-check matrices $H_1 \in \F_2^{m_1 \times n_1}$, $H_2 \in \F_2^{m_2 \times n_2}$ and parameters $[n_1, k_1, d_1]$, $[n_2, k_2, d_2]$ respectively. We define the bipartite Tanner graphs of the codes as 
\begin{subequations}\label{tan}
\begin{align}
    \mathcal{G}_1 &= \mathpzc{T}(B_1, C_1, E_1) \, , \label{tan1} \\
    \mathcal{G}_2 &= \mathpzc{T}(B_2, C_2, E_2) \label{tan2}
\end{align}
\end{subequations}
respectively, where elements of $B_i$ are bit-type vertices, $C_i$ are check-type vertices, and $E_i$ are undirected edges of the form $(b\in B_i, c\in C_i)$ ($i=1,2$). The edges indicate that a certain parity check $c\in C_i$ includes the bit $b\in B_i$.

Define the transpose codes $\mathcal{C}_1^\transpose$, $\mathcal{C}_2^\transpose$ to have parity-check matrices $H_1^\transpose \in \F_2^{n_1 \times m_1}$, $H_2^\transpose \in \F_2^{n_2 \times m_2}$, parameters $[m_1, k_1^\transpose, d_1^\transpose]$, $[m_2, k_2^\transpose, d_2^\transpose]$, and Tanner graphs 
\begin{subequations}\label{tanT}
\begin{align}
    \mathcal{G}_1^\transpose &= \mathpzc{T}(C_1, B_1, E_1) \, , \label{tan1T} \\
    \mathcal{G}_2^\transpose &= \mathpzc{T}(C_2, B_2, E_2) \label{tan2T}
\end{align}
\end{subequations}
respectively. Here, the bit-type and check-type vertices are simply switched from the Tanner graphs of $\mathcal{C}_1$, $\mathcal{C}_2$. 

Geometrically, the HGP code $\mathpzc{H}(\mathcal{C}_1, \mathcal{C}_2)$ of the two classical codes is formed by the Cartesian graph product of $\mathcal{G}_1$ and $\mathcal{G}_2$. The $n_1n_2$ elements in $B_1 \times B_2$ form \Bsectorname{} qubits, $m_1 n_2$ elements in $C_1 \times B_2$ form $X$ stabilizers, $n_1m_2$ elements in $B_1 \times C_2$ form $Z$ stabilizers, and $m_1m_2$ elements in $C_1 \times C_2$ form \Csectorname{} qubits. This makes the entire set of data qubits formed by $(B_1\times B_2) \sqcup (C_1\times C_2)$. Figure~\ref{hgp} shows an example of this construction. Figure~\ref{fig:HGP} shows the general form of this construction, where the $n_1n_2$ qubits in the red box are \Bsectorname{} qubits and $m_1m_2$ qubits in the blue box are \Csectorname{} qubits. 

\begin{figure*}[htbp]
  \centerline{\includegraphics[width=0.8\textwidth]{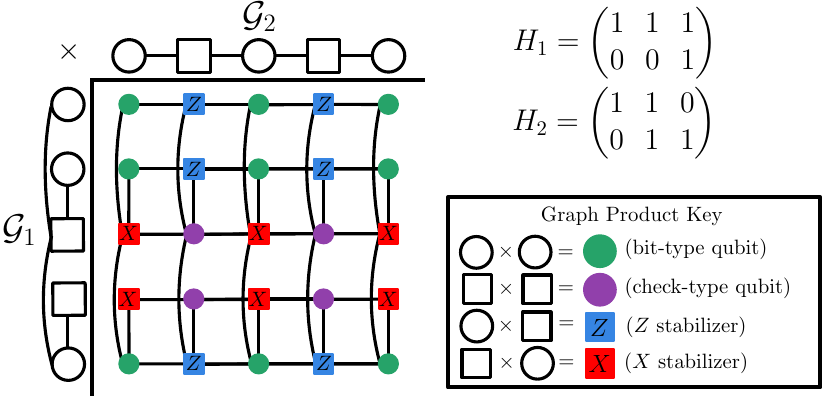}}
  \caption{Hypergraph product $\mathpzc{H}(\mathcal{C}_1, \mathcal{C}_2)$ example construction. In the bipartite classical Tanner graphs $\mathcal{G}_1$, $\mathcal{G}_2$, circles are bits and squares are classical checks. The graph product result is a tripartite quantum Tanner graph with $X$-check, $Z$-check, and qubit nodes. Here, the Tanner graphs correspond to a $[n_1, k_1, d_1] = [3, 1, 2]$, $[n_2, k_2, d_2] = [3, 1, 3]$ classical input codes which yields a $\llbracket n, k, (d_X, d_Z) \rrbracket = \llbracket 13, 1, (3, 2)\rrbracket$ HGP code.}
  \label{hgp}
\end{figure*}

Algebraically, we can show the HGP structure by using the chain complex formalism. Let the input codes $\mathcal{C}_1$, $\mathcal{C}_2$ take the 2-complex structure of \eqref{eq: classical code 2-complex} using the bit sets $B_1$, $B_2$ and check sets $C_1$, $C_2$ (note that these are the same sets from the Tanner graphs \eqref{tan}):
\begin{subequations}\label{eq:C1,C2 chains}
    \begin{align}
    \mathcal{C}_1: \quad B_1 &\xrightarrow{H_1} C_1\, , \\
    \mathcal{C}_2: \quad B_2 &\xrightarrow{H_2} C_2\, .
    \end{align}
\end{subequations}

The HGP code $\mathcal{Q}_\mathpzc{H}$ is the tensor product of the cochain (defined in \eqref{eq: 2-complex transpose}) of the first code, which we denote as $\mathcal{C}_1^\transp$, and the chain of the second, $\mathcal{C}_2$:
\begin{equation}\label{eq:Q=C1xC2}
    \mathcal{Q}_\mathpzc{H} = \mathcal{C}_1^\transp \otimes \mathcal{C}_2\, .
\end{equation}

The resulting chain complex diagram looks as follows, noting that $|B_i| = n_i$ and $|C_i| = m_i$ ($i=1, 2$) and using dual space isomorphisms of $\F_2$:

\begin{equation}
  \begin{tikzcd}[
      row sep=6em,
      column sep=6em,
      /tikz/execute at end picture={
          \path
          coordinate (Y) at ($(current bounding box.north west)!0.15!(current bounding box.south west)$)
          coordinate (X) at ($(current bounding box.north west)!0.15!(current bounding box.north east)$);
        \draw[line width=1pt]
            (Y -| X)                 
            -- (Y -| current bounding box.north east);
        \draw[line width=1pt]
            (Y -| X)                 
            -- (current bounding box.south west -| X);
      },
    ]
    &[-5em] B_2 & C_2 \\[-5em]
    B_1 &[-5em] B_1 \times B_2 & B_1 \times C_2 \\
    C_1 &[-5em] C_1 \times B_2 & C_1 \times C_2
    \arrow["H_2", from=1-2, to=1-3]
    \arrow["H_1^\transp"', from=3-1, to=2-1]
    \arrow["\ident_{n_1} \otimes H_2"', from=2-2, to=2-3]
    \arrow["H_1^\transp \otimes \ident_{n_2}"', from=3-2, to=2-2]
    \arrow["\ident_{m_1} \otimes H_2", from=3-2, to=3-3]
    \arrow["H_1^\transp \otimes \ident_{m_2}", from=3-3, to=2-3]
  \end{tikzcd}
\end{equation}

Reorienting, we can precisely match the 3-chain complex form of CSS codes of \eqref{eq: CSS complex}, which yields our expression of $\mathcal{Q}_\mathpzc{H}$ as a 3-chain complex:

\begin{equation}
  \begin{tikzcd}[
      row sep=3em,
      column sep=1em,
    ]
    S_X &&& C_1 \times B_2 & \\
    Q && B_1 \times B_2 & & C_1 \times C_2 \\
    S_Z &&& B_1 \times C_2 &
    \arrow["H_X^\transpose", from=1-1, to=2-1]
    \arrow["H_Z", from=2-1, to=3-1]
    \arrow["\ident_{m_1} \otimes H_2", from=1-4, to=2-5]{near start}
    \arrow["H_1^\transp \otimes \ident_{n_2}"', from=1-4, to=2-3]
    \arrow["\ident_{n_1} \otimes H_2"', from=2-3, to=3-4]
    \arrow["H_1^\transp \otimes \ident_{m_2}", from=2-5, to=3-4]{near start}
  \end{tikzcd}
\end{equation}

This allows us to easily show the same results as the geometric interpretation and derive the forms of the parity check matrices. Specifically, We have that the basis of $X$ stabilizers $S_X$ are formed by $C_1 \times B_2$, $Z$ stabilizers $S_Z$ by $B_1 \times C_2$, and qubits $Q$ by $(B_1\times B_2) \sqcup (C_1\times C_2)$. Additionally, we have the explicit forms for the parity-check matrices, which match the expressions in \eqref{eq:HGP H_X,H_Z}:
\begin{subequations}\label{eq:HX HZ rederived}
    \begin{align}
        H_X^\transp &= \begin{pmatrix} H_1^\transp \otimes \ident_{n_2} \\ \ident_{m_1} \otimes H_2 \end{pmatrix} \notag \\
        \Longrightarrow H_X &= \begin{pmatrix} H_1 \otimes \ident_{n_2} & \mid & \ident_{m_1} \otimes H_2^\transpose \end{pmatrix} \, , \\
        H_Z &= \begin{pmatrix} \ident_{n_1} \otimes H_2 & \mid & H_1^\transpose \otimes \ident_{m_2} \end{pmatrix} \, .
    \end{align}
\end{subequations}

Using the geometric interpretation of classical codes being 1-dimensional cellulations of edges and vertices and quantum codes being 2-dimensional cellulations of faces, edges, and vertices as done in Figure~\ref{app-2complexclassicalcode} and Figure~\ref{app-3complexquantumcode}, one can describe the HGP code as a product of two 1-dimensional classical cellulations. Such an interpretation is closely analogous to the geometric definition of the HGP code as introduced in Section~\ref{subsec:HGP review}, with Figure~\ref{app-3complexHGPcode} elaborating on that connection.

\begin{figure*}
    \centering
    \includegraphics[width=1\linewidth]{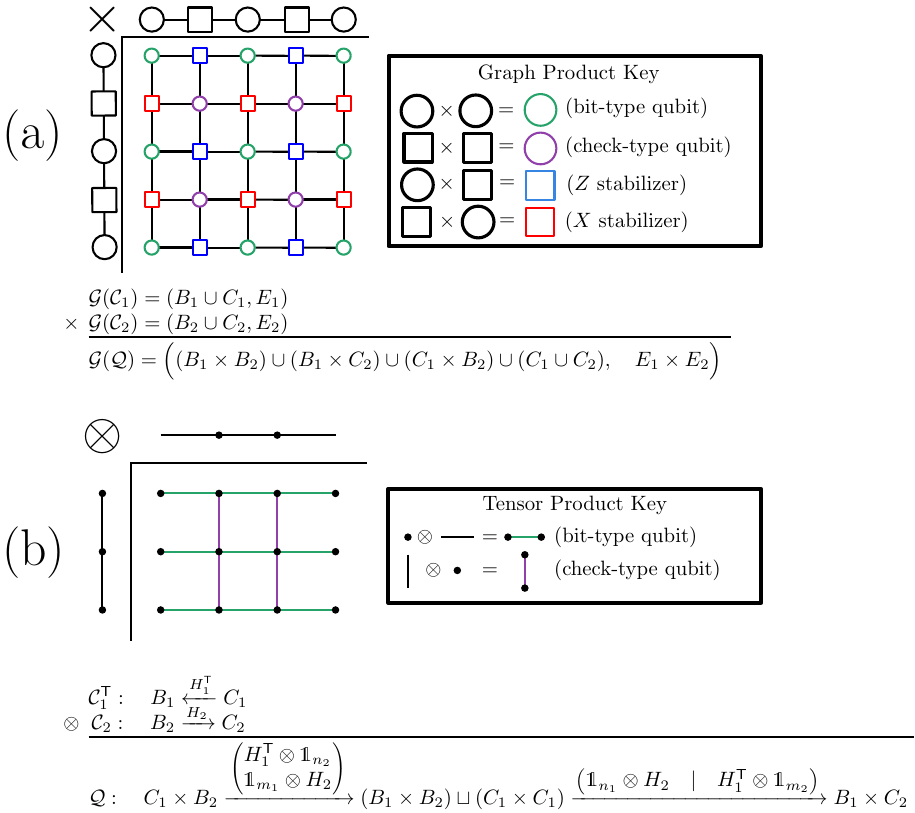}
    \caption{\textbf{(a)} The tripartite Tanner graph product construction introduced in Figure~\ref{hgp} vs. \textbf{(b)} The HGP code of Figure~\ref{app-3complexquantumcode} being represented as the tensor product of two classical codes represented as 1-dimensional cellulations of 2-complexes, as done in Figure~\ref{app-2complexclassicalcode}. As described in Figure~\ref{app-3complexquantumcode}(a), $X$ stabilizers are described by the qubits on the boundary of any (combination of) face(s), and $Z$ stabilizers are described by the qubits that neighbor any (combination of) vertex(es) that are not on a rough edge. (Recall that we are still using the shorthand that the sets in the chain complex are a generating set and algebra is done over $\F_2$.)}
    \label{app-3complexHGPcode}
\end{figure*}

\begin{eg}
    In Figure~\ref{hgp}, we have the parity-check matrices
    \begin{subequations}
        \begin{align}
            H_X &= \begin{pmatrix} \underbrace{\begin{pmatrix} 1 & 1 & 1 \\ 0 & 0 & 1 \end{pmatrix}}_{m_1 \times n_1} \otimes \underbrace{\begin{pmatrix} 1 & 0 & 0 \\ 0 & 1 & 0 \\ 0 & 0 & 1 \end{pmatrix}}_{n_2 \times n_2} & \Bigg| & \underbrace{\begin{pmatrix} 1 & 0 \\ 0 & 1 \end{pmatrix}}_{m_1 \times m_1} \otimes \underbrace{\begin{pmatrix} 1 & 0 \\ 1 & 1 \\ 0 & 1 \end{pmatrix}}_{n_2 \times m_2} \end{pmatrix} \notag \\ &= \underbrace{\left(
\begin{array}{ccccccccccccc}
1 & 0 & 0 & 1 & 0 & 0 & 1 & 0 & 0 & 1 & 0 & 0 & 0 \\
0 & 1 & 0 & 0 & 1 & 0 & 0 & 1 & 0 & 1 & 1 & 0 & 0 \\
0 & 0 & 1 & 0 & 0 & 1 & 0 & 0 & 1 & 0 & 1 & 0 & 0 \\
0 & 0 & 0 & 0 & 0 & 0 & 1 & 0 & 0 & 0 & 0 & 1 & 0 \\
0 & 0 & 0 & 0 & 0 & 0 & 0 & 1 & 0 & 0 & 0 & 1 & 1 \\
0 & 0 & 0 & 0 & 0 & 0 & 0 & 0 & 1 & 0 & 0 & 0 & 1
\end{array}
\right)}_{m_1n_2 \times (n_1n_2 + m_1m_2)} \, ,\label{hxex}
        \end{align} \\
        \begin{align}
            H_Z &= \begin{pmatrix} \underbrace{\begin{pmatrix} 1 & 0 & 0 \\ 0 & 1 & 0 \\ 0 & 0 & 1 \end{pmatrix}}_{n_1 \times n_1} \otimes \underbrace{\begin{pmatrix} 1 & 1 & 0 \\ 0 & 1 & 1 \end{pmatrix}}_{m_2 \times n_2} & \Bigg| & \underbrace{\begin{pmatrix} 1 & 0 \\ 1 & 0 \\ 1 & 1 \end{pmatrix}}_{n_1 \times m_1} \otimes \underbrace{\begin{pmatrix} 1 & 0 \\ 0 & 1 \end{pmatrix}}_{m_2 \times m_2} \end{pmatrix} \notag \\ &= \underbrace{\left(
\begin{array}{ccccccccccccc}
1 & 1 & 0 & 0 & 0 & 0 & 0 & 0 & 0 & 1 & 0 & 0 & 0 \\
0 & 1 & 1 & 0 & 0 & 0 & 0 & 0 & 0 & 0 & 1 & 0 & 0 \\
0 & 0 & 0 & 1 & 1 & 0 & 0 & 0 & 0 & 1 & 0 & 0 & 0 \\
0 & 0 & 0 & 0 & 1 & 1 & 0 & 0 & 0 & 0 & 1 & 0 & 0 \\
0 & 0 & 0 & 0 & 0 & 0 & 1 & 1 & 0 & 1 & 0 & 1 & 0 \\
0 & 0 & 0 & 0 & 0 & 0 & 0 & 1 & 1 & 0 & 1 & 0 & 1
\end{array}
\right)}_{n_1m_2 \times (n_1n_2 + m_1m_2)} \, .\label{hzex}
        \end{align}
    \end{subequations}
\end{eg}

\begin{fact}[Valid parity-check matrices]\label{comm}
    The $X$ and $Z$ parity checks as defined in \eqref{eq:HGP H_X,H_Z} commute.
    \begin{proof} We can show
    \begin{align}
        H_XH_Z^\transpose &= \begin{pmatrix} H_1 \otimes \ident_{n_2} & \mid & \ident_{m_1} \otimes H_2^\transpose \end{pmatrix} \begin{pmatrix} \ident_{n_1} \otimes H_2^\transpose \\ H_1 \otimes \ident_{m_2} \end{pmatrix} \notag \\
        &= (H_1 \otimes \ident_{n_2})(\ident_{n_1} \otimes H_2^\transpose) + (\ident_{m_1} \otimes H_2^\transpose)(H_1 \otimes \ident_{m_2}) \notag \\
        &= 2(H_1 \otimes H_2^\transpose) \equiv \mathbf{0}\text{ (mod 2)} \, . \label{hxhzt}
    \end{align}
    This is consistent with the composition property of the 3-chain describing the code.
    \end{proof}
\end{fact}

We can show the qLDPC parameters of the HGP code as follows. Let the input codes $H_1$, $H_2$ have maximum row weights $w_1^{(r)}$, $w_2^{(r)}$ and column weights $w_1^{(c)}$, $w_2^{(c)}$. Then, by simple inspection of \eqref{eq:HX HZ rederived}, the maximum row weight and column weight of $H_X$ is
\begin{subequations}
\begin{align}
    w_1^{(r)} + w_2^{(c)}& \, , \label{rw H_X} \\
    \max(w_1^{(c)}, w_2^{(r)})& \label{cw H_X}
\end{align}
\end{subequations}
respectively. For $H_Z$, we get maximum row and column weights
\begin{subequations}
\begin{align}
    w_2^{(r)} + w_1^{(c)}& \, , \label{rw H_Z} \\
    \max(w_2^{(c)}, w_1^{(r)})& \, . \label{cw H_Z}
\end{align}
\end{subequations}
respectively.

\subsection{Code dimension}

The number of data qubits in the code is given by
\begin{align} 
n &= |B_1 \times B_2| + |C_1 \times C_2| \notag \\
&= n_1n_2 + m_1m_2 \, . \label{eq:HGP n}
\end{align}

To analyze the number of logical qubits $k$, we first recite a well-known classical construction called a tensor-product code \cite{Elias_1954}.

\begin{defn}[Tensor-product code]
    Given two classical linear codes $\mathcal{C}_1$ and $\mathcal{C}_2$ with parameters $[n_1,k_1,d_1]$ and $[n_2,k_2,d_2]$ respectively, their tensor-product code $\mathcal{C}_1 \otimes \mathcal{C}_2$ has parameters $[n_1n_2,k_1k_2,d_1d_2]$. Codewords of $\mathcal{C}_1 \otimes \mathcal{C}_2$ can be arranged as a $n_1 \times n_2$ matrix whose column and row vectors are in $\mathcal{C}_1$ and $\mathcal{C}_2$ respectively.
\end{defn}

Parity checks for the tensor-product code are given by $H_1 \otimes \ident_{n_2}$ and $\ident_{n_1} \otimes H_2$. Since the code dimension of a tensor-product code is $k_1k_2$, we have the following fact.

\begin{fact}\label{fact:tensor code rank} 
    For the input parity-check matrices $H_1$ and $H_2$ in the tensor-product code, we have
    \begin{align}
        \rank H_{\rm TP} = \rank \begin{pmatrix} H_1 \otimes \ident_{n_2} \\ \ident_{n_1} \otimes H_2 \end{pmatrix} = n_1n_2 - k_1k_2 \, .
    \end{align}
\end{fact}

The usual formula for the code dimension of a CSS code is
\begin{align}\label{eq:CSS k app}
    k = n - \rank H_X - \rank H_Z \, .
\end{align}
Combining \eqref{eq:HGP H_X,H_Z} with Fact \ref{fact:tensor code rank} and the fact that $\rank H = \rank H^\transpose$, we obtain
\begin{subequations}\label{eq:HGP rank H_X,H_Z}
\begin{align}
    \rank H_X &= m_1n_2 - k^\transpose_1k^{}_2  \\
    \rank H_Z &= n_1m_2 - k^{}_1k^\transpose_2 \, .
\end{align}
\end{subequations}
Plugging in \eqref{eq:HGP n} and \eqref{eq:HGP rank H_X,H_Z} into \eqref{eq:CSS k app}, we arrive at the formula for the code dimension of an HGP code in terms of those of its input codes:
\begin{align}\label{eq:HGP k app}
    k_{\mathpzc{H}} &= n_1n_2 + m_1m_2 - m_1n_2 - n_1m_2 + k^\transpose_1k^{}_2 + k^{}_1k^\transpose_2  \notag \\
    &= (n_1-m_1)(n_2-m_2) + k^\transpose_1k^{}_2 + k^{}_1k^\transpose_2  \notag \\
    &= \left( k^{}_1 - k^\transpose_1 \right) \left( k^{}_2 - k^\transpose_2 \right) + k^\transpose_1k^{}_2 + k^{}_1k^\transpose_2  \notag \\
    &= k_1k_2 + k^\transpose_1k^\transpose_2 \, ,
\end{align}
where in the penultimate line we used the rank-nullity relation $n-m=k-k^\transpose$.

Alternatively, we could use the HGP algebraic formalism to show the same thing. We use the K\"unneth formula, which states the following:
\begin{subequations}
    \begin{align}
        \mathcal{H}_\ell(X \otimes Y) &\simeq \bigoplus_{\{i, j\ \mid\ i+j=\ell\}} \mathcal{H}_i(X) \otimes \mathcal{H}_j(Y)\, ,\label{eq:kunnethX} \\
        \mathcal{H}^\ell(X\otimes Y) &\simeq \bigoplus_{\{i, j\ \mid\ i+j=\ell\}} \mathcal{H}^i(X) \otimes \mathcal{H}^j(Y)\label{eq:kunnethZ}
    \end{align}
\end{subequations}
where $X$ and $Y$ are chain, cochain complexes in \eqref{eq:kunnethX}, \eqref{eq:kunnethZ} respectively.

Then, the same result as \eqref{eq:HGP k app} can be shown by applying the K\"unneth formula \eqref{eq:kunnethX} to \eqref{eq:Xlogicals hom}, using \eqref{eq:C1,C2 chains} and \eqref{eq:Q=C1xC2} (note that the same result is obtained from applying \eqref{eq:kunnethZ})
\begin{align}
    k &= \dim \bar{\mathcal{X}}(\mathcal{Q}_\mathpzc{H})\notag \\
    &= \dim \mathcal{H}_1(\mathcal{C}^\transp_1\otimes\mathcal{C}_2) \notag \\
    &= \dim\left(\mathcal{H}_0(\mathcal{C}^\transp_1) \otimes \mathcal{H}_1(\mathcal{C}_2) \right) + \dim\left( \mathcal{H}_1(\mathcal{C}^\transp_1) \otimes \mathcal{H}_0(\mathcal{C}_2) \right)  \notag \\
    &= \dim\left(\dfrac{\F^{n_1}_2}{\im H^\transp_1}\right) \cdot \dim\ker H_2 \notag \\ &\qquad\qquad + \dim\ker H^\transp_1 \cdot \dim\left(\dfrac{\F^{m_2}_2}{\im H_2}\right) \notag \\
    &= \dim\ker H_1 \cdot \dim\ker H_2 + \dim\ker H^\transp_1 \cdot \dim\ker H^\transp_2  \notag \\
    &= k^{}_1k^{}_2 + k^\transp_1 k^\transp_2 \, .
\end{align}

\subsection{Logical operators}

    We next construct the logical $\bar{X}$ and $\bar{Z}$ Pauli operator structure for the $k$ logical qubits. We do this straight from the algebraic CSS-homology correspondence of Definition \ref{def:CSS chain complex}.
\begin{subequations}
    \begin{align}
        \bar{\mathcal{X}}(\mathcal{Q}_\mathpzc{H})
        &= \mathcal{H}_1(\mathcal{C}_1^\transpose \otimes \mathcal{C}_2) \notag \\
        &\simeq \big(\mathcal{H}_0(\mathcal{C}_1^\transpose) \otimes \mathcal{H}_1(\mathcal{C}_2)\big) \oplus \big(\mathcal{H}_1(\mathcal{C}_1^\transpose) \otimes \mathcal{H}_0(\mathcal{C}_2) \big) \notag \\
        &= \left( \dfrac{\F_2^{n_1}}{\im H_1^\transpose} \otimes \ker H_2 \right) \oplus \left( \ker H_1^\transpose \otimes \dfrac{\F_2^{m_2}}{\im H_2} \right) \notag \, , \label{Xkunneth} \\
    \end{align}
    \begin{align}
        \bar{\mathcal{Z}}(\mathcal{Q}_\mathpzc{H})
        &= \mathcal{H}^1(\mathcal{C}_1^\transpose \otimes \mathcal{C}_2) \notag \\
        &\simeq \big(\mathcal{H}^0(\mathcal{C}_1^\transpose) \otimes \mathcal{H}^1(\mathcal{C}_2)\big) \oplus \big(\mathcal{H}^1(\mathcal{C}_1^\transpose) \otimes \mathcal{H}^0(\mathcal{C}_2) \big) \notag \\
        &= \left( \ker H_1 \otimes \dfrac{\F_2^{n_2}}{\im H_2^\transpose} \right) \oplus \left( \dfrac{\F_2^{m_1}}{\im H_1} \otimes \ker H_2^\transpose \right) \notag \, . \label{Zkunneth} \\
    \end{align}
\end{subequations}
    The first (second) term after the $\oplus$ defines the logical operator structure on the bit- (check-) type qubits. Using the notation that $\cdot|_\mathscr{B}$ ($\cdot|_\mathscr{C}$) is the restriction to the first $n_1n_2$ ($m_1m_2$) or bit- (check-) type qubits, we have
\begin{subequations}\label{eq:kunneth logical structure}
    \begin{align}
        \rs G^\mathscr{B}_X|_\mathscr{B} &\simeq \dfrac{\F_2^{n_1}}{\im H_1^\transpose} \otimes \ker H_2 \, ,\\
        \rs G^\mathscr{B}_Z|_\mathscr{B} &\simeq \ker H_1 \otimes \dfrac{\F_2^{n_2}}{\im H_2^\transpose} \, , \\
        \rs G^\mathscr{C}_X|_\mathscr{C} &\simeq \ker H_1^\transpose \otimes \dfrac{\F_2^{m_2}}{\im H_2} \, ,\\
        \rs G^\mathscr{C}_Z|_\mathscr{C} &\simeq \dfrac{\F_2^{m_1}}{\im H_1} \otimes \ker H_2^\transpose
    \end{align} 
\end{subequations}

    Now let us construct an explicit form for these operators. Define the $n_1$-dimensional unit row vector $\mathbf{e}_j$ to be zero in every column except column $j$:
    \begin{align}
      \includegraphics[scale=1]{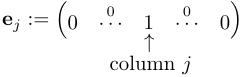} \label{ej}
    \end{align}

    \begin{lem}\label{ejlem}
        For $j \in \{1, ..., k_1\}$, $\mathbf{e}_j \notin \rs H_1$.
        \begin{proof}
            By column permuting and row reducing $H_1$ into its canonical form shown below (for some matrix $A_1$), one can see that it is not possible to take linear combinations of the rows in $H_1$ and get all 0 in the last $\operatorname{rank}H_1$ columns:
            \begin{align}
              \includegraphics[width=0.4\textwidth]{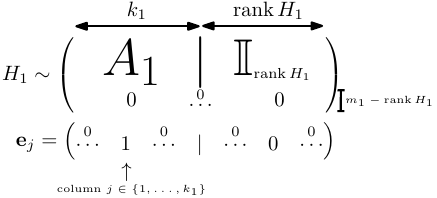} \label{h1}
            \end{align}
        \end{proof}
    \end{lem}

   Define the $k_1 \times n_1$ ($k_2 \times n_2$) matrices $E_1^\mathscr{B}$ ($E_2^\mathscr{B}$) whose rows are $\mathbf{e}_j$ for $j = 1, \dots, k_1$ ($j = 1, \dots, k_2$) so that the all rows of $E_1^\mathscr{B}$, ($E_2^\mathscr{B}$) are not contained in $\rs H_1$, ($\rs H_2$). Let $G_1 \in \F_2^{k_1 \times n_1}$ ($G_2 \in \F_2^{k_2 \times n_2}$) be the generator matrices of $\mathcal{C}_1$ ($\mathcal{C}_2$); i.e. $H_1G_1^\transpose = \mathbf{0}$ ($H_2G_2^\transpose = \mathbf{0}$). Let the canonical basis of the $\bar{X}^\mathscr{B}$ and $\bar{Z}^\mathscr{B}$ logical operators on the $k_1k_2$ \Bsectorname{} logical qubits be notated by $G_X^\mathscr{B}\in \F_2^{k_1k_2 \times n}$ and $G_Z^\mathscr{B}\in \F_2^{k_1k_2 \times n}$. 

   Now, given that $\rs G_i = \ker H_i$ ($\rs G_{i, T} = \ker H_i^\transpose$) and choosing $E_i^\mathscr{B} \in \F_2^{k_i\times n_i}$ such that $\mathbb{F}_2^{n_i} = \rs H_i \oplus \rs E_i^\mathscr{B}$ and $E_i^\mathscr{B}G_i^\transp = \ident_{k_i}$ ($i=1,2$), we can define the explicit logical operator forms for the \Bsectorname{} qubits as
    \begin{subequations} \label{eq: bit-type ops}
        \begin{align}
            G_X^\mathscr{B} &= \begin{pmatrix} E_1^\mathscr{B} \otimes G_2 & \mid & \mathbf{0}_{k_1k_2 \times m_1m_2}\end{pmatrix} \, , \label{gxB} \\
            G_Z^\mathscr{B} &= \begin{pmatrix} G_1 \otimes E_2^\mathscr{B} & \mid & \mathbf{0}_{k_1k_2 \times m_1m_2}\end{pmatrix} \, . \label{gzB}
        \end{align}
    \end{subequations}
    The zeros indicate that there is no logical support on the \Csectorname{} qubits. It is evident from the structure that $\bar{X}^\mathscr{B}$ ($\bar{Z}^\mathscr{B}$) logicals look like the codewords of $G_2$ ($G_1$). Similarly, the canonical basis of $\bar{X}^\mathscr{C}$ and $\bar{Z}^\mathscr{C}$ logical operators on the \Csectorname{} logical qubits, $G_X^\mathscr{C}\in \F_2^{k_1^\transpose k_2^\transpose \times n}$ and $G_Z^\mathscr{C}\in \F_2^{k_1^\transpose k_2^\transpose \times n}$, is defined by the rows of
    \begin{subequations} \label{eq: check-type ops}
    \begin{align}
        G_X^\mathscr{C} &= \begin{pmatrix} \mathbf{0}_{k_1^\transpose k_2^\transpose \times n_1n_2} & \mid & G_{1, T} \otimes E_2^\mathscr{C} \end{pmatrix} \label{gxC} \\
        G_Z^\mathscr{C} &= \begin{pmatrix} \mathbf{0}_{k_1^\transpose k_2^\transpose \times n_1n_2} & \mid & E_1^\mathscr{C} \otimes G_{2, T} \end{pmatrix} \label{gzC}
    \end{align}
    \end{subequations}
    for some $G_{i, T} \in \F_2^{k_i^\transpose\times m_i}$ that satisfies $H_i^\transpose G_{i, T}^\transpose = \mathbf{0}$ and $E_i^\mathscr{C} \in \F_2^{k_i^\transpose \times m_i}$ that satisfies $\mathbb{F}_2^{m_i} = \rs H_i^\transp \oplus \rs E_i^\mathscr{C}$ ($i=1, 2$), where $E_i^\mathscr{C}G_{i, T}^\transp = \ident_{k_i^\transp}$. Again, the zeros indicate that there is no logical support on the \Bsectorname{} qubits. These definitions are clearly consistent with the results from \eqref{eq:kunneth logical structure}.

    \begin{eg}
        For Figure~\ref{hgp}, note that $k_i = n_i - \rank H_i = 1$ ($i = 1, 2$). Then, we have
        \begin{subequations}
            \begin{align}
                G_X^\mathscr{B} &= \begin{pmatrix} \underbrace{\begin{pmatrix} 1 & 0 & 0 \end{pmatrix}}_{k_1 \times n_1} \otimes \underbrace{\begin{pmatrix} 1 & 1 & 1 \end{pmatrix}}_{k_2 \times n_2} & \Big| & \underbrace{\begin{pmatrix} 0 & 0 & 0 & 0 \end{pmatrix}}_{k_1k_2 \times m_1m_2} \end{pmatrix} \\
                &= \underbrace{\left( \begin{array}{cccccccccccccc} 1 & 1 & 1 & 0 & 0 & 0 & 0 & 0 & 0 & | & 0 & 0 & 0 & 0\end{array} \right)}_{k_1k_2 \times (n_1n_2 + m_1m_2)}\, ,\label{gxb}
            \end{align} \\
            \begin{align}
                G_Z^\mathscr{B} &= \begin{pmatrix} \underbrace{\begin{pmatrix} 1 & 1 & 0 \end{pmatrix}}_{k_1 \times n_1} \otimes \underbrace{\begin{pmatrix} 1 & 0 & 0 \end{pmatrix}}_{k_2 \times n_2} & \Big| & \underbrace{\begin{pmatrix} 0 & 0 & 0 & 0 \end{pmatrix}}_{k_1k_2 \times m_1m_2} \end{pmatrix} \\
                &= \underbrace{\left( \begin{array}{cccccccccccccc} 1 & 0 & 0 & 1 & 0 & 0 & 0 & 0 & 0 & | & 0 & 0 & 0 & 0\end{array} \right)}_{k_1k_2 \times (n_1n_2 + m_1m_2)} \, .
            \end{align}
        \end{subequations}
        $G_X^\mathscr{C}$ and $G_Z^\mathscr{C}$ are both trivial since $G_{1, T}$ and $G_{2, T}$ are both trivial. Hence, there are no logical operators on the \Csectorname{} qubits since $k_1^\transpose = k_2^\transpose = 0$. These logical operators are shown in Figure~\ref{hgplogis}. Note that the qubit labeling goes left-to-right then top-to-bottom. 
        \begin{figure}[htbp]
          \centerline{\includegraphics[width=0.5\textwidth]{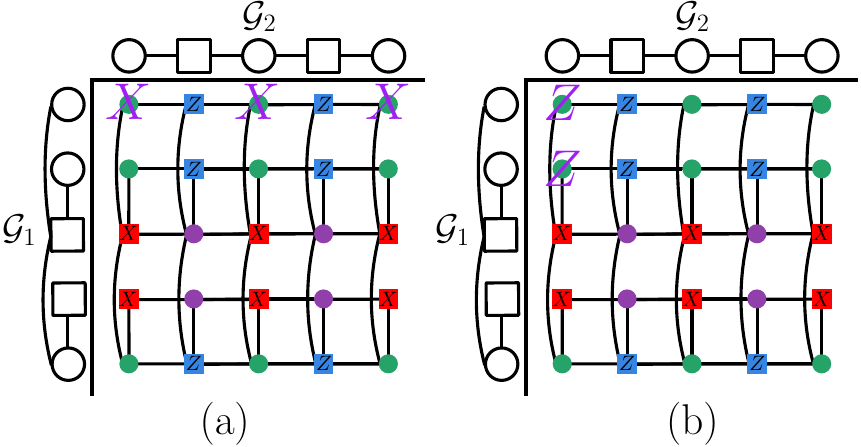}}
          \caption{For the HGP example in Figure~\ref{hgp}, \textbf{(a)} Canonical $\bar{X}^\mathscr{B}$ logical operator on \Bsectorname{} qubits; \textbf{(b)} Canonical $\bar{Z}^\mathscr{B}$ logical operator on \Bsectorname{} qubits.}
          \label{hgplogis}
        \end{figure}
    \end{eg}
    
    \begin{fact}[Anticommutation of logicals]\label{anticomm}
    The logical operators as defined in \eqref{eq: bit-type ops} and \eqref{eq: check-type ops} anticommute when acting on the same qubit and commute when acting on different qubits. Equivalently, for any logical $\bar{X}$ and $\bar{Z}$ operator on any logical qubits $i$, $j$, $\bar{X}_i\bar{Z}_j = (-1)^{\delta_{ij}}\bar{Z}_j\bar{X}_i$ for $i, j \in \{1, ..., k\}$, where $\delta_{ij} = \begin{cases} 1, & \text{if } i = j \\ 0, & \text{if } i \neq j \end{cases}$ is the Kronecker delta.
    \begin{proof} 
        We can show
        \begin{align}
            G_X^\mathscr{B}(G_Z^\mathscr{B})^\transp
            &= \begin{pmatrix} E_1^\mathscr{B} \otimes G_2 & \mid & \mathbf{0}_{k_1k_2 \times m_1 m_2}\end{pmatrix}
            \begin{pmatrix} G_1^\transp \otimes (E_2^\mathscr{B})^\transp \\ \mathbf{0}_{m_1 m_2 \times k_1 k_2} \end{pmatrix} \notag \\
            &= \left(E_1^\mathscr{B} \otimes G_2\right)
            \left(G_1^\transp \otimes (E_2^\mathscr{B})^\transp\right) \notag \\
            &= E_1^\mathscr{B}G_1^\transp \otimes G_2(E_2^\mathscr{B})^\transp \notag \\
            &= \begin{pmatrix} \ident_{k_1} & \mid & \mathbf{0}_{k_1 \times r_1}\end{pmatrix}
            \begin{pmatrix} \ident_{k_1} \\ A_1 \end{pmatrix}
            \otimes
            \begin{pmatrix} \ident_{k_2} & \mid & A_2^\transp \end{pmatrix}
            \begin{pmatrix} \ident_{k_2} \\ \mathbf{0}_{r_2 \times k_2} \end{pmatrix} \notag \\
            &= \ident_{k_1 k_2} \, . \label{gxBgzBT}
        \end{align} 
        
        Similarly,
        \begin{align}
            G_X^\mathscr{C}(G_Z^\mathscr{C})^\transp
            &= \begin{pmatrix}
                   \mathbf{0}_{k_1^\transp k_2^\transp \times n_1 n_2}
                   & \mid &
                   G_{1,T} \otimes E_2^\mathscr{C}
               \end{pmatrix}
               \begin{pmatrix}
                   \mathbf{0}_{n_1 n_2 \times k_1^\transp k_2^\transp} \\
                   (E_1^\mathscr{C})^\transp \otimes G_{2,T}^\transp
               \end{pmatrix} \notag \\
            &= \left(G_{1,T} \otimes E_2^\mathscr{C}\right)
               \left((E_1^\mathscr{C})^\transp \otimes G_{2,T}^\transp\right) \notag \\
            &= \left(G_{1,T}(E_1^\mathscr{C})^\transp\right)
               \otimes
               \left(E_2^\mathscr{C}G_{2,T}^\transp\right) \notag \\
            &= \begin{pmatrix}
                    \ident_{k_1^\transp} & \mid & B_1^\transp
               \end{pmatrix}
               \begin{pmatrix}
                    \ident_{k_1^\transp} \\
                    \mathbf{0}_{r_1 \times k_1^\transp}
               \end{pmatrix}
               \otimes
               \begin{pmatrix}
                    \ident_{k_2^\transp} & \mid & \mathbf{0}_{k_2^\transp \times r_2}
               \end{pmatrix}
               \begin{pmatrix}
                    \ident_{k_2^\transp} \\
                    B_2
               \end{pmatrix} \notag \\
            &= \ident_{k_1^\transp k_2^\transp} \, . \label{gxCgzCT}
        \end{align}
        
        Here, we use the fact that $H_i$ and $H_i^\transp$ can be put into the forms
        \begin{subequations}\label{Hcanonical}
        \begin{align}
            H_i &\sim
            \begin{pmatrix}
                A_i & \mid & \ident_{r_i} \\
                \mathbf{0}_{k_i^\transp \times k_i} & \mid & \mathbf{0}_{k_i^\transp \times r_i}
            \end{pmatrix} \, ,\\
            H_i^\transp &\sim
            \begin{pmatrix}
                B_i & \mid & \ident_{r_i} \\
                \mathbf{0}_{k_i \times k_i^\transp} & \mid & \mathbf{0}_{k_i \times r_i}
            \end{pmatrix} \, ,
        \end{align}
        \end{subequations}
        (as also done in \eqref{h1}) for an $r_i \times k_i$ matrix $A_i$ and an $r_i \times k_i^\transp$ matrix $B_i$, where $r_i = \operatorname{rank}H_i = \operatorname{rank}H_i^\transp$ ($i = 1, 2$). Given this form,
        $G_1$, $G_2$, $G_{1,T}$, and $G_{2,T}$ can be taken to be
        \begin{subequations}\label{Gcanonical}
        \begin{align}
            G_1 &= \begin{pmatrix} \ident_{k_1} & \mid & A_1^\transp \end{pmatrix}\, , \\  
            G_2 &= \begin{pmatrix} \ident_{k_2} & \mid & A_2^\transp \end{pmatrix}\, , \\ 
            G_{1,T} &= \begin{pmatrix} \ident_{k_1^\transp} & \mid & B_1^\transp \end{pmatrix}\, ,\\ 
            G_{2,T} &= \begin{pmatrix} \ident_{k_2^\transp} & \mid & B_2^\transp \end{pmatrix} \, .  
        \end{align}
        \end{subequations}
        These satisfy
        \begin{subequations}
        \begin{align}
            H_1G_1^\transp
            &= \begin{pmatrix} A_1 & \mid & \ident_{r_1} \\ \mathbf{0}_{k_1^\transp \times k_1} & \mid & \mathbf{0}_{k_1^\transp \times r_1} \end{pmatrix}
            \begin{pmatrix} \ident_{k_1} \\ A_1 \end{pmatrix}
            = \begin{pmatrix} 2A_1 \\ \mathbf{0} \end{pmatrix}
            \equiv \mathbf{0}\, , \label{h1g1t} \\
            H_2G_2^\transp
            &= \begin{pmatrix} A_2 & \mid & \ident_{r_2} \\ \mathbf{0}_{k_2^\transp \times k_2} & \mid & \mathbf{0}_{k_2^\transp \times r_2} \end{pmatrix}
            \begin{pmatrix} \ident_{k_2} \\ A_2 \end{pmatrix}
            = \begin{pmatrix} 2A_2 \\ \mathbf{0} \end{pmatrix}
            \equiv \mathbf{0}\, , \label{h2g2t} \\
            H_1^\transp G_{1,T}^\transp
            &= \begin{pmatrix} B_1 & \mid & \ident_{r_1} \\ \mathbf{0}_{k_1 \times k_1^\transp} & \mid & \mathbf{0}_{k_1 \times r_1} \end{pmatrix}
            \begin{pmatrix} \ident_{k_1^\transp} \\ B_1 \end{pmatrix}
            = \begin{pmatrix} 2B_1 \\ \mathbf{0} \end{pmatrix}
            \equiv \mathbf{0}\, , \label{h1Tg1Tt} \\
            H_2^\transp G_{2,T}^\transp
            &= \begin{pmatrix} B_2 & \mid & \ident_{r_2} \\ \mathbf{0}_{k_2 \times k_2^\transp} & \mid & \mathbf{0}_{k_2 \times r_2} \end{pmatrix}
            \begin{pmatrix} \ident_{k_2^\transp} \\ B_2 \end{pmatrix}
            = \begin{pmatrix} 2B_2 \\ \mathbf{0} \end{pmatrix}
            \equiv \mathbf{0}\, . \label{h2Tg2Tt}
        \end{align}
        \end{subequations}
    \end{proof}
\end{fact}
    
    \begin{fact}[Valid basis of logicals]\label{valid logi}
        We show that the $\bar{X}^\mathscr{B}$ logical operators \textbf{(a)} satisfy the $Z$ parity-checks and \textbf{(b)} cannot be generated by $X$ stabilizers, and claim that the same can be shown for $\bar{Z}^\mathscr{B}$, $\bar{X}^\mathscr{C}$, and $\bar{Z}^\mathscr{C}$ logical operators.
        \begin{proof}
        Let $H_X^\mathscr{B}$ ($H_Z^\mathscr{B}$) contain the $X$ ($Z$) checks on \Bsectorname{} qubits only while retaining its dimension. Mathematically,
        \begin{subequations}
            \begin{align}
                H_X^\mathscr{B} &= \begin{pmatrix} H_1 \otimes \ident_{n_2} & \mid &  \mathbf{0}_{m_1n_2\times m_1 m_2} \end{pmatrix} \, , \label{eq:HGP H_XB} \\
                H_Z^\mathscr{B} &= \begin{pmatrix} \ident_{n_1} \otimes H_2 & \mid & \mathbf{0}_{n_1m_2\times m_1 m_2} \end{pmatrix} \, . \label{eq:HGP H_ZB}
            \end{align}
        \end{subequations}
        Define the restriction of $\bar{X}$ and $\bar{Z}$ logicals to \Bsectorname{} qubits as \eqref{eq: bit-type ops}.
        \begin{enumerate}[label=(\alph*)]
            \item We can simply show that
            \begin{align} H_Z {G_X^\mathscr{B}}^\transpose 
                &= \begin{pmatrix} \ident_{n_1} \otimes H_2 & \mid & H_1^\transpose \otimes \ident_{m_2} \end{pmatrix} 
                   \begin{pmatrix} {E_1^\mathscr{B}}^\transpose \otimes G_2^\transpose \\ \mathbf{0}_{m_1 m_2 \times k_1 k_2} \end{pmatrix} \notag \\
                &= {E_1^\mathscr{B}}^\transpose \otimes H_2G_2^\transpose \notag \\
                &= \mathbf{0} \, .  \label{hzgxt} 
            \end{align} 
            \item We have that
                \begin{align} \rs G_X^\mathscr{B} 
                &= \rs  \begin{pmatrix} E_1^\mathscr{B} \otimes G_2 & \mid & \mathbf{0}_{k_1k_2 \times m_1m_2} \end{pmatrix} \notag \\
                &= \bigl\{ \begin{pmatrix} v & \mid & \mathbf{0}_{1\times m_1 m_2} \end{pmatrix} \;:\; v \in \rs E_1^\mathscr{B} \otimes \rs G_2 \bigr\} \label{rse1rsg2} \end{align}
                whose span is not contained in
                \begin{align}\rs H_X^\mathscr{B} 
                &= \rs  \begin{pmatrix} H_1 \otimes \ident_{n_2} & \mid & \mathbf{0}_{m_1n_2 \times m_1 m_2} \end{pmatrix} \notag \\
                &= \bigl\{ \begin{pmatrix} v & \mid & \mathbf{0}_{1\times m_1 m_2} \end{pmatrix} \;:\; v \in \rs H_1 \otimes \F^{n_2} \bigr\} \label{rsgxrshxb} \end{align} 
                by definition, as the rows of $E_1^\mathscr{B}$ are not contained in $\rs  H_1$. 
        \end{enumerate}
        \end{proof}
    \end{fact}

\subsection{Minimum distances}

    Here we provide a proof of Proposition~\ref{prop:HGP logical sector weights} for completeness; from here, the argument for the minimum distances $d_Z = \min\left(d^{}_1,d^\transp_2\right)$ and $d_X = \min\left(d^\transp_1,d^{}_2\right)$ follow from Section~\ref{subsec:HGP review}. 
    
    For a binary vector $\mathbf{x} \in \mathbb{F}^{n}_2$, whose nonzero elements denote the support of a logical operator, let $\mathbf{x}|_\mathscr{B} \in \F^{n_1n_2}_2$ ($\mathbf{x}|_\mathscr{C} \in \F^{m_1m_2}_2$) be its restriction to the \Bsectorname{} (\Csectorname{}) qubits. For example, let \eqref{gxb} be the $\mathbf{x}$ representing the logical operator in Figure~\ref{hgplogis}(a). Then, $\mathbf{x}|_\mathscr{B}$ would be the part on the left of the vertical line, $\begin{pmatrix} 1 & 1 & 1 & 0 & 0 & 0 & 0 & 0 & 0 \end{pmatrix}^\transp$, and $\mathbf{x}|_\mathscr{C}$ would be the part on the right of the vertical line, $\begin{pmatrix} 0 & 0 & 0 & 0 \end{pmatrix}^\transp$.

    Furthermore, let $\abs{\mathbf{x}}_\mathscr{B}^r$ and $\abs{\mathbf{x}}_\mathscr{B}^c$ ($\abs{\mathbf{x}}_\mathscr{C}^r$ and $\abs{\mathbf{x}}_\mathscr{C}^c$) denote the number of rows and columns within the \Bsectorname{} (\Csectorname{}) qubits that $\mathbf{x}|_\mathscr{B}$ ($\mathbf{x}|_\mathscr{C}$) intersects respectively. For the example $\mathbf{x}$ previously defined in \eqref{gxb}, we have $\abs{\mathbf{x}}_\mathscr{B}^r = 1$, $\abs{\mathbf{x}}_\mathscr{B}^c = 3$, and $\abs{\mathbf{x}}_\mathscr{C}^r = \abs{\mathbf{x}}_\mathscr{C}^c = 0$ (see Figure~\ref{hgplogis}).
    
    Note that $\abs{\mathbf{x}} \geq \abs{\mathbf{x}|_\mathscr{B}} \geq \abs{\mathbf{x}}_\mathscr{B}^r$ and $\abs{\mathbf{x}} \geq \abs{\mathbf{x}|_\mathscr{B}} \geq \abs{\mathbf{x}}_\mathscr{B}^c$. Likewise, $\abs{\mathbf{x}} \geq \abs{\mathbf{x}|_\mathscr{C}} \geq \abs{\mathbf{x}}_\mathscr{C}^r$ and $\abs{\mathbf{x}} \geq \abs{\mathbf{x}|_\mathscr{C}} \geq \abs{\mathbf{x}}_\mathscr{C}^c$. The following result, first shown by Quintavalle and Campbell (\cite{quintavalle2022reshape}, Prop. 2), concerns the structure of logical operators within the \Bsectorname{} and \Csectorname{} qubits.
    
    \begin{prop*}[Proposition~\ref{prop:HGP logical sector weights}]
    \label{prop:HGP logical sector weights app}
        For any nontrivial logical $\bar{X}$ with $X$-support $\mathbf{x} \in \ker{H_Z}\setminus \rs{H_X}$, we have
        \begin{align}\label{eq:B sector weights app}
            \abs{\mathbf{x}}_\mathscr{B}^c \geq d_2 \quad\text{or}\quad \abs{\mathbf{x}}_\mathscr{C}^r \geq d^\transp_1 \, .
        \end{align}
        Likewise, for any nontrivial logical $\bar{Z}$ with $Z$-support $\mathbf{z} \in \ker{H_X}\setminus \rs{H_Z}$,
        \begin{align}\label{eq:C sector weights app}
            \abs{\mathbf{z}}_\mathscr{B}^r \geq d_1 \quad\text{or}\quad \abs{\mathbf{z}}_\mathscr{C}^c \geq d^\transp_2 \, .
        \end{align}
    \end{prop*}
    
    \begin{proof}
        From the K\"unneth decomposition \eqref{Xkunneth}, the logical $\bar{X}$ space is the direct sum of the \Bsectorname{} and \Csectorname{} logical spaces. Thus, any nontrivial logical $\bar{X}$ has a nonzero component in at least one of these two sectors.

        First suppose its \Bsectorname{} component is nonzero. From \eqref{Xkunneth}, this component is a nonzero element of
        \begin{align}
            \frac{\F_2^{n_1}}{\im H_1^\transp} \otimes \ker H_2
            =
            \frac{\F_2^{n_1}}{\rs H_1} \otimes \mathcal{C}_2 \, .
        \end{align}
        Viewing the \Bsectorname{} qubits as an $n_1\times n_2$ array, consider each column modulo $\rs H_1$. Since the \Bsectorname{} logical component is nonzero, at least one component with respect to a basis of $\F_2^{n_1}/\rs H_1$ gives a nonzero codeword of $\mathcal{C}_2$. By definition of $d_2$, this codeword is supported on at least $d_2$ columns. Adding an $X$-stabilizer changes each \Bsectorname{} column only by an element of $\rs H_1$, so these column classes are invariant under stabilizer deformation. Therefore,
        \begin{align}
            \abs{\mathbf{x}}_\mathscr{B}^c \geq d_2 \, .
        \end{align}

        Otherwise, the \Bsectorname{} component is zero and the \Csectorname{} component must be nonzero. Again from \eqref{Xkunneth}, this component is a nonzero element of
        \begin{align}
            \ker H_1^\transp \otimes \frac{\F_2^{m_2}}{\im H_2}
            =
            \mathcal{C}_1^\transp \otimes \frac{\F_2^{m_2}}{\rs H_2^\transp} \, .
        \end{align}
        Viewing the \Csectorname{} qubits as an $m_1\times m_2$ array, consider each row modulo $\rs H_2^\transp$. At least one component with respect to a basis of $\F_2^{m_2}/\rs H_2^\transp$ gives a nonzero codeword of $\mathcal{C}_1^\transp$, which has weight at least $d_1^\transp$. Adding an $X$-stabilizer changes each \Csectorname{} row only by an element of $\rs H_2^\transp$, so these row classes are invariant under stabilizer deformation. Therefore,
        \begin{align}
            \abs{\mathbf{x}}_\mathscr{C}^r \geq d_1^\transp \, .
        \end{align}
        This proves \eqref{eq:B sector weights app}.

        The argument for $\bar{Z}$ logicals is analogous. From \eqref{Zkunneth}, a nonzero \Bsectorname{} component is an element of
        \begin{align}
            \ker H_1 \otimes \frac{\F_2^{n_2}}{\im H_2^\transp}
            =
            \mathcal{C}_1 \otimes \frac{\F_2^{n_2}}{\rs H_2} \, ,
        \end{align}
        so the same argument applied to the \Bsectorname{} rows gives
        \begin{align}
            \abs{\mathbf{z}}_\mathscr{B}^r \geq d_1 \, .
        \end{align}
        If instead its \Bsectorname{} component is zero, its nonzero \Csectorname{} component lies in
        \begin{align}
            \frac{\F_2^{m_1}}{\im H_1} \otimes \ker H_2^\transp
            =
            \frac{\F_2^{m_1}}{\rs H_1^\transp} \otimes \mathcal{C}_2^\transp \, ,
        \end{align}
        and applying the same argument to the \Csectorname{} columns gives
        \begin{align}
            \abs{\mathbf{z}}_\mathscr{C}^c \geq d_2^\transp \, .
        \end{align}
        This proves \eqref{eq:C sector weights app}.
    \end{proof}

\section{Graph theory}

\subsection{Adjacency matrix}

\begin{figure}[htbp]
    \centerline{\includegraphics[scale=1]{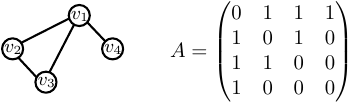}}
    \caption{Example adjacency matrix. Rows and columns are indexed by $v_1, v_2, v_3, v_4$ respectively.}
    \label{adjmtx}
\end{figure}

\begin{defn}[Adjacency matrix]
    Let $G=(V,E)$ be a simple graph with vertices $\{v_i \in V\}$ for $i=1,\dots,\abs{V}$. Its adjacency matrix $A(G)$ is a $\abs{V} \times \abs{V}$ matrix with entries
    \begin{align}
        A_{ij} = \begin{cases} 
            1 & \text{if } (v_i,v_j)\in E \\
            0 & \text{otherwise} \, .
        \end{cases}
    \end{align}
\end{defn}

\subsection{Vertex coloring}

\begin{figure}[htbp]
    \centerline{\includegraphics[scale=1]{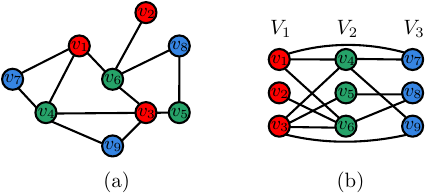}}
    \caption{Example of \textbf{(a)} 3-colorable graph; \textbf{(b)} Reorganization of vertices to show disjoint structure.}
    \label{coloring}
\end{figure}

A graph $G=(V, E)$ can be $\chi$-colored if $V$ can be partitioned into $\chi$ disjoint sets $V_1, \dots, V_\chi$ (i.e., $V=V_1 \cup \cdots \cup V_\chi$ and $V_i \cap V_j = \emptyset$ for $i, j = 1, \dots, \chi$; $i\neq j$) such that no $(v_1, v_2) \in E$ exists for some $v_1, v_2 \in V_i$, $i\in 1, \dots, \chi$ (i.e., no edg
e connects two vertices in the same color). See Fig~\ref{coloring} as an example.

\subsection{Bipartite Double-Cover}
\label{app:double cover}

\begin{figure}[htbp]
    \centerline{\includegraphics[scale=1]{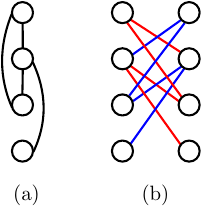}}
    \caption{ Example of \textbf{(a)} undirected graph $G$; \textbf{(b)} $G$'s bipartite double cover $G_\text{bip}$. Red edges correspond to edges $(u_A, v_B)$; blue edges correspond to edges $(u_B, v_A)$.}
    \label{doublecov}
\end{figure}

An undirected graph $G=(V, E)$'s bipartite double-cover, $G_\text{bip} = (A \sqcup B, E_\text{bip})$ is a bipartite covering of $G$, constructed as follows. Create two sets of vertices $A$ and $B$ with $|A|=|B|=|V|$. Then, for each $(u, v)\in E$, add edges $(u_A, v_B)$ and $(u_B, v_A)$ to $E_\text{bip}$, where a vertex $v_A \in A$, $v_B \in B$ corresponds to $v \in V$ in $A$ and $B$ respectively. An example construction is shown in Figure~\ref{doublecov}. If the input graph $G$ has maximum degree $\Delta$, then $G_{\rm bip}$ also has maximum degree $\Delta$. In addition, any $\ell$-cycle (cycle of length $\ell$) in $G$ induces an $\ell$-cycle in $G_{\rm bip}$ when $\ell$ is even and a $2\ell$-cycle when $\ell$ is odd. Conversely, any cycle in $G_{\rm bip}$ projects to a closed walk of the same length in $G$, which contains a cycle of length no greater. As a consequence, the minimum cycle length, or girth, of $G_{\rm bip}$ is at least that of $G$.

\section{Tanner codes}
\label{app:Tanner codes}

The Tanner construction \cite{Tanner_1981} is a procedure to build a classical linear code from a regular base graph $G$ and a local code $\hat{\mathcal{C}}$. Sipser and Spielman \cite{Sipser_1996} utilized this approach to construct the first explicit (non-probabilistic) family of classical LDPC codes with asymptotically good parameters $[n,\mathrm{\Theta}(n),\mathrm{\Theta}(n)]$.

\begin{defn}[Tanner code \cite{Tanner_1981}]
\label{defn:Tanner code}
    Let $G=(V,E)$ be a $\Delta$-regular simple graph called the \emph{base graph}, and let $\hat{\mathcal{C}}$ be a linear code of length $\Delta$ called the \emph{local code}. The Tanner code $\mathcal{T}(G,\hat{\mathcal{C}})$ is defined as
    \begin{align}
        \mathcal{T}\big(G,\hat{\mathcal{C}}\big) = \left\{ c\in\mathbb{F}^{\abs{E}}_2 : \forall v \in V \,,\, c|_{E(v)} \in \hat{\mathcal{C}} \right\} \, ,
    \end{align}
    where $E(v)$ denotes the edges incident to $v$.
\end{defn}

In other words, given a base graph where each vertex has degree $\Delta$, physical bits are placed on the edges of $G$, and parity checks of the local code $\hat{\mathcal{C}}$ are placed on the vertices. The local parity checks of $\hat{\mathcal{C}}$ enforce that all codewords of $\mathcal{T}(G,\hat{\mathcal{C}})$ are codewords of $\hat{\mathcal{C}}$ when restricted to any vertex. When the local code is a single-parity-check code with $h=\mathbf{1}_\Delta$, then the Tanner code is also known as the cycle code on $G$. Definition \ref{defn:Tanner code} requires $G$ to be a simple graph, but one can relax this definition to allow $G$ to be a Tanner graph of another code. The Tanner construction can then be viewed as ``upgrading'' each parity check of some input code to a local code, typically done as a way to boost the code distance at the cost of a lower rate.

\begin{thm}[Expander code \cite{Sipser_1996}]
\label{thm:expander code}
    Let $G$ be a simple $\Delta$-regular graph with second-largest eigenvalue $\lambda$, and let $\hat{\mathcal{C}}$ be a linear code of length $\Delta$, rate $\hat{r}$ and relative distance $\hat{\delta}$. Then the Tanner code $\mathcal{T}(G,\hat{\mathcal{C}})$ has rate at least
    \begin{align}
        r \geq 2\hat{r}-1
    \end{align}
    and relative distance at least
    \begin{align}
        \delta \geq \left(\frac{\hat{\delta}-\lambda/\Delta}{1-\lambda/\Delta} \right)^2 \, .
    \end{align}
\end{thm}

Note that the rate and relative distance are bounded away from zero when $\hat{r}>1/2$ and $\hat{\delta}>\lambda/\Delta$. For certain spectral expander graph families \cite{Lubotzky_1988, Margulis_1988} with an upper bound on $\lambda$, choosing appropriate local codes satisfying the above inequalities leads to family of asymptotically good LDPC codes by Theorem \ref{thm:expander code}.

\begin{lem}\label{lem:Tanner code max puncture}
    If the classical input code is a Tanner code $\mathcal{T}(\mathcal{G},\hat{\mathcal{C}})$ with local parity-check matrix $\hat{h}$ and parameters $[\Delta,\hat{k},\hat{d}]$, then the number of remaining neighboring $X$-checks of a \Csectorname{} qubit in the HGP code with itself obeys $\Delta'\geq \rank\hat{h} = \Delta-\hat{k}$.
\end{lem}

\begin{proof}
    Recall that for the quantum Tanner transformation, instead of a single \Csectorname{} qubit, the local transformation acts on a set of $\hat{m}=\rank\hat{h}=\Delta-\hat{k}$ local \Csectorname{} qubits and takes the form $w=\hat{g}$, where $\hat{g}$ is a generator matrix for $\hat{\mathcal{C}}$ that satisfies $\hat{g}\hat{h}^\transp=0$. Suppose, on the contrary, that we have deleted more than $\hat{k}$ neighboring $X$-checks so that $\Delta'<\hat{m}=\rank\hat{h}$. As a consequence, we now have more local \Csectorname{} qubits than remaining local $X$-checks, which implies that $\rank\hat{h}'<\rank\hat{h}$. However, $\rank\hat{h}'<\rank\hat{h}$ then implies that $\rank{H'_2}<\rank{H_2}$, in contradiction with Lemma \ref{lem:cannot puncture check}.
\end{proof}